\documentclass[11pt,twoside]{article}
\usepackage[T1]{fontenc} 
\usepackage[title]{appendix}

\usepackage{multibib}
\newcites{appendix}{Appendix References}
\usepackage{amsmath, amsthm, amssymb, mathrsfs, upgreek, enumerate, mathtools,
  comment, natbib, subfigure, graphicx, mathabx, booktabs, threeparttable,
tikz, relsize, exscale, epstopdf, multirow, dsfont}
\usepackage[headheight=110pt,margin=1.15in]{geometry}
\usepackage{caption}
\newcommand\fnote[1]{\captionsetup{font=small}\caption*{#1}}
\usepackage[section]{placeins} 
\renewcommand{\(}{\left(}
\renewcommand{\)}{\right)}
\renewcommand{\[}{\left[}
\renewcommand{\]}{\right]}
\newcommand{\bs}{\begin{align}\begin{split}\nonumber}
\usepackage{fancyhdr, setspace} 
\newcommand{\ba}{\begin{array}}

\newcommand{\ea}{\end{array}}

\newtheorem{assumption}{Assumption}
\usepackage{tcolorbox}
\usepackage{lscape}
\usepackage[figuresright]{rotating}
\usepackage[nolists,noheads,tablesfirst]{endfloat}

\usepackage{hyperref}
\usepackage{marginnote}
\hypersetup{
    colorlinks=true,
    citecolor=blue,
    filecolor=black,
    linkcolor=red,
    urlcolor=blue,
    bookmarksopen=true,
    pdfstartview=FitH
}
\newtheorem*{theorem*}{Theorem}

\newtheorem{lemma}{Lemma}[section]

\theoremstyle{definition}

\newtheorem{remark}{Remark}

\makeatletter

\def\thm@space@setup{
  \thm@preskip=15pt \thm@postskip=15pt 
}
\makeatother

\renewcommand{\qed}{\hfill \mbox{\raggedright \rule{0.08in}{0.08in}}} 
\renewenvironment{proof}[1][\proofname]{{\noindent\sc#1. }}{\qed\vspace{15pt}} 

\usepackage{xcolor}

\usepackage{authblk}
\usepackage{latexsym}

\title{\bf\sc Statistical tests for replacing human decision makers with algorithms}

\author[1]{Kai Feng}
\author[2]{Han Hong}
\author[1]{Ke Tang}
\author[3]{Jingyuan Wang}
\affil[1]{\small{Institute of Economics, School of Social Sciences, Tsinghua University, \authorcr Email: fengk22@mails.tsinghua.edu.cn, ketang@tsinghua.edu.cn}}
\affil[2]{\small{Department of Economics, Stanford University, Email: doubleh@stanford.edu}}
\affil[3]{\small{School of Economics and Management, Beihang University, Email: jywang@buaa.edu.cn}}
\begin{document}
\maketitle
\thispagestyle{empty}

\vspace{-0.3in}
\begin{abstract}
\noindent {\sc Abstract.}
This paper proposes a statistical framework of using artificial intelligence
to improve human decision making.
The performance of each human decision maker
is 
benchmarked against that of machine predictions. We
replace the diagnoses made by 
a subset of the decision
makers with the recommendation from the 
machine learning algorithm.
We apply both a heuristic frequentist 
approach and a 
Bayesian posterior loss function approach 
to abnormal birth detection using a nationwide dataset of doctor diagnoses from prepregnancy checkups of reproductive age couples and pregnancy outcomes. 
We find that
our algorithm on a test dataset results in 
a higher overall true positive rate
and a lower false positive rate
than the diagnoses made by doctors only.

\vspace{15pt}


\end{abstract}

\newpage

\section{Introduction}\label{introduction}

\setcounter{page}{1}
In the current era of machine learning,
artificial intelligence (AI) algorithms
have emerged as valuable tools
that can learn significant features
from data
and potentially facilitate
decision making
in various disciplines. 
In medical research, \cite{esteva2017dermatologist} and \cite{kermany1} utilize deep neural network to automate diagnosis based on medical images. 
%
\cite{berg2020rise} and \cite{sadhwani2021deep}
leverage machine learning algorithms to 
predict the default probability for bank lending decision. 
\cite{mullainathan2017machine} and \cite{athey2019machine} 
apply machine learning models to broad optimal policy assignment problems. 

A considerable number of papers have compared
the performance of algorithms with that of the representative human decision makers. 
To mention a few, 
\cite{esteva2017dermatologist} claims that their algorithm outperformed the aggregate 
dermatologist on some skin cancer classification tasks; 
\cite{kermany1} 
targets 
macular degeneration dagnosis 
and concluded that the machine algorithm
outperforms some retinologists; 
\cite{rajpurkar1} highlights
the performance gain of a 
deep convolutional neural network over an aggregate cardiologist for 
processing ECG sequences;
\cite{QJEbail} trains a gradient boosted decision tree model to predict crime risk during pre-trial bail 
and suggested that substituting human judges with the algorithm could result in large welfare gains.

The findings of the aforementioned literature rely mostly on the observation
that the pair of false positive rate (FPR)
and true positive rate (TPR) point
lies strictly below the receiver operating characteristic (ROC) curve
generated by the machine learning algorithm,
implying that algorithms can achieve
a higher TPR for a given FPR
or a lower FPR for a given TPR. 
In binary classification / decision making, FPR is the number of wrongly classified negative events divided by the number of
actual negative events, and TPR is the number of correctly
classified positive events divided by the number of total positive events. 
FPR/TPR are synonymous to size and power in classical hypothesis testing.

We first \emph{caution against} such interpretations
without 
a deeper
understanding of human
decision making processes.
The literature findings can be rationalized not only by the
superior information quality of algorithms but also by the \textit{incentive heterogeneity}, i.e. the 
differential preferences for taking the same action by multiple human decision makers,
some of whom can be no less accurate than algorithms in processing information
from observational data. 
Machine algorithms provide information, but decision making is a combination of incentives and information. 
Using machine algorithms to assist with decision making necessitates the modeling of incentives.
Second,
the FPR/TPR pairs of human decision makers
are typically \emph{imprecisely measured}, especially when the number of cases for each decision makers is not large.
The data 
in our empirical analysis 
shows not only substantial
heterogeneity between
individual doctors, 
but also that the randomness in
measuring the quality of decision makers plays a key role when the performance of decision
makers is compared to that of machine algorithms.

To illustrate the first issue of concern, consider
Figure \ref{figure 2}, in which a collection of FPR/TPR pairs for human decision makers, denoted as $j = 1, \ldots, J$, 
all lie approximately on a ROC curve generated by their employed decision rules $\widehat Y_{i,j} =
\mathds{1}\(p\(X_{i,j}, U_{i,j}\) > c_{j}\)$ with different individual cutoff threshold
$c_{j}$ for decision maker $j$. In the decision rule, $X_{i,j}$ are observed features used in the machine learning algorithm; $U_{i,j}$ are private information 
only observable to the human decision makers, and $p\(x,u\) =
\mathbb{P}\(Y = 1\vert x,u\)$ is a 
propensity score function.  Also drawn is the machine ROC curve based on 
classification rules $\widehat Y_i =
\mathds{1}\(p\(X_i\) > c\)$, 
which collects the corresponding FPR/TPR pairs 
by varying the cutoff threshold $c$, and 
where $p\(x\) =
\mathbb{P}\(Y = 1\vert x\)$ is the correctly specified propensity score function using observable features $x$ discovered by the machine learning algorithm.
However,
after averaging over all decision makers,
the aggregate human FPR/TPR point
lies visibly below the machine ROC curve. 
\begin{figure}
  \begin{center}
    \includegraphics[height=.40\textheight]{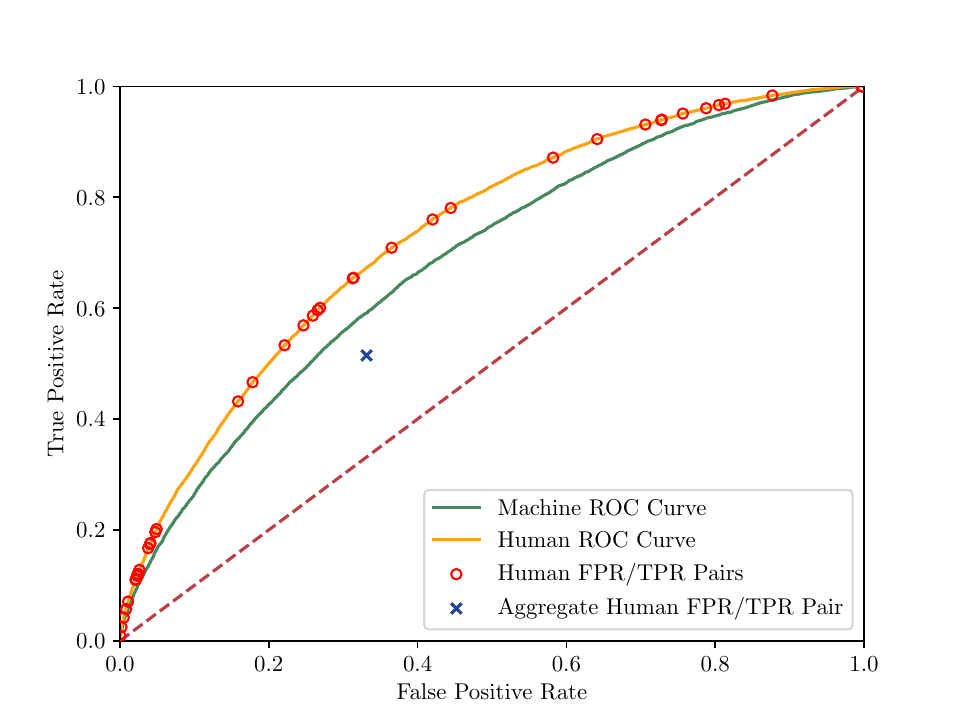}
    \caption{Individual and aggregate FPR/TPR pairs}\label{figure 2}
\fnote{\textit{Notes}:
When doctors use a correctly specified propensity score model to process both public and private information variables, 
the collection of FPR/TPR pairs for human decision makers, represented by the empty red circles, all lie
above the green machine ROC curve.  However, the aggregate FPR/TPR of human decision makers can still lie
below the machine ROC curve. See section \ref{Discussion: incentives and costs}  for a more detailed discussion of 
how incentive heterogeneity can mask the nuances of comparing between machine learning and human decision making through an implication of Jensen's inequality.
}
  \end{center}
\end{figure}


Our paper makes several 
contributions. 
First, we offer a conceptual framework of information and 
incentives for interpreting comparisons between the machine algorithms and human 
decision makers.  Second, we provide a model 
to identify
human decision maker 
candidates for  replacement by machine algorithms and  
to generate machine algorithm based decisions that incorporate human decision maker preferences. 
The replacement model accounts for the sampling uncertainty of individual decision makers. 
Third, an empirical application to 
abnormal birth detection demonstrates that our proposed replacement algorithm significantly improves the overall performance of risky pregnancy detection. 
Fourth, we conduct an extensive synthetic data analysis that serves as guidance for
evaluating empirical procedures incorporating complementarity in the presence of information asymmetry
between human decision maker and machine algorithms. 

In this paper, in order to compare machines with human decision makers, we 
restrict incentive heterogeneity and focus on information processing capacity. 
If the individual FPR/TPR pairs are precisely known without sampling errors, they can be directly compared to the machine ROC curve;
we interpret a human FPR/TPR pair below the machine ROC curve as evidence of the superiority of the machine algorithms. 
As shown in Figure \ref{SimpleDoc1},
a 
human FPR/TPR pair
 below the machine ROC curve can be dominated by any point on the line segment of the ROC curve between A and B. 
Replacing the human FPR/TPR pair by any point on the A-B segment of the machine ROC curve segment results in higher TPR without increasing  FPR, or lower FPR without sacrificing TPR.
 

In reality, only an
\textit{estimate} of the FPR/TPR pair
can be obtained from the
empirical data. 
A statistical framework to
compare the performance of algorithms with that of human decision makers
requires accounting for
estimation sampling
errors.
We focus on two main issues.
First, we seek statistical evidence supporting the information superiority of an algorithm that favors the replacement of a human decision maker.
Second, we aim to identify the most appropriate
point on the machine ROC curve for this purpose. 
To address these two issues,
we experiment with both
a heuristic frequentist confidence set approach
and a subsequent 
Bayesian inference approach.    
In both scenarios,
confidence and credible regions are formed to inform
the decision making process. The Bayesian approach results in more replacement and additional improvement in the aggregate FPR/TPR pair. We suggest the Bayesian analysis as an applicable decision framework.

At a disaggregated level, each individual human can have either less or more information processing capacity than  an algorithm has.
For instance,
\cite{liang_evaluation_2019} proposes a machine learning
algorithm to diagnose
childhood diseases that 
outperforms junior physician groups but marginally underperforms senior physician groups.
Similar findings are reported in 
\cite{peng2021deep}.
\cite{esteva2017dermatologist} and \cite{kermany1} also find substantial variability in 
the human experts' performance.
Our framework 
chooses between each human decision maker and the algorithm 
for making future decisions. We identify a subset of human decision makers to be replaced by the algorithm,
whereas the rest of human decision makers 
are retained.

We apply our statistical framework to 
a unique National
Free Prepregnancy Checkups (NFPC) dataset. 
The NFPC is a free
health checkup service for 
conceiving couples 
across 31 provinces in China. 
In addition to the pregnancy outcome and numerous patient-case features, 
the dataset also includes the doctors' IDs and diagnoses of
adverse pregnancy outcomes.
We first split the data 
into two parts.
The first part is used to compare doctors with algorithms.
Specifically, we
employ a random forest (RF) method for diagnosing risky pregnancy, which
achieves an area under the curve (AUC) above 0.68.
Using 95\% credible level,
our Bayesian approach suggests 
that the random forest algorithm
outperforms 46.1\% of doctors.
The second part of the dataset is used
to evaluate the quality of the diagnosis 
procedure based on the  combination of retained doctors and the algorithm. 
The combined decision making procedure achieves
an increase of 46.6\% in the TPR 
and a reduction of 10.1\%
in the FPR.
Additional detailed analysis further shows that 
the potential for machine algorithms to improve decision making quality is substantial even when only a smaller portion of doctors are replaced. 
We also find that doctors practicing only in country hospitals have a lower replacement ratio 
compared with those having practiced in 
township clinics.
The result of our Bayesian analysis suggests that 
51.2\% of doctors who have practiced only at  township clinics are replaced compared to 39.8\% of doctors who have practiced only at county hospitals. 

\paragraph{Related Literature}
Our paper relates to several strands of literature.
An extensive literature studies the utilization of AI to enhance medical diagnostic capabilities.
Machine learning, especially deep learning methods, 
which demonstrate great potential in tackling diverse and complex tasks, 
has been experimented to diagnose in multiple medical subfields, 
including ophthalmology \citep{kermany1}; 
cardiology \citep{rajpurkar1, hannun2019cardiologist}; 
dermatology \citep{esteva2017dermatologist}; 
respirology \citep{ardila2019end, liang_evaluation_2019}; 
clinical psychology and psychiatry \citep{ashar2017empathic, yoon2022d}; 
tumor and cancer detection  \citep{mckinney2020international, peng2021deep, uhm2021deep, tian2024prediction}. 
Additional surveys can be found in \cite{erickson2017machine}, \cite{johnson2018artificial}, \cite{dwyer2018machine}, 
\cite{chan2020machine}, \cite{huang2020fusion} and \cite{swanson2023patterns}.

Beyond 
medical diagnoses,
AI is widely explored in social-economic decision making. 
See for example \cite{berg2020rise}, \cite{sadhwani2021deep} and \cite{fekadu2022machine} for loan approval and credit risk modeling; 
\cite{fuster2019role}, \cite{vallee2019marketplace} and \cite{fuster2022predictably} 
for market structure implications of employing machine learning algorithms; 
\cite{berk2017impact}, \cite{zeng2017interpretable}, \cite{QJEbail} for bail and parole analysis; 
\cite{chalfin2016productivity} and \cite{sajjadiani2019using} for productivity 
assessment. 

Numerous papers in the literature have benchmark algorithms against 
human decision makers\footnote{
See for example
\cite{long2017artificial}, \cite{kermany1}, \cite{de2018clinically}, 
\cite{rajpurkar1}, \cite{hannun2019cardiologist}, \cite{esteva2017dermatologist}, \cite{han2020augmented}, 
\cite{daneshjou2022disparities}, \cite{ardila2019end}, \cite{liang_evaluation_2019}, \cite{mckinney2020international}, 
\cite{peng2021deep}, \cite{uhm2021deep}, \cite{tian2024prediction}, \cite{berk2017impact}, \cite{QJEbail}, \cite{han2020augmented}.
}. 
%
While machine algorithms often outperform some aggregation 
of human decision makers, these findings are not uniform across
individual human decision makers. 
Unlike machine algorithms, human decision makers can be highly heterogeneous 
in both incentives and information processing capacity. 
The health economics literature suggests 
that doctors may be swayed by
multiple non-medical incentives including
lawsuit avoidance \citep{studdert2005defensive}, financial gain \citep{johnson2016physicians}, 
demand pressure from patient \citep{lopez2018contribution} and procedural skills \citep{currie3}. 
Our analysis highlights the importance of 
modeling and restricting incentive heterogeneity in a human-AI comparison.  

A growing literature also explores AI assisted human decision making.
While \cite{brennan2019comparing}, \cite{han2020augmented}, \cite{peng2021deep}, 
\cite{yang2021deep} and \cite{tian2024prediction} report salient performance improvement
when doctors can refer to
the prediction results of algorithms,
\cite{jin2024evaluating} and \cite{yu2024heterogeneity} 
report contradicting evidence where the performance of a portion of doctors is worsened by AI assistance.  
By offering doctors monetary incentives prior to the experiment,
\cite{wang2023diagnosing} and 
\cite{wang2024overcoming} find that 
doctor diagnoses can be significantly influenced by both the machine information 
and the preset incentive structures. 
While both  \cite{bansal2021does} and \cite{wang2023diagnosing}  find that
giving doctors explanations increases the acceptance of machine recommendations, the effect of
providing explanation on performance is inconclusive and varies.
In the beverage vending machine business, a field experiment on product assortment decision 
conducted by \cite{kawaguchi2021will} 
finds heterogeneous adoption by workers dependent on location and both worker and
machine characteristics. 
\cite{alur2023auditing} propose a test whether human experts help improve machine 
algorithm predictions. 
Our analysis also suggests that the 
human machine complementarity, argued for by 
\cite{donahue2022human}, is not only theoretically challenging  but also difficult to implement empirically. 
Overall, 
the mechanism by which AI assists human decision making is unclear. 
Our paper thus focuses on a simple strategy that identifies and replaces underperforming 
doctors. 
This modeling choice is also consistent with the
analysis in \cite{mullainathan2022diagnosing} and \cite{chan2022selection}
who suggest that doctor performance can not be
explained by preference heterogeneity alone.
In addition,
the case-specific replacement algorithm described in 
Appendix \ref{fine-grained replacement} 
does not outperform the baseline method for our NFPC dataset.  


A substantial body of literature discusses the discriminatory bias of algorithms 
\citep{o2017weapons, eubanks2018automating, berk2021fairness}. 
But the potential of AI and algorithmic decision making 
to reverse extant bias in the training data is also noted by \cite{rambachan2019bias} 
and \cite{rambachan2020economic}.
Through efficent information provision and scarce resource utilization,
AI can be instrumental for inequality reduction. Using
artificial intelligence to empower medical diagnose in developing areas 
has been advocated by researches including 
\cite{adepoju2017mhealth, guo2018application, trivedi2019risks, mondal2023artificial} and others. 
\cite{bulathwela2021could} discuss the comprehensive application of AI 
to narrow education inequality gap.
Our finding using the NFPC dataset is also consistent with anecdotal information regarding
the quality disparity across tiers of the hospital classification system in China.

Optimal decision making using machine learning techniques is also drawing increasing attention from econometrics. 
\cite{kitagawa2018should}, \cite{mbakop2021model}, \cite{athey2021policy} and \cite{feng2024statistical} 
explore the asymptotic properties of empirical algorithm based optimal decision making.
\cite{manski2018credible} considers the identification problem 
when only the prediction conditional on a subset of covariates and the conditional distribution of the ``omitted'' covariates are known.
In this paper, we assume, as in the baseline cases of 
\cite{kitagawa2018should}, that the machine algorithm is known and deterministic,
in view of the large data size 
and the comprehensive list of covariates 
in the NFPC dataset. 

The rest of this paper is organized as follows.
Section \ref{theorysection} presents the statistical model
of human-algorithm comparison.
Section \ref{data and algorithm} describes 
the data and the algorithm. 
Section \ref{classify doctors} reports results from 
the empirical analysis of
    the machine algorithm and human decision makers. Section \ref{synthetic data analysis}
 presents a set of synthetic data analysis for additional insights,
and Section \ref{conclusion} concludes. Additional results are collected in the Appendix.

\section{Replacing Human Decision Makers with Algorithms}\label{theorysection}
Consider a sample dataset with labels $Y_i \in \{0,1\}$
and 
features $X_i,i=1,\ldots,n$.
Define 
\begin{align}\begin{split}\label{sample tpr fpr}
\textup{TPR} = \frac{\frac1n\sum_{i=1}^n Y_i \widehat{Y}_i}{\hat p},\quad
\textup{FPR} = \frac{\frac1n\sum_{i=1}^n (1-Y_i) \widehat{Y}_i}{1 - \hat p},
\quad \hat p = \frac1n \sum_{i=1}^n Y_i,
\end{split}\end{align}
where
$\hat Y_i \in \{0, 1\}$
is a predictor for $Y_i$ based on a general decision rule, such that 
$\(X_i, Y_i, \hat Y_i\)$ are i.i.d draws from an underlying population.
The Law of Large Number implies that
the FPR/TPR converge to their
population analogs
(denoted as $\beta$ and $\alpha$) as $n\rightarrow \infty$:
\begin{align}\begin{split}\label{alpha_beta}
\beta = \frac{1}{p}\mathbb{E} Y_i \widehat{Y}_i, \quad
\alpha = \frac{1}{1-p} \mathbb{E}\[\(1-Y_i\)\widehat{Y}_i\], \quad
\text{where}\quad p = \mathbb{E}Y_i = \mathbb{P}\(Y_i=1\).
\end{split}\end{align}

A key concept in this paper is
\textit{incentive heterogeneity}.
In general terms,
incentive refers to the motivation of certain actions,
and incentive heterogeneity refers to different payoffs across multiple human decision makers for taking
the same actions.
Incentive heterogeneity can typically be
classified intuitively as either intra-individual or inter-individual.
On the one hand, inter-individual incentive heterogeneity refers to situations where
the cost assessment only varies across decision makers but not across individual cases for each decision maker.
On the other hand, intra-individual incentive heterogeneity
refers to situations 
where the cost assessments across individual cases for a given decision maker vary  
based
on both publicly observed features $x$ that are accessible by the machine algorithm and private information features $u$ that are only available to human decision makers. 
To compare human decision makers and machines, we make the following two assumptions regarding incentive heterogeneity:

\begin{assumption} \label{basic assump}
The machine ROC curve is generated by 
a propensity score model of prediction:
\begin{align}\begin{split}\nonumber
\hat{y}_{i, M} = \mathds{1}\(m\(x_i\) > c\).
\end{split}\end{align}
The decision maker $j$'s
decision is also based on a perceived propensity score model:
\begin{align}\begin{split}\nonumber
\hat{y}_{i,j} = \mathds{1}\(q_{j}\(x_{i,j}, u_{i,j}\) > c_j\).
\end{split}\end{align}
Both $m\(\cdot\)$ and $q_{j}\(\cdot\)$ can be correctly specified or misspecified.
\end{assumption}

Assumption \ref{basic assump} rules out
intra-individual incentive heterogeneity
and allows us to focus 
on inter-individual incentive heterogeneity.
Doctors in our dataset exhibit highly variable degrees of conservativeness in their diagnoses. 
The FPR of the doctors with at least 300 diagnoses has a mean of 0.236 and a standard deviation of 0.230. 
Some doctors tend to diagnose almost all patients as highly risky, and some doctors 
rarely make a positive diagnosis. 
Figure \ref{frequentist replaced doctor scatter plot-300} illustrates. 
The diffusive distribution of doctors' FPR/TPR pairs also indicates both information processing capacity difference and inter-individual heterogeneity. 
The analysis in \cite{chan2022selection} shows that when only either diagnostic skill difference or inter-individual incentive heterogeneity is present,
the doctors FPR/TPR paris will concentrate in a small region, which is not the case in Figure \ref{frequentist replaced doctor scatter plot-300}. 
%
Importantly, Assumption \ref{basic assump} also allows for private information $u_{i,j}$ that is only observed by the doctors and not used in the machine learning algorithm. 
In the absence of the private information variable $u_{i,j}$, the doctors' diagnoses are perfectly predictable by the machine learning algorithm, which is very unlikely.

\begin{assumption}\label{ROC assump}
The machine 
propensity score model $m\(x_i\)$ and the $j$th decision maker's propensity score model $q_j\(x_{i,j},u_{i,j}\)$ are used without sampling errors in the machine decision making procedure and by the $j$th decision maker, respectively.
\end{assumption}

The first requirement 
is justified by the fact that the machine learning model is implemented using a training dataset that is by orders of magnitude larger than the number of cases for each individual human decision maker. 
For instance, around 150 thousand cases in the NFPC training dataset are used
to estimate the machine ROC curve, 
whereas the median number of diagnoses per doctor
in the classification set is only approximately 400. 
The estimation error in the machine ROC curve is 
negligible relative to the error in estimating doctors' FPR/TPR pairs. 
The sampling error in estimating the machine ROC curve is studied in a subsequent paper by \cite{feng2024statistical}. 

The second requirement that 
the 
$q_j\(x_{i,j},u_{i,j}\)$ are used without sampling errors by the  
$j$th decision maker is 
interpreted as stating that the behavior of the $j$th decision maker is consistent 
with the use of a propensity score model $q_j\(x_{i,j},u_{i,j}\)$. The method in this paper does not require the researcher 
to know the propensity score model $q_j\(x_{i,j},u_{i,j}\)$ for the $j$th decision maker, and does not refer to any of the doctor propensity score models $q_j\(x_{i,j},u_{i,j}\), j = 1,\ldots,J$. 
The second requirement also rules out 
decision makers 
learning from the sample cases. A doctor $j$ does not update their propensity score model   $q_j\(x_{i,j},u_{i,j}\)$ during the data sample period.

This is a reasonable assumption in our empirical dataset because the eventual pregnancy outcomes are typically not available to doctors when they diagnose consecutive patient cases. 
We are not concerned with how doctors procure knowledge to develop their diagnosis model $q_j\(x_{i,j},u_{i,j}\)$. 
Doctor $j$ might procure
$q_j\(x_{i,j},u_{i,j}\)$ from trainings obtained from a medical school or from working with a large pool of patients prior to 
the sampling period.  
Section \ref{data and algorithm} provides more details about the data processing procedure, where
sampling errors arise mainly from the relatively fewer number of observations for each individual decision maker. 

\subsection{Human FPR/TPR pairs and ROC curves in the population}
\label{compare pair with curve}

The parameter of interest is 
the population pair of true positive and false positive rates for an individual
decision maker:
$\theta_{H} = \(\alpha_{H}, \beta_{H}\)$,
where $\alpha$ and $\beta$ are
defined as in \eqref{alpha_beta} and the
subscript $H$ refers to ``human''.
If 
$\theta_{H}$ is above a machine ROC curve, as in Figure \ref{figure 2},
then given the FPR
$\alpha_{H}$,
the algorithm has a lower
TPR $\beta_{M}$ than humans have. Similarly, given the TPR $\beta_{H}$,
the machine has a larger FPR $\alpha_{M}$
than humans have.
In this sense,
humans outperform the algorithm. In particular, under Assumption \ref{basic assump}, 
if humans correctly use at least
as much information as the algorithm does, in the sense that the propensity model 
is correctly specified, then 
$\theta_{H}$ 
lies above the machine ROC. The formal argument is given in Lemma \ref{better information higher roc} in the 
appendix.
%

Figure \ref{SimpleDoc1} shows a $\theta_{H}$
pair for a human decision maker that is below the machine ROC curve. 
A machine decision rule
corresponding to a point on the machine ROC curve between A and B
performs better than the human
decision maker.
The segment of the curve to the left of B
(to the right of A)
has a smaller $\alpha_{M}$
but a smaller $\beta_{M}$ (a larger $\beta_{M}$
but a larger $\alpha_{M}$)
than the human decision maker.
These segments of the machine ROC curve are not
directly comparable to the human decision maker.

\begin{figure}
	\centering
	\includegraphics[width=0.8\columnwidth]{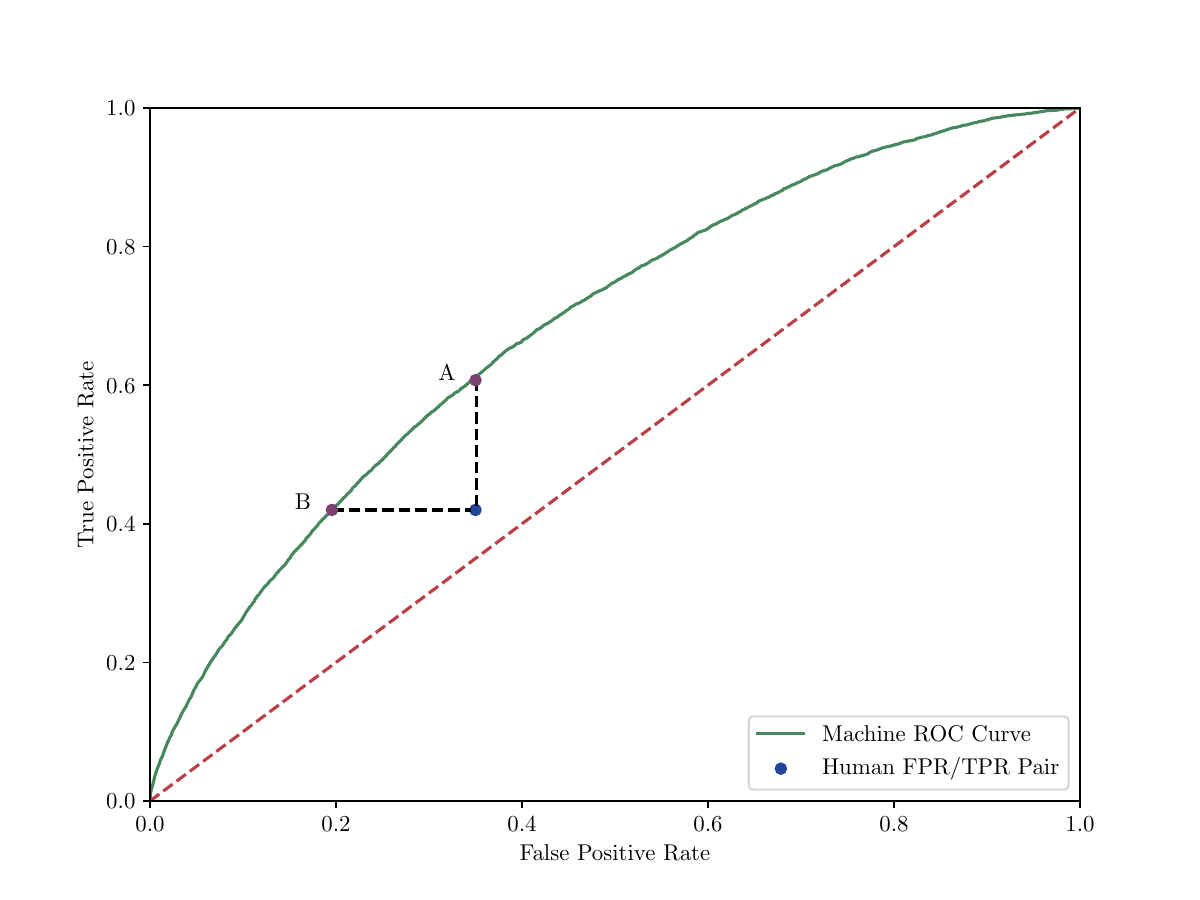}
	\caption{Population human FPR/TPR pair and the machine ROC curve}
\fnote{\textit{Notes}:
Point A matches the human FPR 
but has a higher TPR.
Similarly,
Point B matches the human TPR but has a lower
FPR.
The FPR and TPR
of any point between points A and B on the machine ROC curve are smaller 
and larger, respectively, than the human decision maker's FPR/TPR pair.
}
	\label{SimpleDoc1}
\end{figure}
The classical Neyman-Pearson lemma states that when $p\(x\)=\mathbb{P}\(Y=1\vert x\)$ is correctly specified, for any pair of positive numbers $\(\eta, \phi\)$, there exists a cutoff value $c$, 
such that the point $\(\alpha_c, \beta_c\)$ on the ROC curve corresponding to the decision rule 
$\hat y = \mathds{1}\(p\(x\)>c\)$ maximizes a linear combination $\phi\beta - \eta\alpha$ of the TPR and FPR. 
Conversely, each point on this ROC curve optimizes some linear combination of TPR and FPR. Neyman-Pearson Lemma is also motivated by a Bayesian decision maker 
who minimizes posterior expected loss of misclassification upon observing the features $x$, based on a loss matrix of the following form:
        \begin{center}
        \begin{tabular}{ l c r }
                Loss matrix & $\hat y=0$ & $\hat y =1$ \\
                $y=0$ & $0$ & $c_{01}$ \\
                $y=1$ & $c_{10}$  & $0$,
        \end{tabular}
                 \end{center}
where $c_{10}$ represents the cost of misclassifying $1$ as $0$ and $c_{01}$ represents the cost of misclassifying $0$ as $1$. Each point on the ROC curve, corresponding to a decision rule $\hat y = \mathds{1}\(p\(x\) > c\)$, minimizes Bayesian posterior expected loss for some choice of $c_{10}$ and $c_{01}$ through the relation $c = \frac{c_{01}}{c_{10} + c_{01}}$. Importantly, the cost matrix $c_{01}, c_{10}$ and linear combination coefficients $\eta, \phi$ in the Neyman-Pearson Lemma are independent of the features $x$ in order to provide optimality justification of the points on the ROC curve. Assumption \ref{basic assump}
is then interpreted as
each doctor behaving as a Bayesian decision maker adopting the same constant cost matrix $c_{10}$ and $c_{01}$ among their patients. The ex ante expected minimized cost for the Bayesian decision maker corresponds to a linear combination of FPR and TPR:
\bs
c_{10} \mathbb{E}\[p\(X\)\(1-\hat Y\)\] + c_{01} \mathbb{E}\[\(1-p\(X\)\) \hat Y\] = \(1-p\) c_{01} \text{FPR} - p c_{10} \text{TPR} + p c_{10}.
\end{split}\end{align}



If we denote point A and B in Figure \ref{SimpleDoc1}
as $\(\alpha_{H}, \beta_{R}\)$ and
$\(\alpha_{R}, \beta_{H}\)$, 
then any point $\(\alpha_{M}, \beta_{M}\)$
between $\(\alpha_{R}, \beta_{H}\)$
and $\(\alpha_{H}, \beta_{R}\)$
on the machine ROC curve satisfies
$\alpha_{M} < \alpha_{H}$ and
$\beta_{M} > \beta_{H}$.
For arbitrary $\phi, \eta > 0$ and $c_{10}, c_{01}$,
\begin{align}\begin{split}\nonumber
&\eta \alpha_{M} - \phi \beta_{M} <
\eta \alpha_{H} - \phi \beta_{H},\quad\text{and}\quad\\
&\(1-p\) c_{01} \alpha_M - p c_{10} \beta_M <
\(1-p\) c_{01} \alpha_H - p c_{10} \beta_H.
\end{split}\end{align}
Therefore, compared with human decisions,
any point $\(\alpha_{M}, \beta_{M}\)$ between $\(\alpha_{R}, \beta_{H}\)$
and $\(\alpha_{H}, \beta_{R}\)$
on the ROC curve achieves a better linear combination
for all possible values of $\phi, \eta$, and lower expected cost for all loss matrixes of $c_{01}, c_{10}$. 
In this sense, 
$\(\alpha_{M}, \beta_{M}\)$ 
\textit{dominates} $\(\alpha_{H}, \beta_{H}\)$.


Recall that $\theta_{H}$
is a population parameter that measures the
expected performance of the human decision maker.
When $\theta_H$ lies below the machine ROC curve, {\it on average},
the machine algorithm can outperform
the human decision maker by
increasing the population TPR or
reducing the population FPR. 
However, unless the algorithm can make
predictions with perfect accuracy,
for each specific patient case,
the human decision maker can
either underperform or outperform the algorithm.

Our framework is likely to be applicable when people make many repeated decisions to allow for a sufficiently precise estimate of their error rates. It
is less applicable in situations when people only ever make a few decisions of a type, a setting that may arise even when decisions are frequent if some action rarely is or should be taken. 
The lack of information renders the task of distinguishing humans from algorithms difficult, even if humans might be detectably worse in aggregate\footnote{We are grateful to the associate editor for pointing out this important clarification.}. Our application features a comprehensive set of covariates, a large dataset, reliable and measurable outcomes, and a stable prediction environment that are 
consistent with the scenarios in the 
important papers by \cite{kahneman2009conditions} and \cite{currie3}.

\subsection{Human FPR/TPR pair and ROC curve comparison in the sample}
\label{heuristic approach}
Typically, the population value of  $\theta_{H}$ is not directly observed, and instead
needs to be estimated from a dataset. 
Denote the the sample estimates $\(\text{FPR}, \text{TPR}\)$ in \eqref{sample tpr fpr}
as $\hat{\theta}_{H} =
\(\hat{\alpha}_{H}, \hat{\beta}_{H}\)$.
In the presence of sampling uncertainty,
the inference problem pertains to how we can make a probabilistic statement regarding
whether $\theta_{H}$
is above or below the machine ROC curve. From a frequentist analysis point of view,
because the sample estimate $\hat{\theta}_{H}$ is a vector function of a multinomial distribution, 
the finite sample distribution
function of $\hat{\theta}_{H}$ can be analytically intractable.
Large sample frequentist analysis and simulation methods, however,
are facilitated by the joint asymptotic normal distribution
of $\hat{\theta}_{H}$.


\begin{lemma}\label{distributionFPRTPR}
For an i.i.d. sample $\left\{Y_{i}, \widehat{Y}_{i}\right\}^{n}_{i}$, 
the joint asymptotic distribution of
a human FPR/TPR pair
$\hat{\theta}_{H} =
\(\hat{\alpha}_{H}, \hat{\beta}_{H}\)$
is multivariate normal.
In particular
\begin{align}\begin{split}\nonumber
\sqrt{n}\(\hat{\theta}_{H} - \theta_{H}\)
\overset{d}{\longrightarrow} N\(0, \Sigma\), \quad
\Sigma =
\(
\begin{array}{cc}
\frac{\alpha_{H}\(1 - \alpha_{H}\)}{1 - p} & 0\\
0 & \frac{\beta_{H}\(1 - \beta_{H}\)}{p}
\end{array}\).
\end{split}\end{align}
\end{lemma}

The proof of Lemma \ref{distributionFPRTPR} is given in the appendix and follows from standard arguments for approximating multinomial distributions with normal limits.
The asymptotic covariance $\Sigma$ can be consistently estimated by a sample analog $\hat\Sigma$ where 
the parameters $\alpha_H, \beta_H$ and $p$ are replaced 
by $\hat{\alpha}_H, \hat{\beta}_H$ and $\hat{p}$. 
The resulting asymptotic confidence set
of the FPR/TPR pairs 
is elliptically shaped.
As illustrated in Figure 
\ref{fg:case1}, the blue ellipse indicates a 95\% level confidence set for a decision maker
in our dataset. 

\begin{figure}
  \centering 
  \subfigure[Case 1 
  ]
  {\includegraphics[width=0.48\columnwidth]{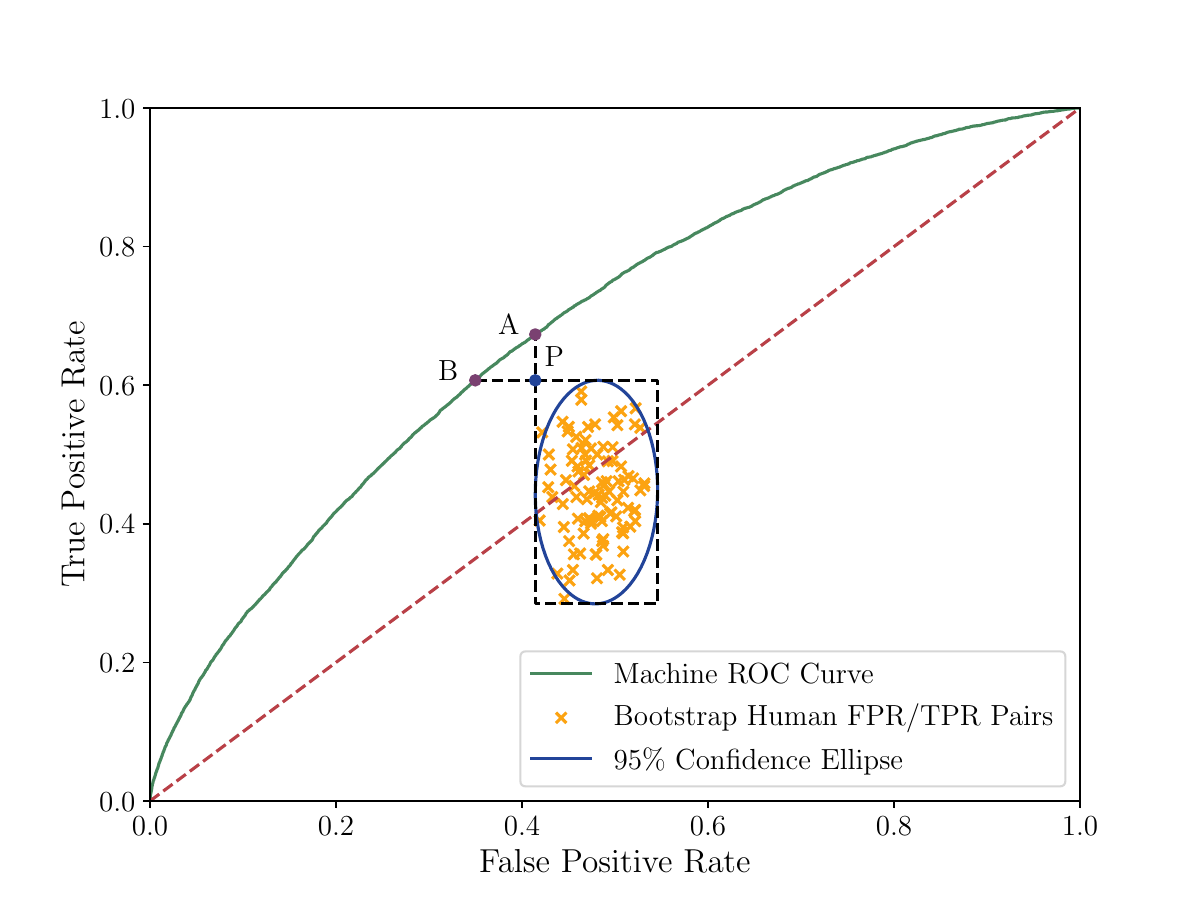}
  \label{fg:case1}}
  \subfigure[Case 2 
  ]{\includegraphics[width=0.48\columnwidth]{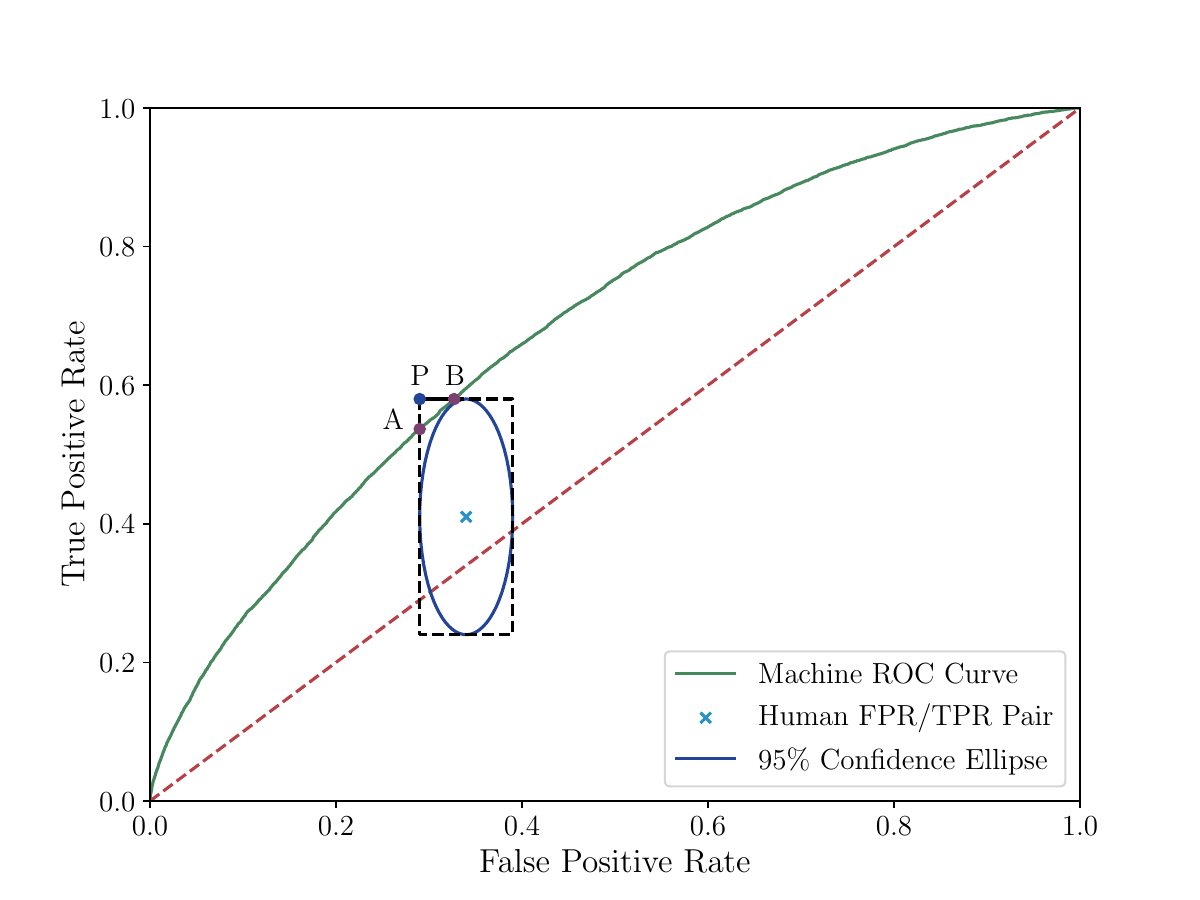}
	\label{fg:case2}}
  \subfigure[Case 3 
  ]{\includegraphics[width=0.48\columnwidth
  ]{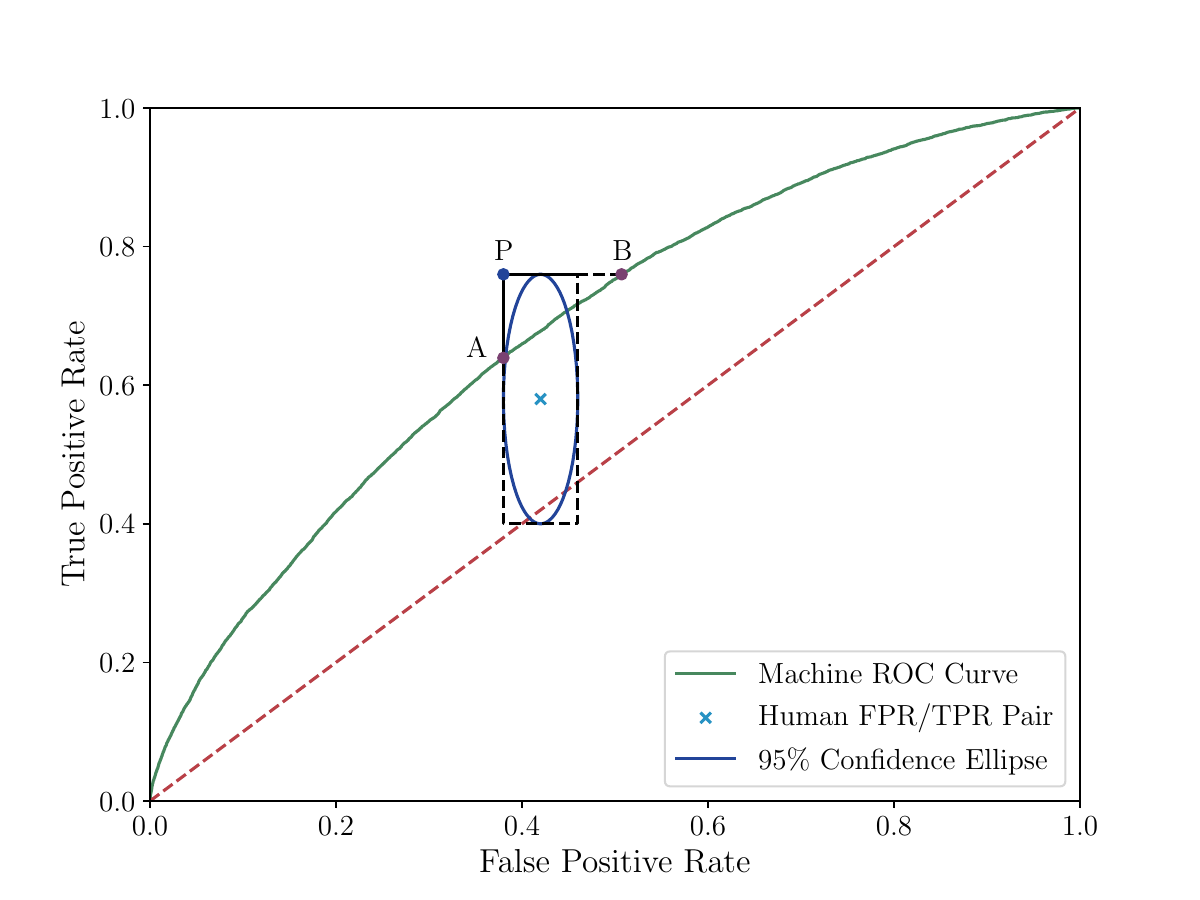}
  \label{fg:case3}}
  \caption{Three cases of the heuristic approach}
  \fnote{\textit{Notes}: 
  The heuristic frequentist approach makes use of the point $\text{P} \coloneqq \(\alpha_{low}, \beta_{high}\)$ that corresponds to 
the smallest FPR and the highest TPR of the elliptical confidence set. 
  Panel \subref{fg:case1} shows the case where P lies below the machine ROC curve and where there exist dominating points on the curve. The yellow points 
are the bootstrapped FPR/TPR pairs. 
  Panel \subref{fg:case2} shows the case where the confidence set lies below the machine ROC curve, but there is no dominating point on the curve.
  Panel \subref{fg:case3} shows the case where a portion of the confidence set lies above the machine ROC curve. 
  }
  \label{heuristic frequentist comparison}
\end{figure}

An immediate consequence of Lemma \ref{distributionFPRTPR} is that a direct comparison of the sample estimate $\(\hat\alpha_H,\hat\beta_H\)$ with the machine ROC curve is likely to be
insufficient. Consider a situation where the true population pair of $\(\alpha_H, \beta_H\)$ lies strictly below the machine ROC curve: $\beta_H < g\(\alpha_H\)$ where $g\(\cdot\)$ denotes the machine ROC curve. 
Then with probability converging to $1$, $\(\hat\alpha_H, \hat\beta_H\)$ lies strictly below the machine ROC. 
Suppose we replace $\(\hat\alpha_H, \hat\beta_H\)$ by 
$\(\hat\alpha_H, g\(\hat\alpha_H\)\)$.  
By Lemma \ref{distributionFPRTPR}, $\mathbb{P}\(\hat\alpha_H \geq \alpha_H\)\rightarrow 0.5$. In other words, the replacement point $\(\hat\alpha_H, g\(\hat\alpha_H\)\)$ on the ROC curve 
does not dominate the true population pair of $\(\alpha_H, \beta_H\)$ approximately half of the time. 
A similar situation arises if we replace $\(\hat\alpha_H, \hat\beta_H\)$ by $\(g^{-1}\(\hat\beta_H\), \hat{\beta}_{H}\)$. 
 In order to provide a probabilistic statement regarding the property of the replacement procedure, 
a more formal statistical framework is necessary. 
We consider both a heuristic frequentist approach and a Bayesian approach.

\subsection{A heuristic frequentist approach}\label{heuristic frequentist approach}

We initially experiment
with a heuristic procedure based on the frequentist principle, the predominant school of thought in statistical inference which typically begins with confidence set construction. 
The procedure is described as follows.
\begin{itemize}
\item Form a confidence set
$\hat S$ for $\theta_{H}$.
By Lemma \ref{distributionFPRTPR}, 
$n\(\hat{\theta}_{H} - \theta_{H}\)^{T}\hat{\Sigma}^{-1}\(\hat{\theta}_{H} - \theta_{H}\)$ 
converges to a Chi-square distribution with two degrees of freedom,
where $\hat{\Sigma}$ is any consistent estimator of $\Sigma$. 
A  confidence set on this asymptotic distribution 
is elliptically shaped.
\item For each $s=\(\alpha_s, \beta_s\) \in \hat S$, define $A_s$ as the set of points
that {\it dominate} $s$:
\begin{align}\begin{split}\label{dominate points}
A_s = \left\{\theta = \(\alpha, \beta\):
\alpha \leq \alpha_s, \beta \geq \beta_s\right\}.
\end{split}\end{align}
\item Define
$A = \bigcap_{s \in \hat{S}} A_s$: 
$A$ is the set of points that
dominate all the points in $\hat S$.
\item Next, define
$\bar A = A \cap \text{ROC}$:
$\bar A$ is the set of points on the machine ROC curve that simultaneously dominate all the
points in $\hat S$.
\end{itemize}

The convention of choosing a large confidence level 
suggests that the status quo of human decision making is to be replaced by an algorithm only when there is overwhelming evidence from the data indicating 
the superiority of the algorithm. In our application, we will only consider
the machine algorithm 
to outperform the human decision when 
$\bar A \neq \emptyset$, i.e., when there exists a point on the machine ROC curve  that dominates an entire confidence set of 
the human decision maker's $\theta_{H}$ parameter.
In essence, 
we use estimates of the sampling uncertainty of
the FPR/TPR pair of $\hat\theta_H$ 
to provide guidance on how far $\hat\theta_H$ needs to be below the machine ROC curve in order to justify replacing the human decision maker with the machine algorithm.

In the elliptical confidence set depicted in Figure \ref{fg:case1},
denote by $\(\alpha_{\beta_{high}}, \beta_{high}\)$
the point that achieves the highest TPR on the confidence set,
and denote by $\(\alpha_{low}, \beta_{\alpha_{low}}\)$
the point that achieves the smallest FPR on the confidence set. 
The point $\text{P}$ in Figure \ref{fg:case1} corresponds to $\(\alpha_{low}, \beta_{high}\)$. 
There are three possibilities for 
the relation between the position of the confidence set
for the human FPR/TPR pair $\theta_H$ and the machine ROC curve. 

\begin{enumerate}
\item Case 1 (Figure
\ref{fg:case1}): The human confidence set and point
$P=\(\alpha_{low}, \beta_{high}\)$
are both below the machine ROC curve.
We consider the human decision maker to perform
``worse'' than the algorithm
and hence can be replaced by the algorithm.
On the machine ROC curve,  we can find two
points, namely, $\(\alpha_{R}, \beta_{high}\)$
and $\(\alpha_{low}, \beta_{R}\)$,
labeled  as points B and A in Figure \ref{fg:case1},
such that any point on the ROC curve between A and B can provide a machine decision rule to 
replace human decision making.

\item Case 2 (Figure \ref{fg:case2}): The entire human
confidence set is below the ROC curve,
but the point $P=\(\alpha_{low}, \beta_{high}\)$
is above the ROC curve. In this case, even though each point in the confidence region of $\theta_{H}$ is dominated by some
point on the machine ROC curve, 
it is not possible to find a nonempty fraction of the machine ROC curve that simultaneously dominates
all the points in the human
confidence set.
The data evidence is not sufficiently convincing to account for the
randomness of estimating $\theta_H$ by $\hat{\theta}_{H}$. Consequently,
the human decision maker is not replaced by 
the algorithm.

\item Case 3 (Figure \ref{fg:case3}): The human confidence set
has a certain area above
the machine ROC curve. 
The human decision maker
is not replaced by the algorithm either.
\end{enumerate}

In summary, if the  point $P=\(\alpha_{low}, \beta_{high}\)$
is below the ROC curve, the human decision maker is replaced; otherwise, they are not.
In cases 2 and 3, it is possible that the human decision maker
is sufficiently capable compared to the machine algorithm. 
It is also possible that the human 
$\theta_{H}$ is not estimated precisely enough,
for example, due to a lack of historical data that results in 
a large confidence set.
In frequentist statistical inference, $\theta_{H}$
is a fixed number. 
It is either above or below the machine ROC curve.
The confidence set itself $\hat S$ is a random set, such that
$\mathbb{P}\(\theta_H \in \hat S\)
\longrightarrow 1 - \alpha.$
For a conventional size of $\alpha = 0.05$,
if the confidence set is to be constructed 100 times
from different samples, $\hat S$ will
contain $\theta_{H}$ approximately 95 times.
Consequently, by construction, 
$\bar A$ is a random set
such that asymptotically,
\begin{align}\begin{split}\nonumber
\mathbb{P}\(\theta_H \in \hat S\)=
\mathbb{P}\left\{\bar{A} = \emptyset
\text{ and P dominates } \theta_{H}
\right\} +
\mathbb{P}\left\{\bar{A} \neq \emptyset
\text{ and $\bar{A}$ dominates $\theta_{H}$}\right\}
\geq 1 - \alpha.
\end{split}\end{align}
Therefore, we have asymptotically
\begin{align}\begin{split}\nonumber
\mathbb{P}\left\{\bar{A} \text{ dominates } \theta_{H}
\vert \bar{A} \neq \emptyset\right\} \geq
\frac{1 - \alpha - \mathbb{P}\left\{
\bar{A} = \emptyset
\text{ and P dominates } \theta_{H}\right\}}
{\mathbb{P}\left\{\bar{A} \neq \emptyset\right\}}.
\end{split}\end{align}
If $\theta_{H}$ lies below the ROC curve,
$\mathbb{P}\left\{\bar{A} = \emptyset\right\} \to 0$
as the number of cases for the human decision maker
increases to infinity. Furthermore,
$\mathbb{P}\left\{
\bar{A} = \emptyset
\text{ and P dominates } \theta_{H}\right\} <
\mathbb{P}\left\{\bar{A} = \emptyset\right\} \rightarrow 0$, implying
that asymptotically,
$\mathbb{P}\left\{\bar{A} \text{ dominates } \theta_{H}
\vert \bar{A} \neq \emptyset\right\} \geq 1 - \alpha.$
If $\theta_H$ lies above the ROC, we can
control  $\mathbb{P}\left\{\bar A = \emptyset\right\} 
\geq 1-\alpha$ by
reasoning that
\begin{align}\begin{split}\nonumber
\mathbb{P}\left\{\bar A = \emptyset\right\}  =
\mathbb{P}\left\{\text{P is above ROC}\right\} \geq
\mathbb{P}\left\{\text{P dominates $\theta_H$}\right\}
\geq 
\mathbb{P}\(\theta_H \in \hat S\)
\geq 1-\alpha.
\end{split}\end{align}
However, the finite sample size of the test can be severely distorted. By insisting on finding a segment of the machine ROC curve that 
simultaneously dominates the entire confidence region of $\theta_H$, the frequentist procedure is conservative by construction.
While conventional confidence sets are mostly symmetric, it is well known in statistics that by inverting a hypothesis test, alternative confidence sets can be
constructed as the collection of parameter values for which the null hypothesis is not rejected. For example, if we let the null hypothesis be $H_0: \theta_H = \theta_{0H}$ and the alternative 
hypothesis be $H_1: \alpha_H > \alpha_{0H}\ \text{or}\ \beta_H < \beta_{0H}$, a test can be formed to reject the null hypothesis $H_0$ when either $\hat\alpha_H > \alpha_H + c$ or when 
$\hat\beta_H < \beta_H - d$, where $c$ and $d$ are chosen by the size of the test. The confidence set of $\theta_{0H}$ where the test does not reject 
takes the form of
$\{\theta_{0H}: \alpha_{0H} \geq \hat\alpha_H -c,\ \text{and}\ \beta_{0H} \leq \hat\beta_H +d\}$, which is lower rectangular. 
While forming confidence sets based on inverting inequality tests is theoretically intriguing, it is not often used empirically.

An alternative in Appendix \ref{macine decision thresholds} is to test whether the population value of $\theta_H$ is above or below the machine ROC curve. 
We do not use this test because of its limitation that 
even if the null hypothesis of the status quo of human decision making is rejected, it does not 
locate a point of the machine ROC curve to replace human decisions, as discussed in section \ref{heuristic approach}.

\subsection{A Bayesian approach}\label{the Bayesian approach}

In the following we focus on an alternative Bayesian method, which often allows for a more transparent interpretation than frequentist inference.
The Bayesian approach 
combines a prior distribution of the underlying population parameters with the likelihood of the
data given the population parameters to produce a posterior distribution of the parameters given the data, which is then used 
in a decision theoretic framework where a threshold is chosen to minimize the posterior expected loss function. 
The combination of a multinomial likelihood function with a 
Dirichlet conjugate prior drastically improves computability. 
Simulating the posterior distribution of the 
key FPR/TPR parameters to construct a Bayesian credible set is 
facilitated  by writing them 
as  a vector function of the multinomial parameters whose posterior distribution is given analytically through the conjugate prior. 
The empirical results 
suggest that the Bayesian 
method 
is less conservative and suggests replacement of more doctors by the machine algorithm than the frequentist method does.

%

Computing a Bayesian posterior distribution requires two key inputs: 
the likelihood of the data given the population value of $\theta_{H}$,
and a prior distribution for 
the underlying population parameters.
Each observation in the dataset consists of both the doctor diagnoses and the ground truth of pregnancy outcomes, 
and follows a multinomial distribution with four categories: diagnosed as risky and abnormal birth outcome; not diagnosed as risky and abnormal birth outcome; diagnosed as risky and normal birth outcome; not diagnosed as risky and normal birth outcome. The parameters of the multinomial probabilities are:
\bs
t_1 = \mathbb{E} Y \hat Y,\ t_2 = \mathbb{E}\[\(1 - Y\)\hat Y\],\ t_3 = \mathbb{E}\[Y \(1-\hat Y\)\],\ t_{4} = \mathbb{E}\[\(1 - Y\)\(1-\hat Y\)\].
\end{split}\end{align}
Since $t_1 + t_2 + t_3 + t_4 = 1$, we choose the first three as free parameters. The multinomial distribution is a completely specified parametric model and allows for exact likelihood Bayesian posterior distribution computation. The 
likelihood of the data given the free parameters $t=\(t_1, t_2, t_3\)$ depends only on a set of sufficient statistics $\hat t=\(\hat t_1, \hat t_2, \hat t_3\)$ 
where
\begin{align}\begin{split}\label{definition of t1 to t3} 
\hat t_1 = \frac1n \sum_{i=1}^n Y_i \hat Y_i,\quad
\hat t_2 = \frac1n \sum_{i=1}^n\[\(1-Y_i\) \hat Y_i\], 
\quad
\hat t_3 = \frac1n \sum_{i=1}^n\[Y_i \(1-\hat Y_i\)\].
\end{split}\end{align}
In particular, we can write the data likelihood as (where $\mathbb{D}$ refers to \textit{Data}):
\begin{align}\begin{split}\label{multinomial likelihood}
L(\mathbb{D}|t) = \binom{n}{n\hat{t}_1} t_1^{n\hat{t}_1} \binom{n-n\hat{t}_1}{n\hat{t}_2} t_2^{n\hat{t}_2} \binom{n-n\hat{t}_1-n\hat{t}_2}{n\hat{t}_3} t_3^{n\hat{t}_3} t_4^{n\(1-\hat{t}_1- \hat t_2 -\hat t_3\)}.
\end{split}\end{align}

Given a prior distribution of $t$, denoted as $\pi\(t\)$,
the posterior distribution of the parameter $t$, denoted as $p\(t \vert \mathbb{D}\)$ can usually be simulated or calculated analytically:
\begin{align}\begin{split}\nonumber
p\(t \vert \mathbb{D}\) = \frac{
\pi\(t\) L\(\mathbb{D} \vert t\)}{
\int \pi\(t'\) L\(\mathbb{D} \vert t'\) \mathrm{d}t'
}.
\end{split}\end{align}
We are interested in the posterior distribution of $\theta_H=\(\alpha_H, \beta_H\)$, which can be written as functions of the 
underlying multinomial parameters: 
$h\(t\) = \(\alpha_{H}\(t\), \beta_{H}\(t\)\)$.
If we can simulate from the posterior distribution $p\(t \vert \mathbb{D}\)$ and denote the simulated draws as $t_r, r=1,\ldots,R$, then the 
posterior distribution of $\theta_{H}$, denoted as
$p\(\theta_H \vert \mathbb{D}\)$, can be estimated using the empirical distribution of $h\(t_r\), r=1,\ldots,R$.


Computing the posterior distribution using a multinomial likelihood is facilitated by the Dirichlet conjugate prior. A Dirichlet prior on the K-dimensional simplex of $t_k,k=1,\ldots,K$, where
$\sum_{k=1}^K t_k=1$, is specified by hyper-parameters 
$\gamma=\(\gamma_k,k=1,\ldots,K\)$. Its density is given by
$\pi\(t\vert \gamma\)=\frac{1}{B\(\gamma\)} \prod_{k=1}^K t_k^{\gamma_k-1}$, where $B\(\gamma\)$ is the multivariate Beta function. 
See for example \cite{kotz2019continuous}.
The case with $\gamma_k=\gamma, \forall k=1,\ldots,K$
 is called 
a symmetric Dirichlet distribution. 
In one dimension, where $K=1$, the
Dirichlet distribution specializes to a Beta distribution, which in turn includes the uniform distribution as a special case
when $\gamma=1$.

%

The posterior distribution is also Dirichlet with parameters $\hat\gamma=\(\hat\gamma_k,k=1,\ldots,K\)$, where
$\hat{\gamma}_k = \gamma_k + n \hat t_k$ for each $k=1,\ldots,K.$
Simulating from the posterior distribution of
$\theta_{H} = \(\alpha_{H}, \beta_{H}\)$
can be accomplished by recomputing  $\theta_{H,r} = h\(t_r\)$ after
drawing for $t_r$ from 
the posterior Dirichlet distribution with parameters
$\hat\gamma$. 
Specifically, for each 
$t_r = \(t_{r,k}\)^{K}_{k = 1}$, we calculate 
\begin{align}\begin{split}\label{trans t to theta}
&\theta_{H, r} = h\(t_r\) = \(\alpha_{H}\(t_r\),
\beta_{H}\(t_r\)\), \quad\text{where}\\ 
&\alpha_{H}\(t_r\) = \frac{t_{r,2}}{t_{r,2} + t_{r,4}},\ 
\beta_{H}\(t_r\) = \frac{t_{r,1}}{t_{r,1}+t_{r,3}}.
\end{split}\end{align}
The simulated draws $\theta_{H,r}, r=1,\ldots,R$ can be used to estimate posterior probabilities. For example, 
%
the posterior probability that 
$\theta_{H}$ lies below the ROC curve, denoted as
\begin{align}\begin{split}\label{probbelowroc}
\int_{\theta_{H} \ \text{below}\ \text{ROC}}
p\(\theta_{H} \vert \mathbb{D}\) \mathrm{d}\theta_{H}, 
\end{split}\end{align}
can be estimated by the fraction of times where
$\theta_{H, r}$ lies below the ROC curve:
\bs
\frac1R \sum_{r=1}^R \mathds{1}\(\theta_{H,r}\ \text{lies below the machine ROC curve}\) = \frac1R \sum_{r=1}^R \mathds{1}\(\beta_{H}\(t_r\) \leq g\(\alpha_{H}\(t_r\)\)\).
\end{split}\end{align}

We are mainly interested in the {\it maximum} posterior
probability of $\theta_{H}$ that are 
\textit{dominated} simultaneously 
by a single point on the ROC curve, denoted as
\begin{align}\begin{split}\label{max bayesian simple point prob}
q_{max} \equiv \max_{\alpha \in \[0, 1\]} 
\int_{\alpha_{H} \geq \alpha, 
\beta_{H} \leq g\(\alpha\)}
p\(\theta_{H} \vert \mathbb{D}\)
\mathrm{d}\theta_{H},
\end{split}\end{align}
and the corresponding maximizing point, denoted as 
$\theta_{max} = \(\alpha_{max}, \beta_{max}=g\(\alpha_{max}\)\)$, where
\begin{align}\begin{split}\label{argmax posteriorprob}
\alpha_{max} \equiv 
{\arg\max}_{\alpha \in \[0, 1\]} \left\{
\int_{\alpha_{H} \geq \alpha,
\beta_{H} \leq g\(\alpha\)}
p\(\theta_{H} \vert \mathbb{D}\)
\mathrm{d}\theta_{H}\right\}. 
\end{split}\end{align}
Both 
\eqref{max bayesian simple point prob} 
and \eqref{argmax posteriorprob} 
can be estimated by the simulated posterior draws of
$\theta_{H, r}, r = 1, \ldots, R$,  by
\begin{align}\begin{split}\label{bayesian basic numerical}
\hat{q}_{max} & = \max_{\alpha \in \[0, 1\]}\frac1R 
\sum_{r=1}^R  \mathds{1}\(\alpha_{H}\(t_r\) 
\geq \alpha, \beta_{H}\(t_r\) \leq g\(\alpha\)\),\quad\\
\hat{\alpha}_{max} =& 
{\arg\max}_{\alpha \in \[0, 1\]}
\sum_{r=1}^R  \mathds{1}\(\alpha_{H}\(t_r\)
\geq \alpha, \beta_{H}\(t_r\) \leq g\(\alpha\)\).
\end{split}\end{align}
%

More generally, decision-theoretic Bayesian analysis
minimizes 
posterior expected losses. Let
$\rho\(\theta_{M}, \theta_{H}\)$ 
denote a loss function that represents the disutility of replacing a human FPR/TPR pair $\theta_H$ by a machine FPR/TPR pair $\theta_M$. Consistent with Assumptions \ref{basic assump} and \ref{ROC assump}, we specify
$\rho\(\theta_{M}, \theta_{H}\) = 0$ 
when $\alpha_{M} \leq \alpha_{H}$ and $\beta_{M}
\geq \beta_{H}$, or when the machine FPR/TPR pair strictly dominates the human FPR/TPR pair. Otherwise, $\rho\(\theta_{M}, \theta_{H}\)$ may be strictly positive. Given the choice of 
$\rho\(\theta_M, \theta_H\)$, a point on the ROC curve can 
be chosen to minimize the posterior expected loss
\begin{align}\begin{split}\label{posterior expected loss}
\theta_M^0 = {\arg\min}_{\theta_{M}: \beta_M = g\(\alpha_M\)}
\int 
\rho\(\theta_{M}, \theta_{H}\)
p\(\theta_{H} \vert \mathbb{D}\) 
\mathrm{d}\theta_{H}, 
\end{split}\end{align}
which is also estimated using the simulated posterior draws of $\theta_{H,r}, r=1,\ldots,R$ as
\bs
\hat\theta_M^0 = {\arg\min}_{\theta_{M}: \beta_M = g\(\alpha_M\)}
\frac1R \sum_{r=1}^R \rho\(\theta_M, \theta_{H,r}\).
\end{split}\end{align}

In particular, \eqref{max bayesian simple point prob} is a
special case of \eqref{posterior expected loss} when
$\rho\(\theta_M, \theta_{H}\) = 
1 - \mathds{1}\(\alpha_{H} \geq 
\alpha_{M}, \beta_{H} \leq \beta_{M}\)$.
It is also possible to explore more general loss functions.
For example, we may weight the loss of replacing $\theta_H$ by $\theta_M$ by the distance between $\theta_H$ and the machine ROC curve:
\begin{align}\begin{split}\label{euclidean distance loss function}
\rho\(\theta_{M}, \theta_{H}\) = 
1 - 
\mathds{1}\(\alpha_{H} \geq 
\alpha_{M}, \beta_{H} \leq \beta_{M}\)
\cdot 
\inf_\theta\left\{\Vert \theta_{H} - \theta\Vert: 
\theta \text{ on ROC}\right\}, 
\end{split}\end{align}
where $\Vert\cdot\Vert$ is the Euclidean norm. 
Intuitively, the farther a human FPR/TPR pair $\theta_H$ is from the ROC curve, the smaller the loss (or the larger the benefit) of replacing it with a \textit{dominating} point on the machine ROC curve. 
Alternatively, we can replace the distance from $\theta_H$ to the ROC curve by the distance of each $\theta_M$ to the complement of the set of FPR/TPR points that dominate $\theta_H$, and use the loss function:
\begin{align}\begin{split}\label{complement set euclidean distance loss function}
\rho\(\theta_{M}, \theta_{H}\) = 
1 - 
\mathds{1}\(\alpha_{H} \geq 
\alpha_{M}, \beta_{H} \leq \beta_{M}\)
\cdot 
\min\(\alpha_H-\alpha_M, \beta_M - \beta_H\).%
\end{split}\end{align}


The Euclidean distance from the machine ROC curve may also not fully capture the loss of
not replacing a human FPR/TPR pair $\theta_H$ that lies under
the machine ROC curve. For example, a diagonal ROC curve 
corresponds to
completely randomized decision making without using feature information. Similarly, a ROC curve below the diagonal can be generated by a decision rule 
that awards rather than penalizes mis-classification errors. 
These considerations suggest that replacing a human FPR/TPR pair $\theta_H$ below the diagonal line with a dominating $\theta_M$ along the machine ROC curve should not entail losses. We can 
adopt a convention of normalizing the loss to be $1$ when $\theta_M$ does not dominate $\theta_H$, i.e., when $\alpha_M \geq \alpha_H$ or $\beta_M \leq \beta_H$.

We 
then 
specify the loss $\rho\(\theta_M, \theta_H\)$ when 
$\theta_M$ dominates $\theta_H$ and when $\theta_{H}$ lies between the machine ROC curve 
and the diagonal line. Each of such $\theta_H$ can be reproduced 
by a convex combination of a point on the machine ROC curve and a point on the diagonal line, using a decision rule that randomizes
between the point on the machine ROC curve and the point on the diagonal. 
For example, $\(\alpha_H, g\(\alpha_H\)\)$ is generated by a machine based decision rule $\hat Y_{i} = \mathds{1}\(m\(X_{i}\) > c_0\)$ for some threshold $c_0$, while $\(\alpha_H, \alpha_H\)$ is generated by a fully randomized decision $\hat Y_{i} = \mathds{1}\(U_{i} \geq 1-\alpha_H\)$ where $U_{i}$ is independently uniformly distributed on $\(0,1\)$. 
Then $\(\alpha_H, \beta_H\)$ can be generated by a lottery that places weight
$\(\beta_H-\alpha_H\)/\(g\(\alpha_H\)-\alpha_H\)$ on $\hat Y_{i} = \mathds{1}\(m\(X_{i}\) > c_0\)$ and the remaining weight on $\hat Y_{i} = \mathds{1}\(U_{i} \geq 1-\alpha_H\)$. 
More precisely, if 
$V_{i}$ denotes 
an independent random variable uniformly distributed on $\(0,1\)$, the decision rule
\begin{align}\begin{split}\nonumber
\hat Y_{i} = \mathds{1}\(V_{i} \leq \frac{\beta_{H} - \alpha_{H}}
{g\(\alpha_{H}\) - \alpha_{H}}\)
\mathds{1}\(m\(X_{i}\) > c_{0}\) + 
\mathds{1}\(V_{i} > \frac{\beta_{H} - \alpha_{H}}
{g\(\alpha_{H}\) - \alpha_{H}}\)
\mathds{1}\(U_{i} > 1 - \alpha_{H}\)
\end{split}\end{align}
generates the $\(\alpha_H, \beta_H\)$ FPR/TPR pair. 
The larger the weight $\lambda_{\theta_{H}} 
\equiv \frac{\beta_{H} - \alpha_{H}}
{g\(\alpha_{H}\) - \alpha_{H}}$ placed on the machine ROC curve,
the smaller the benefit, or the larger the loss, of replacing $\theta_H$ with a dominating pair $\theta_M$ on the machine ROC curve.
The randomization scheme motivates using 
the weight 
$\lambda_{\theta_{H}}$ 
as a component of the loss function
\begin{align}\begin{split}\label{decomposition weight loss function}
\rho\(\theta_{M}, \theta_{H}\) = 
1 - 
\mathds{1}\(\alpha_H \geq \alpha_M, \beta_H \leq \beta_M\)
\(1 - \max \left\{\lambda_{\theta_{H}}, 0\right\}\).
\end{split}\end{align} 
Other randomization schemes can also be used. 
We refer to the above construction as a vertical decomposition. 
Similarly, a horizontal decomposition replaces the weights by 
$\lambda_{\theta_{H}} 
\equiv \frac{\beta_{H} - \alpha_{H}}{\beta_H - g^{-1}\(\beta_{H}\)}$.  
In addition, we experiment with decomposing the distance between $\theta_M=\(\alpha_M, \beta_M\)$ 
and the diagonal line into two parts. The first part is the distance from $\theta_M$  to the boundary of the set of FPR/TPR points that do not dominate $\theta_H$. The second part is the distance from this boundary to the diagonal line. The corresponding loss function that mirrors 
is
\begin{align}\begin{split}\label{complement set decomposition weight loss function}
\rho\(\theta_{M}, \theta_{H}\) = 
1 - 
\mathds{1}\(\alpha_H \geq \alpha_M, \beta_H \leq \beta_M\)
\(1 - \max \left\{\lambda_{\theta_{H},\theta_M}, 0\right\}\),
\end{split}\end{align} 
where distance can be measured either horizontally, implying that
$\lambda_{\theta_{H},\theta_M}=\frac{\beta_M-\alpha_H}{\beta_M-\alpha_M}$, or vertically, implying that
$\lambda_{\theta_{H},\theta_M}=\frac{\beta_H-\alpha_M}{\beta_M-\alpha_M}$.

When $\alpha_M < \alpha_H$ and $\beta_M > \beta_H$, replacing $\theta_H$ by $\theta_M$ is beneficial.
The benefit, $b\(\theta_M, \theta_H\)$, is likely to be larger when $\theta_H$ is farther below the machine ROC curve. 
When either $\alpha_M > \alpha_H$ or $\beta_M < \beta_H$, we expect a cost $\ell\(\theta_M, \theta_H\)$. 
A general loss function 
consists of both a cost component and a benefit component: 
\begin{align}\begin{split}\label{loss and benefit functions together}
\hspace{-.1in}\rho\(\theta_{M}, \theta_{H}\) = 
\mathds{1}\(\alpha_H < \alpha_M\ \text{or}\ \beta_H > \beta_M\)
\ell\(\theta_M, \theta_H\) -
\mathds{1}\(\alpha_H \geq \alpha_M, \beta_H \leq \beta_M\)
b\(\theta_M, \beta_H\).
\end{split}\end{align} 
For each 
loss function, 
a guardrail threshold is used 
to account for the status quo of human decision making. A human decision maker is only replaced by the machine algorithm when the posterior expected loss of the replacement is below a given prespecified level, or when the posterior expected benefit exceeds a certain level.


In addition to the deterministic decision rules of relying exclusively on either an individual doctor or the machine algorithm, 
randomized 
decision rules can be generated by mixing between a
given point on the machine ROC curve and an individual doctor. Let $\omega$ denote the randomizing probability of using the machine algorithm, and $1-\omega$ the corresponding probability of relying on the human decision makers. The minimized expected posterior loss 
of the randomized decision rule can be written as
\begin{align}\begin{split}\label{mixed strategy minimized posterior loss}
\min_{\omega \in \[0,1\], \theta_M: \beta_M = g\(\alpha_M\)}
\int \rho\(\omega \theta_M + \(1-\omega\) \theta_H, \theta_H\) p\(\theta_H \vert \mathbb{D}\) \mathrm{d}\theta_H.
\end{split}\end{align}
Without expert domain knowledge of medical diagnoses, it is difficult to 
implement \eqref{loss and benefit functions together}  and \eqref{mixed strategy minimized posterior loss}. 
Investigating expert domain knowledge is beyond the scope of the current paper. 

Our decision framework differs conceptually from the minimization of posterior expected loss function in Bayesian point estimation. A Bayesian point estimator of $\theta_H$, denoted as $\hat\theta_{H,B}$, is often defined as the minimizer of an expected posterior loss function:
\bs
\hat\theta_{H,B} = \arg\min_{\theta\in \left[0,1\right]^2} \int \ell\(\theta - \theta_H\) p\(\theta_H \vert \mathbb{D}\) \mathrm{d}\theta_H.
 \end{split}\end{align}
The estimation loss function $\ell\(\theta - \theta_H\)$ typically depends only on $\vert\theta-\theta_H\vert$ 
and is increasing in $\vert \theta-\theta_H\vert$, which is different from our decision loss function $\rho\(\theta_M, \theta_H\)$.

\subsection{Discussion: incentives and costs}
\label{Discussion: incentives and costs}

The visual illusion shown in Figure \ref{figure 2} that hampers the comparison
between machine learning algorithms
and human decision makers
is an immediate artifact of Jensen's inequality. 
An optimal ROC generated by a correctly specified propensity score function
is necessarily concave, since it is the primal function of a concave program. See for example \cite{feng2022properties}.
For the $j$th doctor,
denote the FPR as $\alpha_j=\mathbb E \(1-Y_{ij}\)\hat Y_{ij}/\(1 - p_j\)$ and the TPR as $\beta_j=\mathbb E Y_{ij} \hat Y_{ij} / p_j$, where
$p_j = \mathbb P\(Y_j =1\)$.
If all the $\(\alpha_j, \beta_j\)$ pairs lie on a strictly concave ROC curve $g\(\cdot\)$,
these variables are related by $\beta_j = g\(\alpha_j\)$. 
When the patient cases across doctors are drawn from the same population, $p_j=p$ is a constant. It follows from 
Jensen's inequality that
\begin{align}\begin{split}\nonumber
\bar \beta = \frac1J \sum_{j=1}^J \beta_j =
\frac1J \sum_{j=1}^J g\(\alpha_j\) < g\(\frac1J \sum_{j=1}^J \alpha_i\) = g\(\bar \alpha\),
\end{split}\end{align}
when $\alpha_{j}$ are not all equal.
This important observation appears to have gone largely unnoticed by the
literature.
Similarly, if each patient $i$ is diagnosed using a point on the human ROC curve $\(\alpha_{i,j}, \beta_{i,j}\)$ 
and if $\alpha_{i,j}, i=1,\ldots,n_j$, are not all equal,
it is also the case that
\begin{align}\begin{split}\nonumber
\beta_j = \frac{1}{n_{j}} \sum_{i=1}^{n_j} \beta_{i,j} =
\frac{1}{n_j} \sum_{i=1}^{n_j} g\(\alpha_{i,j}\) < g\(\frac{1}{n_j} \sum_{i=1}^{n_j} \alpha_{i,j}\) = g\(\alpha_j\).
\end{split}\end{align}
Lemma \ref{strictly concave roc} in the supplementary appendix provides a set of primitive conditions under which the aggregate human decision maker FPR/TPR lies strictly below the machine ROC curve. 

In general, without the restrictions in Assumptions \ref{basic assump} and \ref{ROC assump}, the decision maker $j$ may use a
model $\hat y = \mathds{1}\(q_{j}\(x,u\) > c_{j}\(x,u\)\)$. 
For decision maker $j$ minimizing expected loss of misclassification conditional on the feature $x$ and $u$,  consider 
the following loss matrix: 
\begin{center}
        \begin{tabular}{ l c r }
                Loss matrix & $\hat y=0$ & $\hat y =1$ \\
                $y=0$ & $0$ & $c_{01}\(x, u\)$ \\
                $y=1$ & $c_{10}\(x, u\)$  & $0$
        \end{tabular}
                 \end{center}
where the misclassification losses $c_{01}$ and $c_{10}$ are now dependent on $x$ and $u$. 
Equivalently, decision maker $j$ also minimizes the ex ante expected loss:
\begin{align}\begin{split}\nonumber
\mathbb{E}\[c_{10}\(X, U\)q_{j}\(X, U\)\(1-\hat Y\) + c_{01}\(X, U\)\(1 - q_{j}\(X, U\)\) \hat Y\].
\end{split}\end{align}
The optimal decision rule is $\mathds{1}\(q_{j}\(x, u\) > \frac{c_{01}\(x, u\)}{c_{10}\(x, u\) + c_{01}\(x, u\)} \eqqcolon c_{j}\(x, u\)\)$. 
Consequently, for $\(x_{1}, u_1\)$, $\(x_{2}, u_2\)$ 
sharing the same propensity score
$q_{j}\(x_{1}, u_1\) = q_{j}\(x_{2}, u_2\)$,
the choice of the decision maker
can be different due to cost differences:
$c_{j}\(x_{1}, u_1\) \neq c_{j}\(x_{2}, u_2\)$. Such a decision rule 
is considered in \cite{elliott2013predicting}. 
Even if 
all doctors have full information processing capacity, 
so that $q_j\(x,u\)=p\(x,u\)$ for a correctly specified propensity 
function $p\(x,u\)$, 
an arbitrary FPR/TPR pair below the machine ROC curve 
can still be rationalized by some decision rule 
$\hat y = \mathds{1}\(p\(x,u\) > c_j\(x,u\)\)$ corresponding to a reverse-engineered cost function $c_j\(x,u\)$. 


To illustrate using a synthetic example without $u$,
let $x$ be  uniformly distributed on $\(0,1\)$ and $p\(x\) = x$. 
Let the cutoff threshold function be 
\begin{align}\begin{split}\nonumber
c\(x\) = 
\begin{cases}
0,\quad x< 1/4;\\
2\(x-1/4\), \quad 1/4 < x <	3/4;\\
1, \quad x > 3/4.
\end{cases}
\end{split}\end{align}
Then $\hat y = \mathds{1}\(x < \frac12\)$ and $(\text{TPR, FPR})=(1/8, 3/8)$, 
which lies strictly below the diagonal line. 
An example when the feature may enter into the cutoff threshold is bank lending, 
where $y \in \{0,1\}$ denotes whether a loan is paid back, and $\hat y \in \{0,1\}$ denotes whether a loan application is approved. 
The amount of principal and interests are features $x$ that are likely to affect the probability $p\(x\)$ of defaulting by the borrower, 
but that are also likely to affect the cost perception $c\(x\)$ by the lender.
Assumption \ref{basic assump} essentially excludes such difficulty of interpreting a FPR/TPR pair arising from incentive heterogeneity.

To understand whether
the cost perception of doctors might vary systematically with the risk level of their patients, 
we use 
the FPR rank for each doctor to measure the conservativeness of their diagnoses.
Doctors may be more willing to tolerate a higher FPR 
for higher risk patients 
when they perceive
riskier patients as less healthy and more likely to incur a higher cost of underdiagnosis. 
We then estimate the correlation between the FPR ranks
and a pregnancy risk level indicator. 
In particular, we regress 
$q_{i,j}$, 
the predicted abnormal birth rate for patients 
diagnosed by doctor $j$ from a random forest algorithm 
described in section \ref{data and algorithm},
against $r_j$, the doctor's FPR rank: 
\begin{align}\begin{split}\nonumber
q_{i,j} = \beta_{0} + \beta_{1}r_{j} + \epsilon_{i,j}. 
\end{split}\end{align}
The doctors' FPR rank are normalized to $\[0, 1\]$. 
If the FPR of doctor $j$ is higher than
80\% of the doctors, then $r_{j} = 0.8$.
The regression results for doctors with no less than
300 diagnoses 
and for doctors with no less than 500
diagnoses 
are summarized in Table
\ref{quantile machine prediction}. 
The regression results suggest that there is no
strong evidence that the FPR rank of 
doctors is associated with the risk level of
their patients.

\begin{table}[htbp]\centering
    \caption{Predicted abnormal birth probabilities and
    FPR rank of doctors}
    \label{quantile machine prediction}
    \begin{tabular}{l*{2}{c}l*{4}{c}}
    \toprule
    & \multicolumn{2}{c}{Linear Term Only} & &
    \multicolumn{4}{c}{With Nonlinear Term} \\
    \cline{2-3} \cline{5-8}
    & \multicolumn{1}{c}{(1)} & \multicolumn{1}{c}{\;\;\:(2)}
    & \multicolumn{1}{c}{} & \multicolumn{1}{c}{(3)}
    & \multicolumn{1}{c}{(4)}
    & \multicolumn{1}{c}{(5)} & \multicolumn{1}{c}{(6)}
    \\
    \midrule
    FPR rank & 0.003 & \;\;\:-0.000
    & & -0.105 & 0.120
    & -0.072 & 0.156 \\
                 & (0.20) & \;\;\:(-0.01)
    & & (-1.71) & (0.86)
    & (-0.95) & (0.91) \\
    $\text{Rank}^{2}$ & & & & 0.111
    & -0.456
    & 0.072 & -0.493 \\
    & & & & (1.69) & (-1.32)
    & (0.89) & (-1.16) \\
    $\text{Rank}^{3}$ & & & & & 0.385
    & & 0.380 \\
    & & & & & (1.61)
    & & (1.28) \\
    Constant     & 0.206 & \;\;\:0.205
    & & 0.224 & 0.205
    & 0.216 & 0.197 \\
                 & (21.79) & \;\;\:(18.28)
    & & (17.95) & (14.33)
    & (14.44) & (11.33) \\
    \midrule
Number of doctors & 584 & \;\;\:367
    & & 584 & 584
    & 367 & 367 \\
    $R^{2}$ &  0.0001 & \;\;\:0.0000
    & & 0.0059 &  0.0107
    & 0.0025 & 0.0071 \\
    \bottomrule
    \end{tabular}
    \fnote{
    \textit{Notes}: 
    The table shows regression results for the relation between 
doctors' preferences and their patients' risk levels. 
    In the baseline linear regression, the dependent variable is the patient's predicted risk by a random forest algorithm.
    The explanatory variable is the FPR rank of the patient's doctor. 
    Regression results including quadratic and cubic terms are reported in columns (3) to (6) 
    in the table. 
There are 584 doctors diagnosing at least 300 patient cases, and 367 doctors diagnosing at least 500 patient cases. 
    Standard errors  in parentheses are clustered by doctors.
    }
\end{table}

The use of statistical
significance as a benchmark for comparing doctors with the machine algorithm is motivated by 
the difficulty of fixating a social cost function, 
even if we assume that informed human decision makers are fully rational and impose
strong functional form assumptions. Designing a social loss function empirically 
relies 
on the domain knowledge of the specific decision maker.
Our
replacement algorithm allows for the choice of the cutoff threshold to be specific to each
individual doctor, and does not require a homogeneous preference function to be applied
to all doctors.
%
The baseline comparison between true positive rates and false positve rates can also potentially mask the severity of individual cases, the identification of which requires a richer dataset with more precise measurements on the outcome variable. While these information are not 
available in our dataset, exploring richer data information can be a fruitful future direction of research.

A responsible phase-in protocol based on multiple years of data on performance observations can incorporate 
long term health and wellbeing outcomes into the definition of performance and quality for decision making.
Our data is cross-sectional and does not follow doctors or patients over time. 
When a panel data becomes available in which doctors and patients can be observed over multiple time periods, 
future work can potentially 
generalize the current approach to allow 
for sequential testing and decision making. 
The challenge of accounting for 
longer term health outcomes is 
partially reflected by 
the current strategy that replaces a doctor’s diagnosis 
by the machine algorithm only when there is overwhelming 
statistical evidence suggesting the superiority of the machine decision. 

\section{Data and Algorithm Description}\label{data and algorithm}


Our data came from the National Free Prepregnancy Checkups
(NFPC) project.
Starting from 2010, NFPC offers free health checkup for
couples planning to conceive 
and is conducted across 31 provinces in China.
The dataset contains more than 300 features for each observation,
including age, demographic characteristics,
results from medical examinations and clinical tests,
disease and medication history, pregnancy history, lifestyle and
environmental information of both wives and husbands.
The data also include the 
birth outcome, which is denoted as normal ($y = 0$) or abnormal ($y = 1$).
In addition,
the dataset records doctors' IDs and diagnoses of pregnancy risk.
The doctors' diagnoses are graded according to 4 levels:
0 for normal, 1 for high-risk wife,
2 for high-risk husband, and 3 for high-risk wife and husband.
In this paper,
we group the doctor's diagnoses into 2 categories,
where a grade of 0
corresponds to normal pregnancy and
any higher grade corresponds to a diagnosis
of risky pregnancy.

Our original dataset 
includes 
3,330,304 couples that have
pregnancy outcomes between January 1, 2014, and December 31, 2015.
We exclude duplicate features and  samples for which either doctors' diagnoses or 
more than 50\% of feature values are missing.
We were left with 168 features of 
1,137,010 couples who were diagnosed by 28,716 doctors. 
Of these observations, 61,184 couples (5.38\%) had abnormal birth outcomes.

A basic measurement of the prediction quality of a binary classifier
is accuracy, which is the proportion of correct predictions.
However, accuracy itself is inadequate 
in many applications. As the adverse birth rate is about
5\%, a naive classifier that categorizes all cases as low risk would achieve an accuracy of
nearly 95\%. This is clearly controversial. The doctors' overall accuracy is 73.63\%, and
24.04\% of pregnancies are diagnosed as high-risk.
Doctors' aggregate 
false positive rate and true positive rate
are 
0.2379
and 0.2843.
In contrast,
the FPR and TPR of the naive predictor are both zero.
In other words,
doctors are willing to tolerate a higher FPR in order to achieve a higher TPR.

We focus on doctors who diagnosed more than 300 patients. There are 584 such doctors corresponding to 584,181 patient cases.
Using a ratio of 7:3 and stratification methods,
we randomly split our sample into a classification set and a performance set.
We use the classification set to train a 
machine learning algorithm,
and then we categorize a doctor as ``replaced'' or 
``retained''
according to our statistical procedure.
The performance set is used to compute the FPR/TPR of the combined diagnoses.
The data cases are stratified into the four classes based on the actual birth outcomes and the doctors' diagnoses: true positive cases, true negative cases,
false negative cases and false positive cases.  
The random splitting process is performed through each class of data
 with the same ratio of 7:3. Subsequently,
the data in each class are merged to form the classification and performance sets.
In total, there are 408,661 cases
in the classification set
and 175,520 cases
in the performance set.


We further split the classification set into two subsets. The
first subset is the training set used to estimate the parameters in a machine learning algorithm. 
The second subset is the validation set 
used to evaluate
the performance of the algorithm and 
obtain the machine ROC curve.
Splitting of the classification set into the 
training and validation subsets is done using a ratio of 4:3 and 
the same stratification 
method described above.
There are 233,435 cases in the training set and 175,226 cases in the validation set. 

We experimented with several algorithms. 
We begin with
decision trees, which are widely used in many machine learning applications.
However, a single tree method usually has high variance despite its low bias, and tends to overfit by generating results with substantial variation 
 even when 
a small amount of noise is present in the data.
In contrast, random forest (RF) is a commonly used ensemble learning algorithm proposed by
\cite{breiman2001random}, and overcomes these problems by constructing a
collection of decision trees that are trained using different feature subspaces and 
bootstrapped samples.
The predictions of each tree are aggregated to
make predictions. 
Appendix \ref{MachineAlgorithm} provides more details about the random forest algorithm.\footnote{We choose the following parameter settings. The number of estimators $N$ is $100$.
The number of max features per node $M$ is $50$.
The minimum number of samples required to split is 50.}



\begin{remark}
The classification of doctors' diagnoses by both husband and wife 
raises a question regarding the validity of 
Assumption \ref{basic assump} that rules out intra-individual incentive heterogeneity, if 
the relative costs of false positives to false negatives differ
between when a single person is deemed risky and when both are deemed risky. 
A multiclass classification model with four outcome labels differentiating abnormal pregnancy outcomes due to high risks in wife, husband and both spouses can allow for heterogeneous relative costs of false positives to false negatives across these cases. 
However, in our dataset, 
the outcome label is only binary indicating whether the birth is classified as normal or abnormal. The lack of the corresponding four-way 
outcome label makes it difficult to compute an accurate TPR and FPR for each of the diagnosis class.

As an illustration, we empirically calculate 
the FPR/TPR using the binary outcome label for each of the four diagnosis cases,  and report
the aggregate summary in the following table:

\begin{table}[htbp]
\centering
\caption{Doctor performance summary by multiple diagnosis classes}
\vspace{.05in}

\begin{tabular}{cccccccc}
\toprule
\multicolumn{8}{l}{Panel A: FPRs for different diagnosis classes} \\
\midrule
      min patient cases &  num negative  & FP1 & FP2 & FP3 & FPR1 & FPR2 & FPR3 \\
\midrule
        300 & 554195 & 51896 & 33266 & 20954 & 0.0934 & 0.0600 & 0.0524 \\
        500 & 470562 & 42636 & 27362 & 21233 & 0.0906 & 0.0581 & 0.0451 \\
\toprule
\multicolumn{8}{l}{Panel B: TPRs for different diagnosis classes} \\
\midrule
       min patient cases & num positive  & TP1 & TP2 & TP3 & TPR1 & TPR2 & TPR3 \\
\midrule
300 & 29986 & 3250 & 1750 & 1941 & 0.108 & 0.0584 & 0.0647 \\
        500 & 24758 & 2589 & 1384 & 1412 & 0.104 & 0.0559 & 0.0570 \\
\toprule
\multicolumn{8}{l}{Panel C: Precisions across diagnosis classes} \\
\midrule
min patient cases & \multicolumn{2}{c}{wife precision} & \multicolumn{2}{c}{husband precision} & \multicolumn{2}{c}{couple precision} & \\
300 & \multicolumn{2}{c}{0.0589} & \multicolumn{2}{c}{0.0500} & \multicolumn{2}{c}{0.0626} & \\
        500 & \multicolumn{2}{c}{0.0572} & \multicolumn{2}{c}{0.0481} & \multicolumn{2}{c}{0.0624} & \\
\bottomrule
\end{tabular}
\fnote{
\textit{Notes}: 
In the table, FPR1 denotes FPR where the label is abnormal birth and the diagnosis is wife being high risk; FPR2 and FPR3 are similarly defined for the diagnosis classes of husband being high risk and both spouses being high risk. 
Similarly for TPR1, TPR2 and TPR3. 
}
\label{tab:perf_to_acceptrate2}
\end{table}

On the one hand, in Panels A and B of Table \ref{tab:perf_to_acceptrate2}, 
there are some differences between the TPRs 
and FPRs 
among the three diagnosing classes of high risk wife, husband and both spouses. 
However, without the corresponding outcome level it is difficult to interpret these  ratios
since the total number
of abnormal births used in calculating 
the denominators 
 aggregates 
the three unobserved latent 
outcome labels corresponding to high risk wife, husband and both spouses. 

On the other hand, Panel C of 
Table \ref{tab:perf_to_acceptrate2} reports the precision for each of the three diagnosis classes (high risk wife, 
husband and couple). Precision is defined as the fraction of abnormal birth for each of the 
diagnosis classes. 
For example, wife (husband) precision represents the fraction of 
couples with a high  risk wife (husband) diagnosis who have an abnormal birth  outcome.
Couple precision is similarly defined.
In 
Panel C of Table \ref{tab:perf_to_acceptrate2}, the precision parameters for different
diagnosis classes are close to each other. 
The precision parameters, despite being 
an imperfect measurement, provide a more accurate indication suggesting that the relative costs of false positives to false negatives are 
not different 
between cases when a single person is deemed risky and when both are. 

\end{remark}

\section{Empirical Results for AI and doctor decision making}\label{classify doctors}


This section presents the empirical results from 
the heuristic frequentist approach described in section \ref{heuristic frequentist approach} and the Bayesian approach described in section \ref{the Bayesian approach}. Doctors who are not replaced by the machine algorithm are called retained doctors hereafter. 


\subsection{Results of the heuristic frequentist approach} \label{asymptoticResult}


The frequentist approach identifies a dominating segment of the machine ROC for machine algorithm classification. 
The quality of machine algorithm classification is evaluated by averaging over a uniform grid of points on this dominating segment. 
For each doctor, 
we obtain the 95\%  elliptical confidence area of $\theta_{H}$ centered around the sample estimates
using an estimate of the asymptotic covariance matrix based on 100 bootstrap samples.
We find the largest TPR point $\(\alpha_{\beta_{high}}, \beta_{high}\)$ 
and the smallest FPR point 
$\(\alpha_{low}, \beta_{{low}}\)$ of the confidence set.
The reference point 
of the doctor is then given by 
 $P = \(\alpha_{low}, \beta_{high}\)$,
as shown in Figure \ref{fg:case1}, which is used to classify
doctors into two groups.
%
When a doctor's reference point $P$ 
is below the machine ROC curve, 
corresponding to case 1 in section \ref{heuristic frequentist approach}, the 
doctor is replaced by the algorithm.
When a doctor's 
reference point $P$ 
is above the machine ROC curve, 
corresponding to cases 2 or 3 in section
\ref{heuristic frequentist approach},
either because this doctor is more capable than 
the algorithm or this doctor's capability 
cannot be precisely measured due to 
fewer patient observations  and a larger confidence set,
the doctor is retained in 
decision making.

Among the 584 doctors,
175 doctors are classified by the frequentist approach as replaced by the algorithm.
Consequently, patient cases in the performance set are diagnosed by 
175 machine doctors (30.0\%) and 409 (70.0\%) human doctors.
For each of the replaced doctors, 
$N = 100$ points on the dominating segment of the machine ROC curve are chosen corresponding to a uniform grid of the threshold value $c$
for the algorithm to make classification decisions, as 
summarized in Appendix \ref{macine decision thresholds}.
Figure \ref{fg:ai-doctors-300}
shows the results of this experiment.
The aggregate FPR/TPR pair of all doctors on the 
performance set, 
represented by the blue point,
 is 0.2065 for FPR
and 0.2264 for TPR.
The ROC curve of the random forest model on 
the performance set is the green curve.
The area under the curve (AUC) 
is 0.6851.

The yellow point
uses a decision threshold $c$ corresponding to the upper endpoint of the dominating segment of the machine ROC curve 
associated with a highest FPR/TPR pair of 
0.2017 and 0.3326.
Compared to the raw data where 
all diagnoses are performed by doctors, 
the yellow point
increases 
the TPR by 46.9\% and reduces 
the FPR by 2.3\%.
The green point
corresponds to the lower endpoint of the dominating segment of the machine ROC curve  and is 
associated with a lowest FPR/TPR pair of 
0.1836 and 0.2974.
Compared to the aggregate doctor diagnoses,  the green point
increases the TPR
by 31.4\% and reduces the FPR by 11.1\%.

\begin{figure}
  \centering
\subfigure[ROC curve and FPR/TPR pairs of doctors and machine decisions]{\includegraphics[width=0.7\columnwidth]{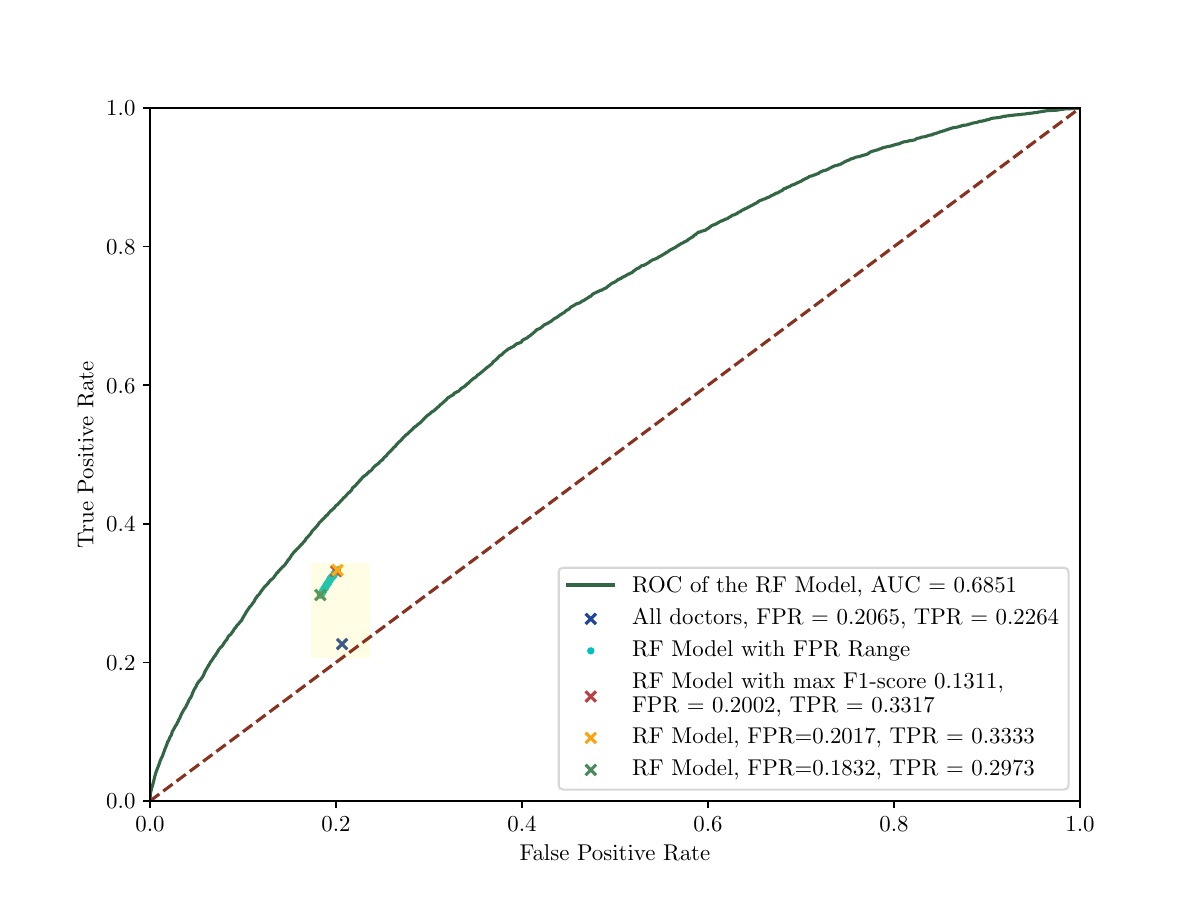}
\label{fg:ai-doctors-300-upper}
}
\subfigure[Zoomed-in view]{\includegraphics[width=0.7\columnwidth]{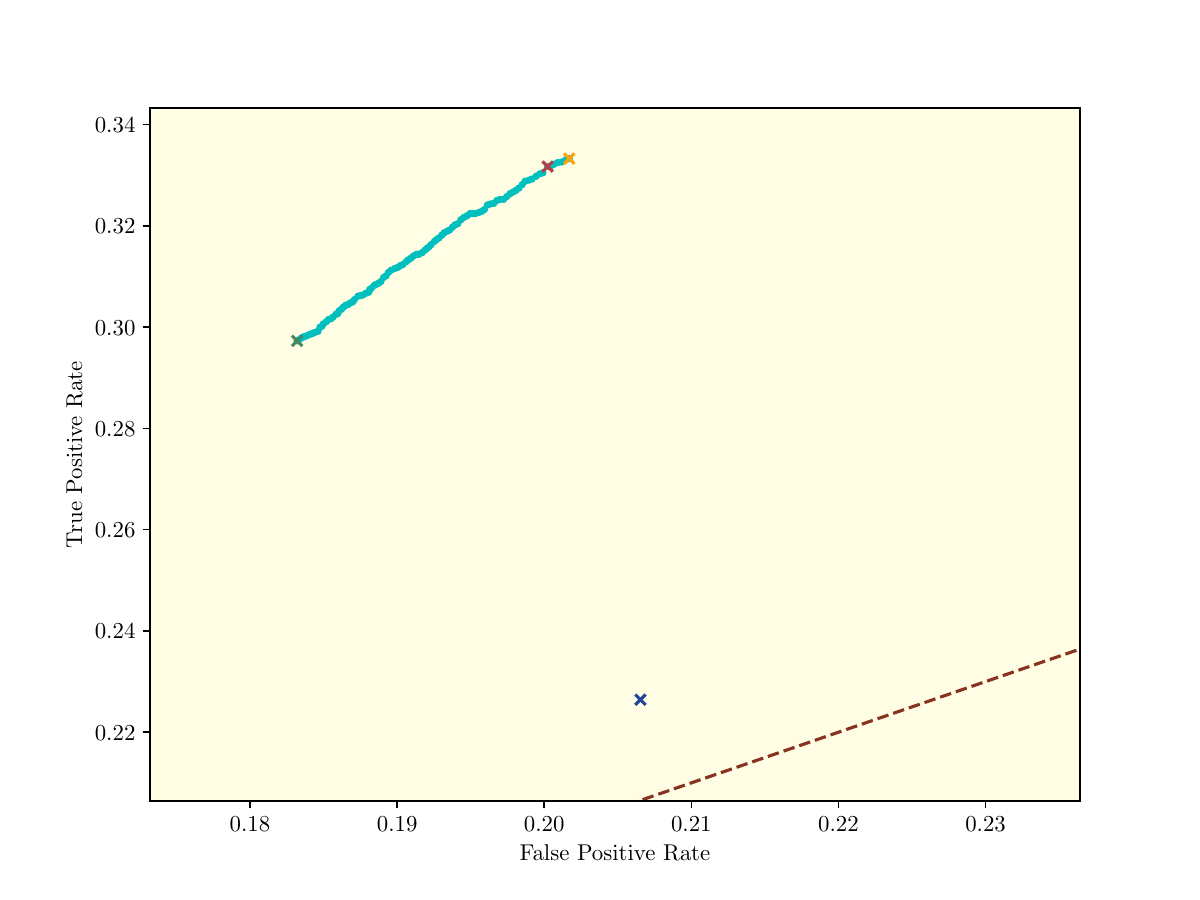}
\label{fg:ai-doctors-300-lower}
}
  \caption{Empirical results of combining doctors' and machine decisions using the heuristic frequentist approach (doctors' diagnoses >= 300).} 
\fnote{\textit{Notes}: Panel 
\subref{fg:ai-doctors-300-upper} displays the ROC curve of the Random Forest classifier 
and 
the blue point representing the FPR/TPR pair of all doctors in the test set. 
The cyan interval shows the aggregate FPR/TPR pairs of combining retained doctors and the machine algorithm 
using different cutoff threshold choices in the frequentist replacement strategy.
Panel \subref{fg:ai-doctors-300-lower} 
provides a magnified view of the cyan interval.
 }
  \label{fg:ai-doctors-300}
\end{figure}

The cyan interval in Figure \ref{fg:ai-doctors-300} reflects a trade-off between improving the TPR and reducing the FPR. 
Any point along the 
interval represents an enhancement over the aggregated doctors' FPR/TPR pair. 
Furthermore, the red point 
achieves the maximum F1 score\footnote{The F1 score is the harmonic average of precision and recall, and is often used as 
a single measurement that conveys the balance between precision and recall. It is widely used in evaluating the performance of algorithms (see for example \cite{baeza1999modern}).} of 
0.1309 along 
the dominating segment of the machine ROC curve. 
Figure \ref{frequentist replaced doctor scatter plot-300} is a scatter plot of the FPR/TPR pairs of both replaced and retained doctors against the ROC curve in the performance dataset. We see that two adjacent points
can be of different replacement status.
This can be due to different numbers of patients treated
or the use of the test set for calculating the FPR/TPR pairs. 
An important caveat 
is that our test 
depends crucially on the sample size of the number
of patients treated by each doctor, 
which introduces an implicit bias that favors experimenting with junior doctors until 
the null hypothesis that a doctor is not to be replaced is rejected 
with sufficient evidence from seeing a large number of patient cases.

\begin{figure}
  \centering
\includegraphics[width=0.8\columnwidth]{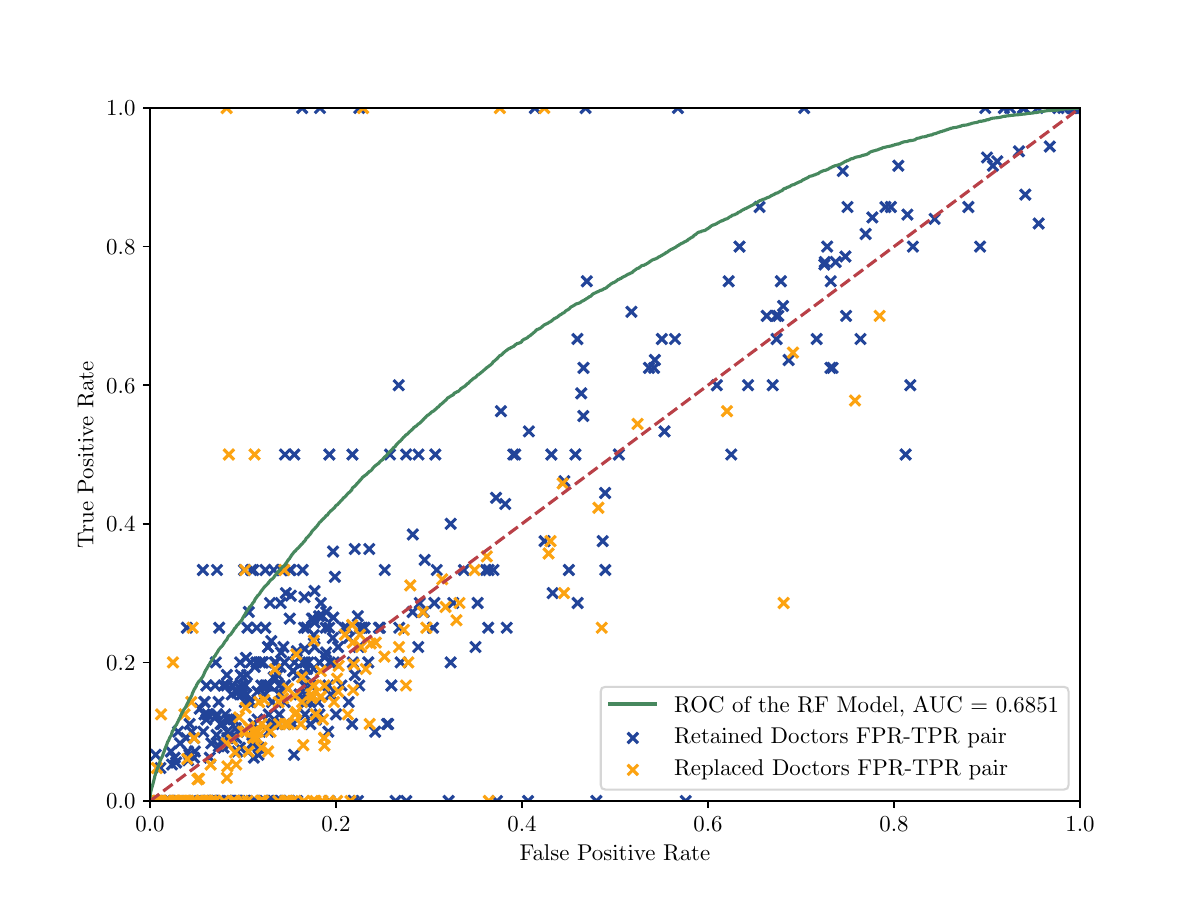}
\caption{Scatter plot of replaced and retained doctors in the test dataset using the frequentist approach (doctors' diagnoses >= 300)}
\fnote{\textit{Notes}: The FPR/TPR pairs in the scatter plot are calculated in the test dataset of patient cases for each doctor. 
It is possible for the FPR/TPR pairs of some of the replaced doctors to lie above the machine ROC curve in the test dataset. Two points
that are quite close to each other can be of different replacement status 
due to different numbers of patients treated.
}
  \label{frequentist replaced doctor scatter plot-300}
\end{figure}
When the quality of the stock of decision makers may
evolve over time,
allowing the juniors to learn from 
making potentially suboptimal decisions for their patients 
has
the potential of increasing the future pool of high quality seniors. This is likely to be a very
delicate issue involving complex welfare tradeoff between current and future generations of
patients.


\subsection{Results of the Bayesian approach} \label{bayesian result}


The prior distribution $\pi(t)$ is chosen to be 
a symmetry Dirichlet distribution with the 
parameters $\(\gamma_1, \gamma_2, \gamma_3, \gamma_4\)$ all set to $\gamma = 0.01$. 
For doctor $j$'s $i$th patient case, 
$\hat{Y}_{ij}$ is the doctor's diagnosis
and $Y_{ij}$ is the ground-truth label of pregnancy outcome.
%
The likelihood for data 
is given in \eqref{multinomial likelihood} where  $\hat Y_i$, $Y_i$ and $n$
in \eqref{definition of t1 to t3} are replaced by 
$\hat Y_{ij}$, $Y_{ij}$ and $n_j$.
The posterior distribution is
also a Dirichlet distribution with parameters 
$(\hat{\gamma}_1, \hat{\gamma}_2, 
\hat{\gamma}_3, \hat{\gamma}_4)$ where
$\hat{\gamma}_{k} = \gamma + n\hat{t}_k$ 
for $k = 1, 2, 3, 4$. 
The posterior distribution of 
$\theta_{H} = (\alpha_{H}, \beta_{H})$ 
can be simulated in a two step procedure.
In the first step, $R$ samples of 
$\{t_r = (t_{1r}, t_{2r}, 
t_{3r}, t_{4r})\}_{r = 1}^{R}$ are drawn from the posterior Dirichlet distribution with parameters $(\hat{\gamma}_1, \hat{\gamma}_2, 
\hat{\gamma}_3, \hat{\gamma}_4)$. 
In the second step, we compute $\{\theta_{H, r} = 
\(\alpha_{H}\(t_{r}\), 
\beta_{H}\(t_{r}\)\)\}_{r = 1}^{R}$ 
with the simulated $R$ samples of $t_r, r=1,\ldots,R$,
using \eqref{trans t to theta}. 

The simulated draws $\theta_{H, r}, r=1,\ldots,R$ are used to calculate
\eqref{bayesian basic numerical}. 
If $\hat{q}_{max} \geq 0.95$,
we mark the doctor as replaced. 
The threshold $c$ that corresponds to 
the point $\(\hat{\alpha}_{max}, g\(\hat\alpha_{max}\)\)$ given by 
\eqref{bayesian basic numerical}
is used as the decision threshold 
for the algorithm. 
In Figure \ref{bayesian replaced},  the maximum coverage rate $\hat q_{max} = 0.9657$ at the black point on the machine ROC curve. The  doctor is replaced by the algorithm.
In Figure \ref{bayesian not replaced}, 
the maximum coverage rate $\hat q_{max} = 0.3179$. There is not 
sufficient evidence to replace the doctor's decisions with 
the algorithm's.

\begin{figure}
  \centering
  \subfigure[The case where a doctor is replaced by the machine algorithm]
  {\includegraphics[width=0.7\columnwidth]{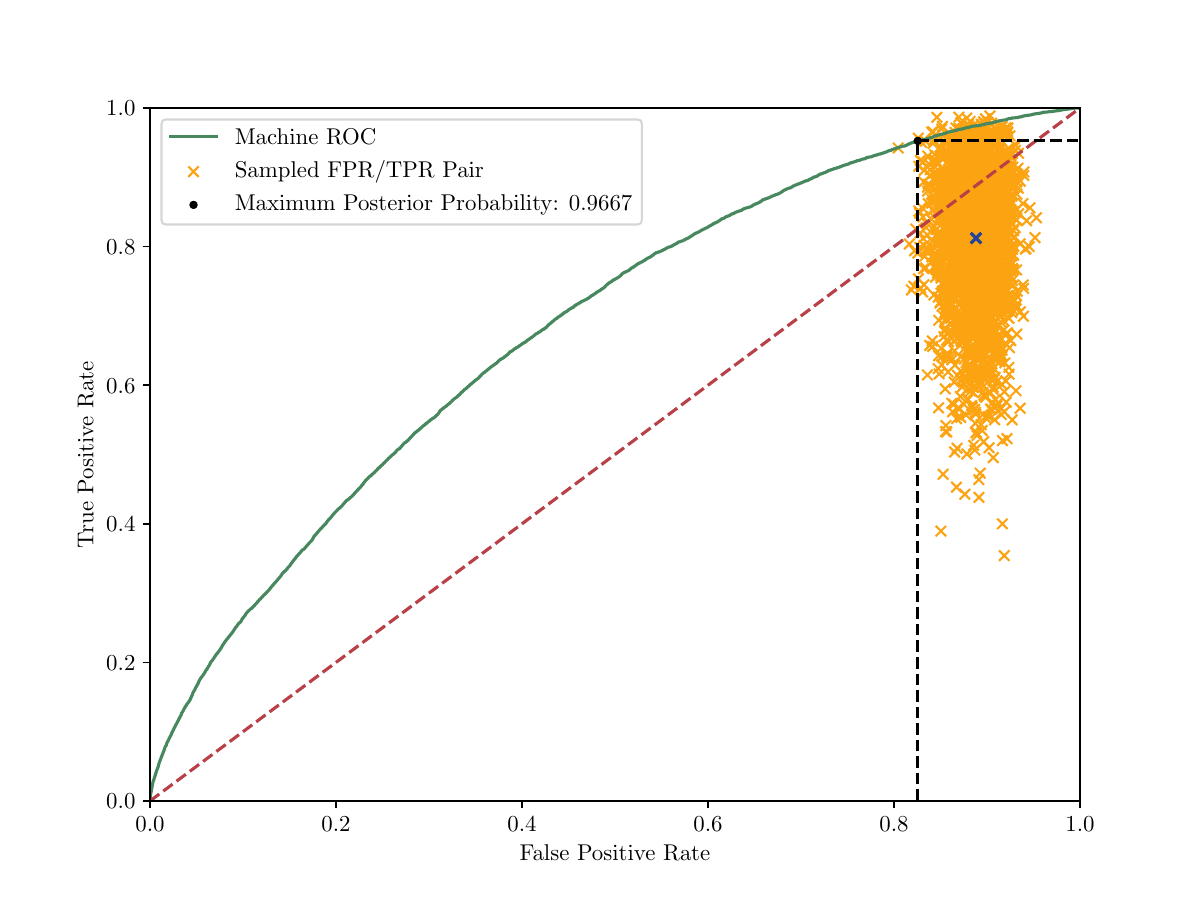}
  \label{bayesian replaced}}
  \subfigure[The case where a doctor is not replaced by the machine algorithm]
  {\includegraphics[width=0.7\columnwidth]{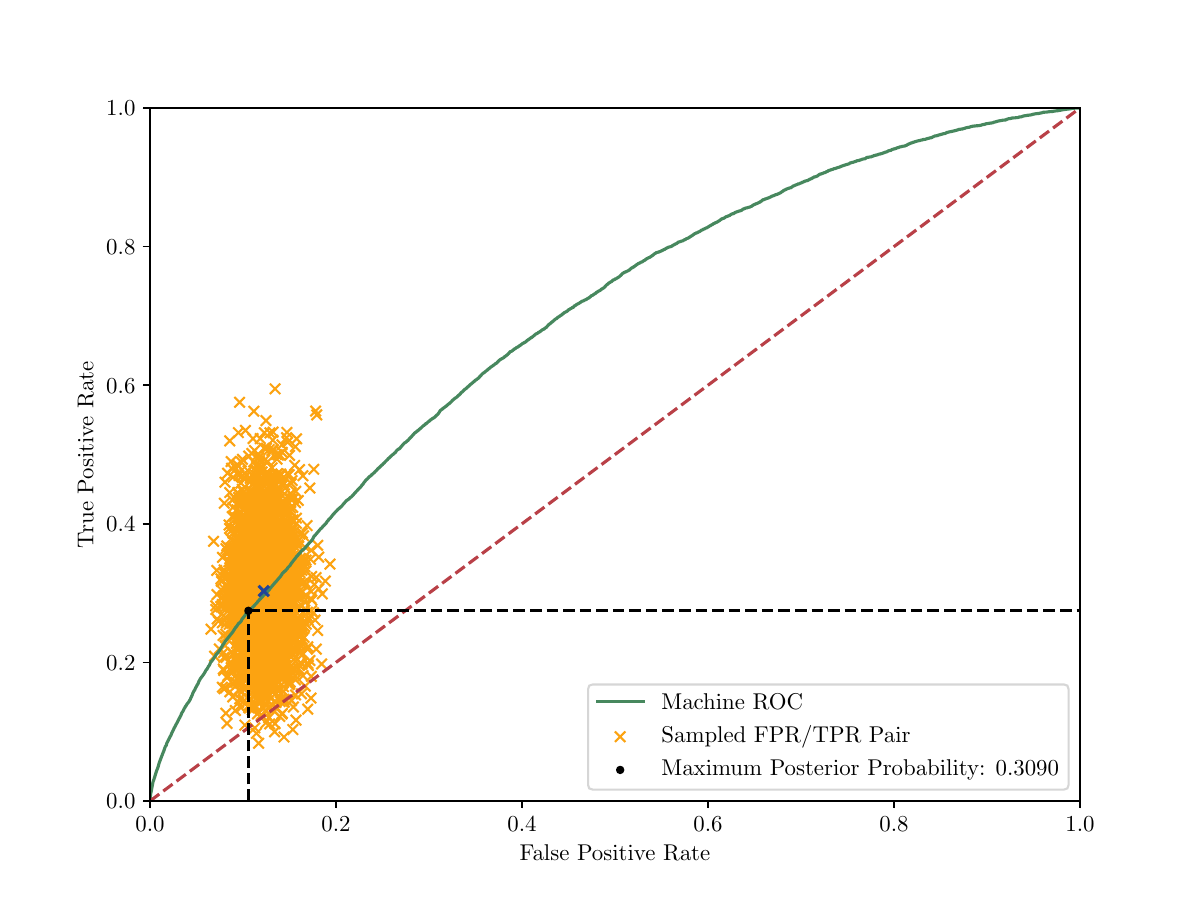}
  \label{bayesian not replaced}}
  \caption{Illustration of replaced and retained doctors by the Bayesian approach}
  \fnote{\textit{Notes}:
  In the above two figures, the yellow points are sampled from the projection of the posterior Dirichlet distribution. One thousand draws are made for each panel. 
  On the machine ROC curves, the black points are the points with the largest dominating posterior probability mass. 
  Panel \subref{bayesian replaced} shows a doctor whose posterior probability of being dominated 
is 
greater than 0.95. This doctor will be replaced by the machine decision rule corresponding to the black point. 
  Panel \subref{bayesian not replaced} illustrates the opposite case where the doctor will not be replaced. 
  }
  \label{fg:bayesian-1}
\end{figure}


Figure \ref{fg:ai-doctors-bayesian-curve-300} shows the experiment results 
for doctors who had diagnosed more than 300 cases.
A total of 269 (46.1\%) out of 584 doctors are classified as replaced. 
The remaining 315 (53.9\%)  are retained human doctors.
The combined decisions by the retained human doctors and the algorithm's decisions
indicated in the figure by the yellow point, 
result in an FPR of 0.1856 and a TPR of 0.3319,
corresponding to an increase of 46.6\% in the TPR and a reduction of 10.1\% in the FPR.
Similar to Figure \ref{frequentist replaced doctor scatter plot-300},
Figure \ref{bayesian replaced doctor scatter plot-300} is a scatter plot of the FPR/TPR pairs
of both replaced and retained doctors against the ROC curve in the performance dataset.

\begin{figure}
  \centering
  \subfigure[ROC curve and FPR/TPR pairs of doctors and machine decisions]{
\includegraphics[width=0.7\columnwidth]{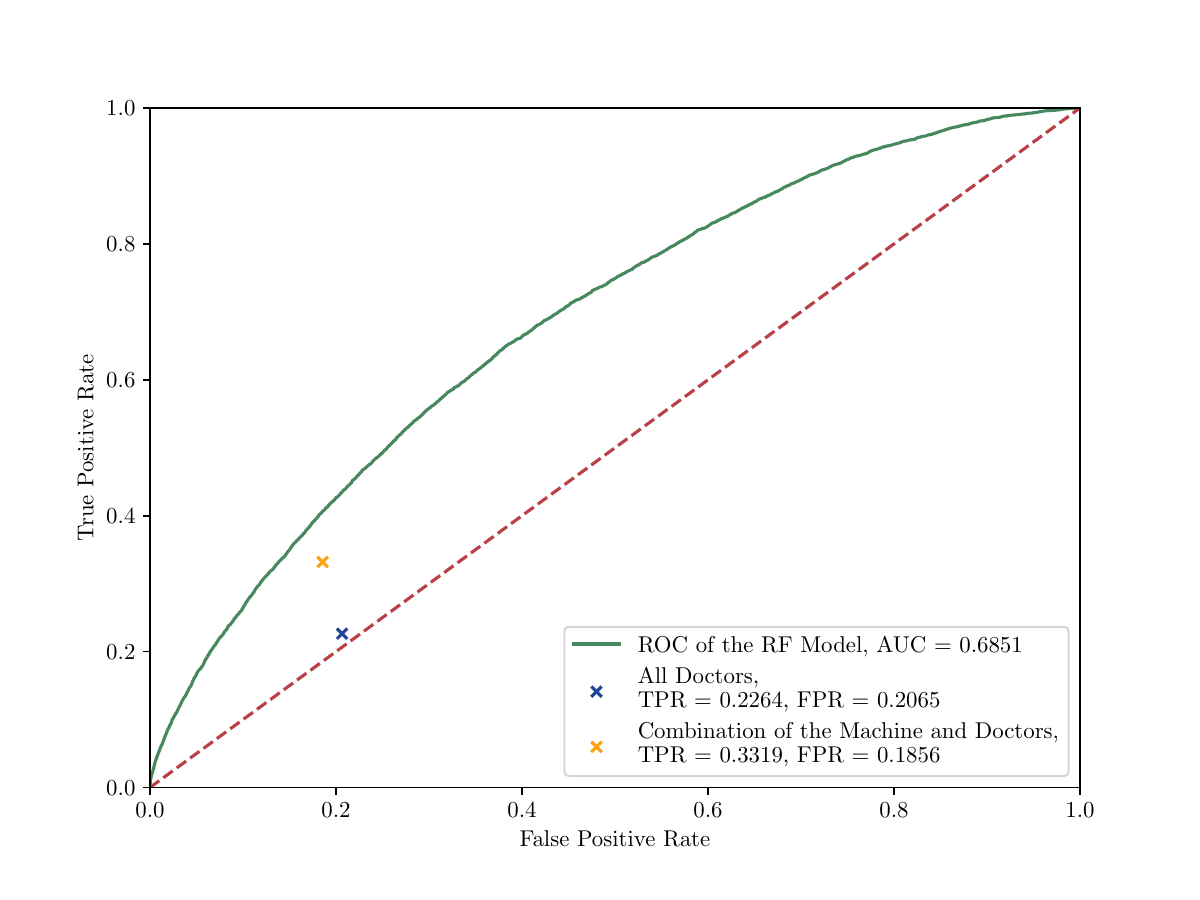}\label{fg:ai-doctors-bayesian-curve-300}
}
  \subfigure[Scatter plot of Bayesian doctor replacement in the performance dataset]{\includegraphics[width=0.7\columnwidth]{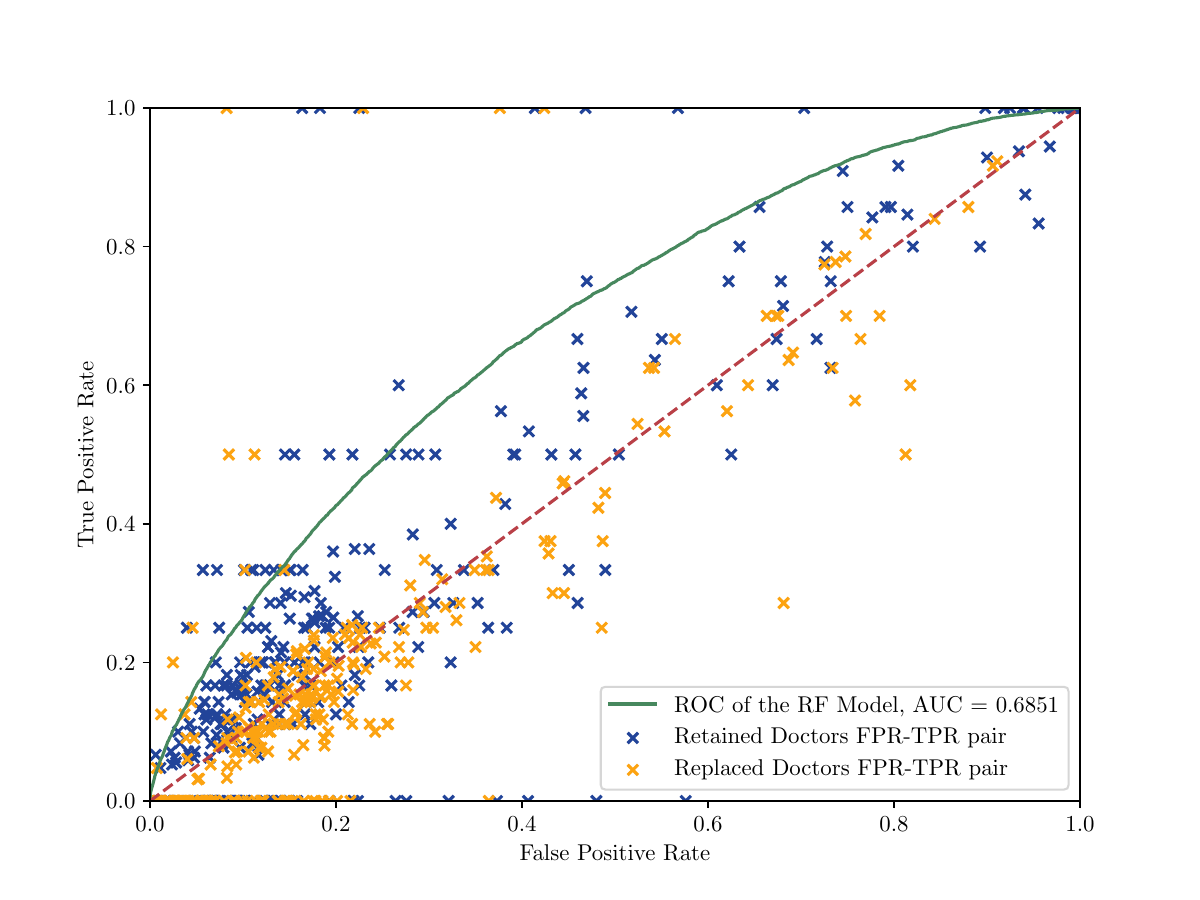}\label{bayesian replaced doctor scatter plot-300}}
  \caption{Empirical results of the Bayesian approach (doctors' diagnoses >= 300).} 
\fnote{\textit{Notes}: The panel \subref{fg:ai-doctors-bayesian-curve-300} draws the ROC curve in the test set of the Random Forest classifier, the
FPR/TPR pair of all doctors, and the FPR/TPR pair of doctor/machine combination. The panel \subref{bayesian replaced doctor scatter plot-300} scatterplot represents the FPR/TPR pairs of individual doctors, such that yellow points correspond to the replaced doctors and the blue points non-replaced doctors.}
  \label{fg:ai-doctors-bayesian-300}
\end{figure}


We experiment with several alternative 
definitions described in equations \eqref{euclidean distance loss function} to 
\eqref{complement set decomposition weight loss function} for the 
Bayesian posterior expected loss functions 
in
\eqref{posterior expected loss}.
%
These results are shown in 
Figure \ref{fg:ai-doctors-bayesian-score-300},  where panel \subref{fg:ai-doctors-bayesian-score-horizonal-lower-300}
magnifies a portion of panel \subref{fg:ai-doctors-bayesian-score-horizonal-upper-300} containing the replacement paths.
Ranking the doctors by their expected posterior losses allows us to replace only 
the doctors whose expected loss of being substituted by the machine algorithm is relatively low. 
Figure \ref{fg:ai-doctors-bayesian-score-300}  traces the dynamics of the aggregate FPR/TPR pairs of 
combining the doctor with the machine decisions, when the ratio of replaced doctors increases 
from 0\% 
to 100\%. 

\begin{figure}
  \centering
   \subfigure[ROC curve and replacement paths]{\includegraphics[width=0.7\linewidth]{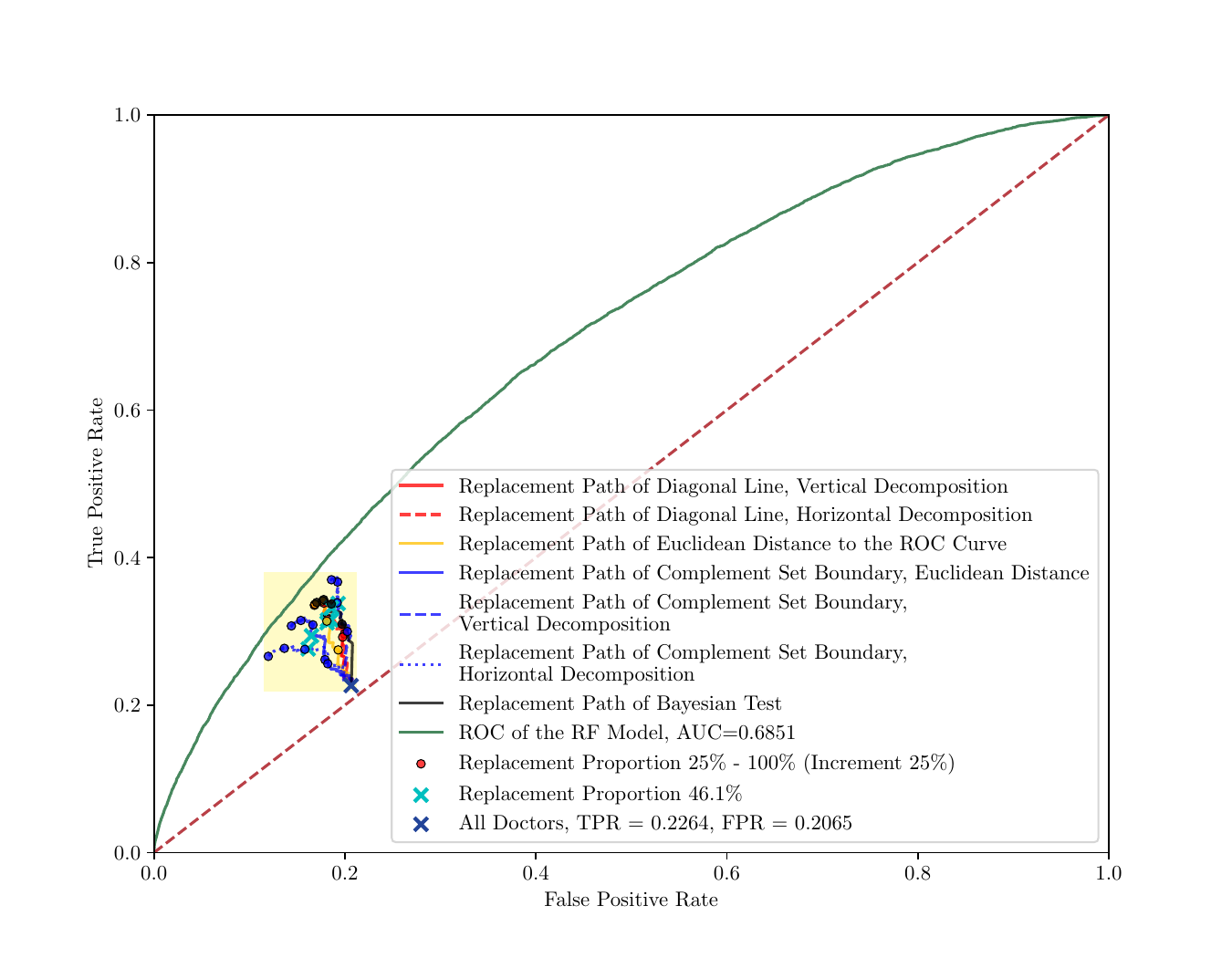}\label{fg:ai-doctors-bayesian-score-horizonal-upper-300}}
   \subfigure[Zoomed-in figure of replacement paths]{\includegraphics[width=0.7\linewidth]{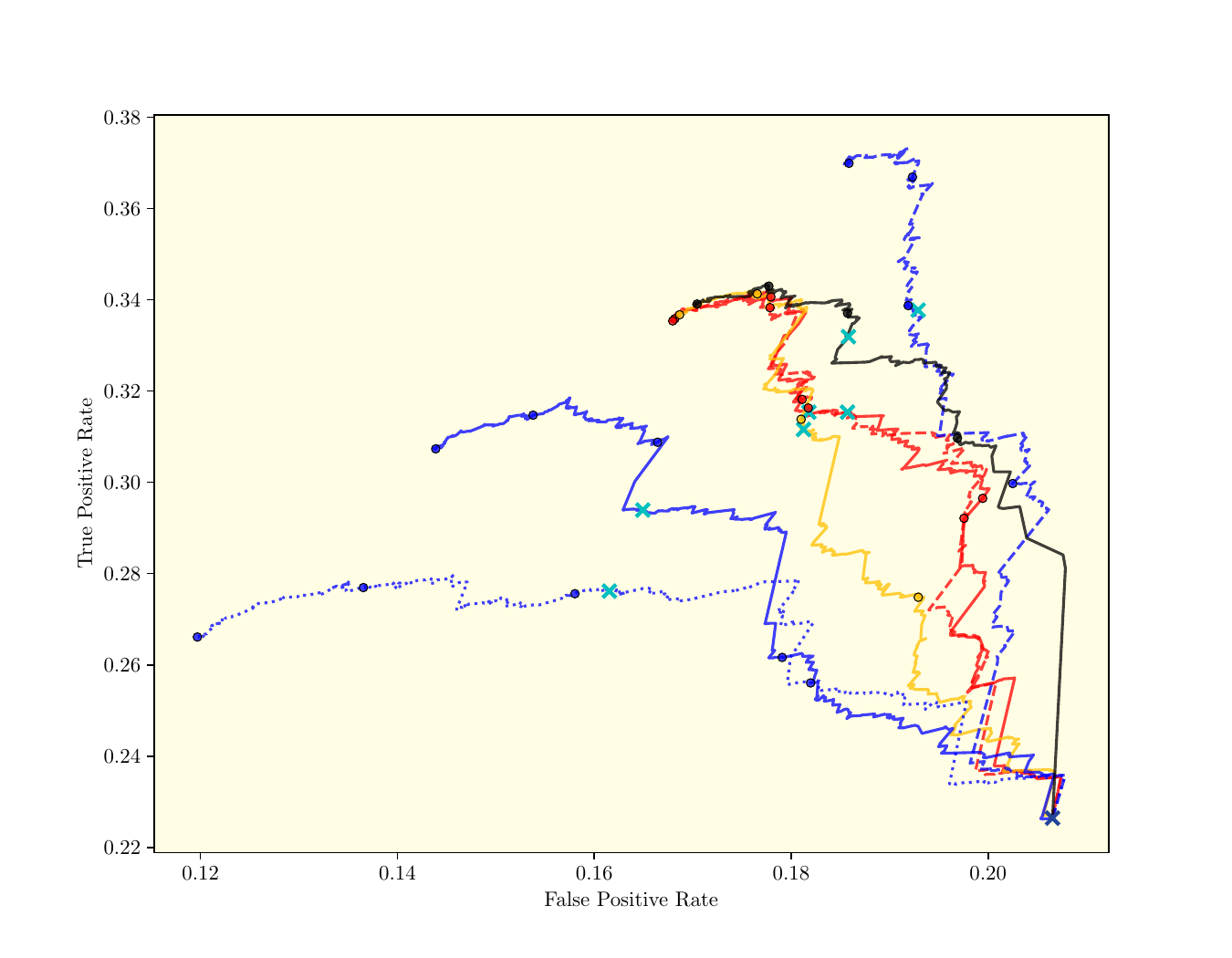}\label{fg:ai-doctors-bayesian-score-horizonal-lower-300}}
  \caption{Bayesian approach with loss functions constructed by Euclidean distance and randomized decomposition (doctors' diagnoses >= 300)}
\fnote{\textit{Notes}: The replacement paths trace the dynamics of the aggregate FPR/TPR pairs of combining 
the doctor decisions with the machine decisions, when the replacement ratio increases from 0\% of decisions entirely by human to 100\% of decisions entirely by machine algorithms. Panel \subref{fg:ai-doctors-bayesian-score-horizonal-lower-300}
magnifies a portion of panel \subref{fg:ai-doctors-bayesian-score-horizonal-upper-300} containing the replacement paths.
}
  \label{fg:ai-doctors-bayesian-score-300}
\end{figure}

In Figure 
\ref{fg:ai-doctors-bayesian-score-horizonal-upper-300} legends, {\it Bayesian Test} refers to the baseline construction in \eqref{max bayesian simple point prob}; 
{\it Euclidean Distance to the ROC curve} refers to the loss function 
in \eqref{euclidean distance loss function}; {\it Complementary Set Boundary Euclidean Distance} refers to the loss function in \eqref{complement set euclidean distance loss function}; {\it Diagonal Line Horizontal and Vertical Decompositions} refer to the loss function in \eqref{decomposition weight loss function}; {\it Complement Set Boundary Vertical and Horizontal Decompositions} refer to the loss function in \eqref{complement set decomposition weight loss function}. 
No individual loss function demonstrates saliently stronger discriminatory power over the entire replacement path. 
Furthermore, 
the incremental benefit going beyond a 50\% replacement rate is 
a small marginal reduction of FPR. 
In addition, the baseline 95\% credible level Bayesian test implemented using \eqref{max bayesian simple point prob} 
is not dominated by any of the competing methods in the FPR/TPR comparison. 

The null hypothesis of
both the Bayesian analysis 
and the frequentist analysis 
is the status quo of human decision making. A human decision maker is replaced by the machine algorithm only when there is strong evidence from the data favoring the machine algorithm. An alternative conceptual paradigm is to treat the machine algorithm as the status quo null hypothesis, where 
the human decision makers are by default replaced by the machine algorithm unless there is strong data evidence validating the superiority of human decision making. 

However, when we experiment with the alternative specification of the status quo, 
all except a few doctors are determined to be replaced 
by the machine algorithm. In the cross-point labeled as the ``dominate method'' in Figure \ref{fg:bayesian_reversed_300}, a doctor is 
retained if there exists a point on the machine ROC curve that is dominated by 
at least 95\% of the  posterior distribution of $p\(\theta_H \vert \mathbb{D}\)$. 
A mere 6 doctors are retained in this method. In the cross-point labeled as the ``above method'' in 
Figure \ref{fg:bayesian_reversed_300}, a doctor is 
retained if the posterior probability of $p\(\theta_H \vert \mathbb{D}\)$ above the machine ROC curve is more than 95\%. Only 7 doctors are retained using this method. For doctors who are replaced by the machine algorithm, the replacement point on the ROC curve is determined by the baseline Bayesian tests in \eqref{bayesian basic numerical}.
 
\begin{figure}
	\centering
\subfigure[ROC curve and FPR/TPR pairs of Doctors and machine decisions]{\includegraphics[width=0.7\columnwidth]{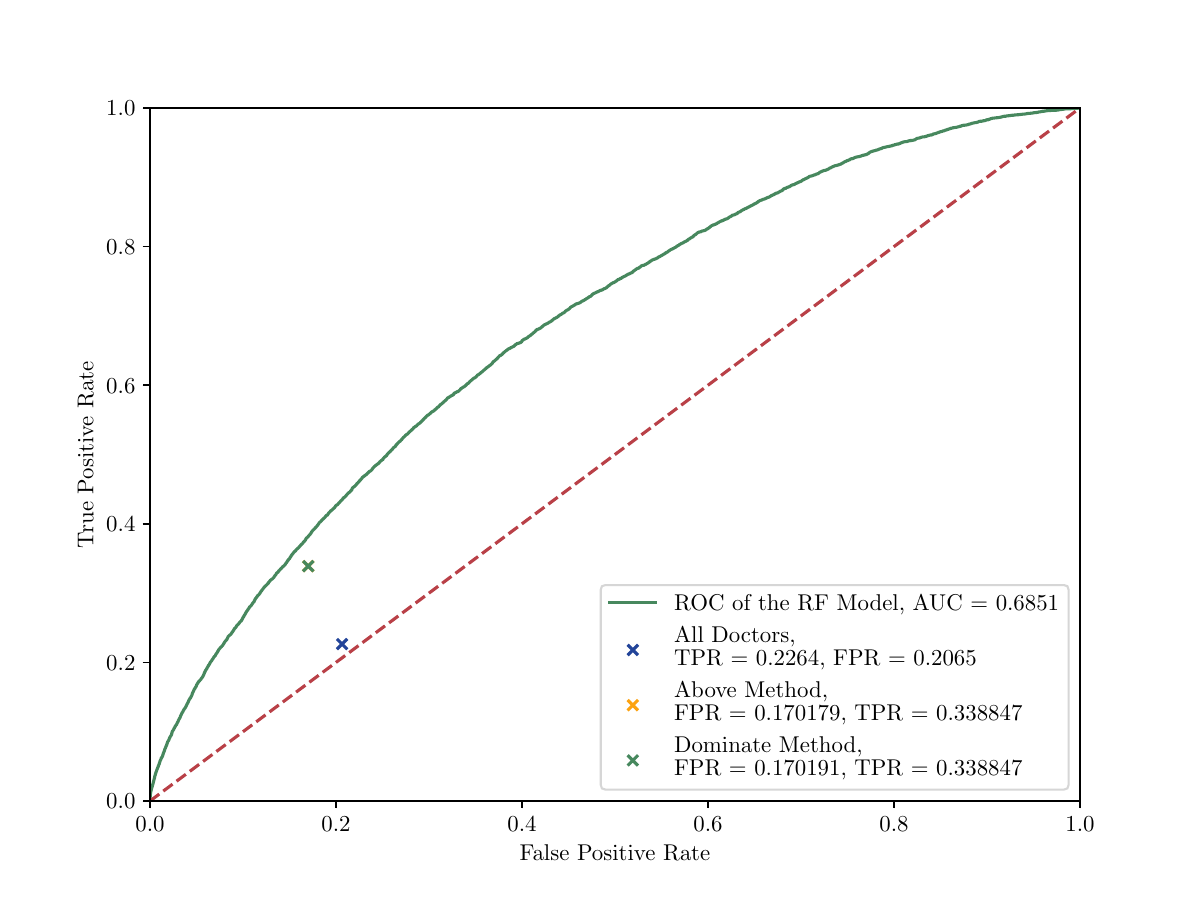}
\label{fg:bayesian_reversed_300-upper}
}
\subfigure[Zoomed-in FPR/TPR pairs of replacement outcomes]{\includegraphics[width=0.7\columnwidth]{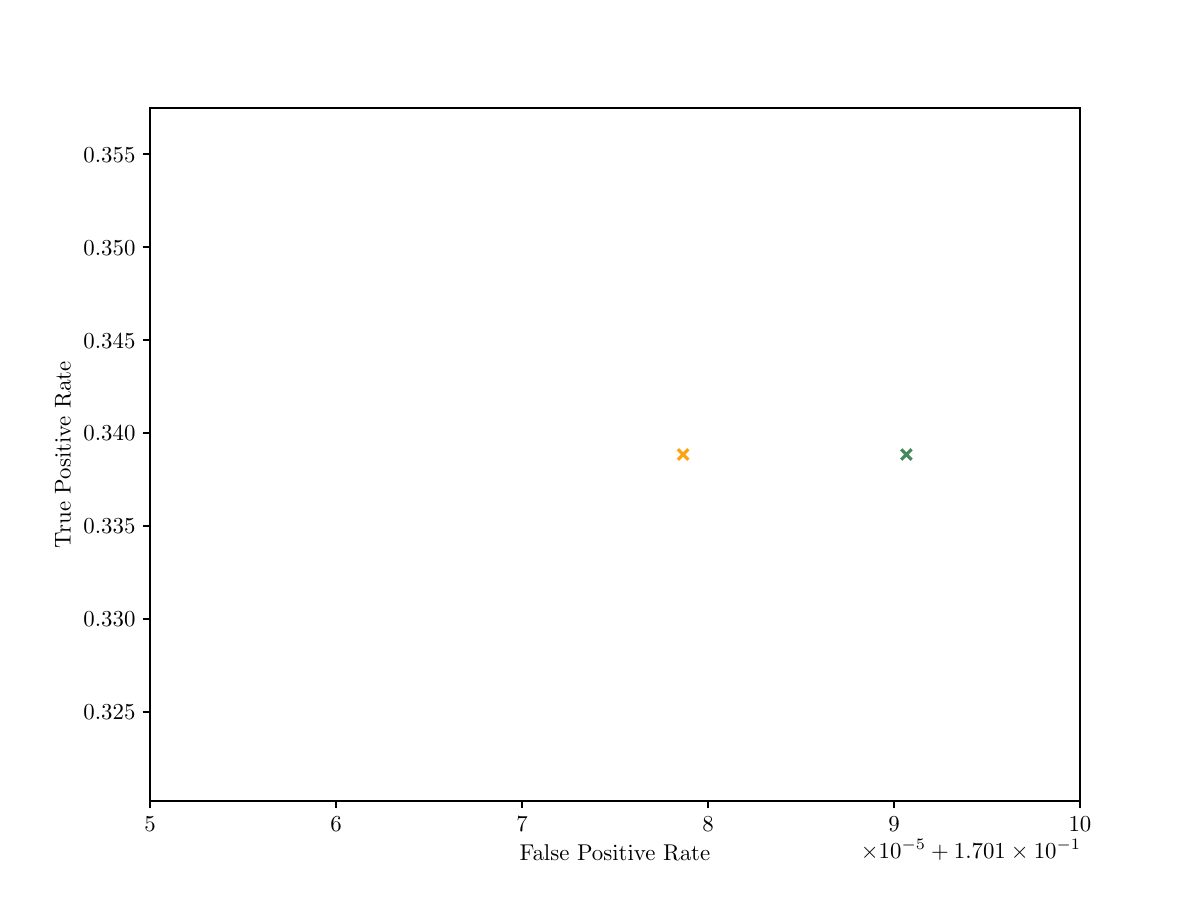}
\label{fg:bayesian_reversed_300-lower}
}
	\caption{Alternative Null Hypothesis of Machine Algorithm: doctors are replaced by default and are retained only with strong evidence.}
\fnote{
\textit{Notes}: In the ``dominate method'' cross point, a doctor is retained if there exists a point on the machine ROC curve that is dominated by at least 95\% of the  posterior distribution of $p\(\theta_H \vert \mathbb{D}\)$. A mere 6 doctors are retained in this method. In the ``above method'' cross point,
a doctor is only retained if the posterior probability of $p\(\theta_H \vert \mathbb{D}\)$ above the machine ROC curve is more than 95\%. Only 7 doctors are retained using this method.  These two points are almost indistinguishable in panel \subref{fg:bayesian_reversed_300-upper}. Their difference is very small as shown in 
\subref{fg:bayesian_reversed_300-lower}.
}
	\label{fg:bayesian_reversed_300}
\end{figure}

At high replacement ratios, the FPR and TPR levels decrease simultaneously across all loss functions. 
Even when all the doctors
are substituted by the machine algorithm, the resulting FPR/TPR pair still lies below the machine ROC curve. 
This is 
because different threshold cutoff values on the machine ROC curve are employed to replace different doctors and because of the concavity arguments formalized in Lemma \ref{strictly concave roc}. 

On the one hand, in 
the current empirical dataset, only a small fraction of doctors have FPR/TPR pairs above the machine ROC curve. 
We have not been able to combine
doctors with the machine algorithm to 
generate a FPR/TPR pair above the machine ROC curve. 
On the other hand, 
using a data generating process that allows for a substantial portion of doctors with FPR/TPR pairs to be above machine ROC curve,
reported in the 
synthetic data analysis
section \ref{synthetic data analysis}, we are able to combine the doctors and the machine algorithm to generate a FPR/TPR pair that outperforms even the machine ROC curve. 


Instead of deterministically replacing some doctors by the machine algorithm, 
Appendix \ref{randomized replacement} 
reports the results of a randomized replacement experiment
where each replaced doctor accepts the machine decision probabilistically according to a pre-specified acceptance rate $\lambda\in\[0,1\]$. 
As expected, the overall FPR/TPR pairs vary incrementally from when $\lambda=0$, corresponding to human decision making, to when $\lambda=1$, corresponding to deterministic replacement of 
doctors by the machine algorithm.


Appendix \ref{fine-grained replacement} considers a  more general form of substitution to replace a subset of patient cases for each doctor. 
We 
use patient case covariates to predict the likelihood that the doctor's diagnoses differ from the ground-truth outcome, and replace only those patient cases where the doctor is prone to make mistakes. 
While the resulting pattern of the replacement paths is quite different, 
the 
patient-case specific substitution approach does not show saliently stronger discriminatory power compared to the baseline Bayesian approach. 
 
\subsection{Characteristics of replaced doctors}
Our paper 
identifies a subset of the doctors to be replaced by 
the machine algorithms according to FPR/TPR indicators.
We create a dummy variable ``replaced'' that is set to 1 when a doctor is replaced by the algorithm and is set to 0 otherwise, and examine the partial correlation coefficient between the ``replaced'' dummy variable with several confounding factors including whether they practice in a clinic or hospital.

We regress the ``replaced'' indicator on   three confounding factors. 
The first factor is county-level per capita GDP (in ten thousand Chinese yuans)  
collected manually from 
the statistical yearbooks of 2014 across cities and provinces,
which is the starting year of our data\footnote{The statistical yearbook of  year $A$ records 
the latest available data before year $A$.}. 
The second factor is the logarithm of the total number of patient cases diagnosed by each doctor. 
The third factor is represented by two dummy variables related to the practice history of each doctor. 
A doctor in our dataset may practice in more than one type of facility. The practice facility can be either a county hospital or a
township clinic. 
Medical facilities in China are typically classified by the administrative level that they belong to. On the one hand, urban residents usually go to a county level hospital. 
On the other hand, rural residents often times go to a township level clinic, which is generally considered to be of lower quality than a county level hospital. 
The doctors in our data fall into three categories. 
Doctors who have practiced in both types of facilities 
are used as the reference level. 
The first dummy variable representing the third factor is set to $1$ for doctors who have practiced only in a county hospital. 
The second dummy variable representing the third factor is set to $1$ for doctors who 
have practiced only in a township clinic.
There are additional variables in the dataset that are only available at a more aggregated 
the city or province level, such as birth rate and the 
average number of medical staff. We do not use these variables because, given how administrative units are divided in China, a municipal unit 
can include many county level units with heterogenuous economic conditions. 

The data summary statistics are reported in Panel A of 
Table \ref{tab:summary_statistics}. As shown in Panel A of Table \ref{tab:summary_statistics}, doctors who have practiced at both a county hospital and a township clinic
have seen the largest number of patient cases and have a high replacement ratio.
Doctors who have practiced only at a township clinic have seen a smaller number of patient cases compared to doctors
who have practiced only at a county hospital, but 
doctors who have practiced only at a township clinic  have a higher replacement ratio than 
doctors who have practiced only at a county hospital have.

\begin{table}
\centering
\caption{Characteristics of replaced doctors}
\begin{tabular}{cccccc}
  \toprule
  & & \multicolumn{4}{c}{Panel A: Summary statistics by doctor type} \\
    \cline{3-6} 
  & & num & log num & GDP & replacement ratio \\
  \midrule
  county \& town & & 276 & 6.797 & 3.956 &  0.522 \\
         & & & (0.73) & (2.94)  &  \\ 
  county only & & 161 & 6.515 & 3.675 & 0.398 \\
         & & & (0.65) & (2.27) & \\
  town only & & 43 & 6.409 &  3.038  & 0.512 \\
        & & & (0.58) & (2.26) & \\
  \midrule
  & \multicolumn{2}{c}{Panel B: Logistic regression} & & \multicolumn{2}{c}{Panel C: TPR kernel regression} \\
  \cline{2-3} \cline{5-6}
  & \quad\quad (1) & (2) & & \quad\quad\quad\quad (3) & (4) \\
  log num & \quad\quad 0.896 & 0.885 & & \quad\quad\quad\quad 0.000251 & 0.000117 \\
  & \quad\quad (2.85) & (3.20) & & \quad\quad\quad\quad (0.04) & (0.02) \\
  GDP percap & \quad\quad -0.0114 & -0.0124 & & & 0.000569 \\
  & \quad\quad (-0.37) & (-0.41) & & & (0.28) \\
  count only & & -0.367 & & \quad\quad\quad\quad 0.0172 & 0.0181 \\
  & & (-1.99) & & \quad\quad\quad\quad (2.04) & (1.99) \\ 
  town only & & 0.359 & & \quad\quad\quad\quad -0.0273 & -0.0235 \\
  & & (0.86) & & \quad\quad\quad\quad (-1.73) & (-1.63) \\
  FPR & & & & \quad\quad\quad\quad 0.958 & 0.954 \\
  & & & & \quad\quad\quad\quad (44.02) & (51.70) \\
  Observations & \quad\quad 457 & 457 & & \quad\quad\quad\quad 480 & 469 \\
  \bottomrule 
\end{tabular}
    \label{tab:summary_statistics}
    \fnote{
\textit{Notes}: 
Panel A 
shows the summary statistics of doctors grouped by type of their facility. 
Panel B 
shows regression results of a logistic regression with province level fixed effect. 
Cluster robust \textit{t}-statistics are reported in parentheses. 
A more direct analysis of doctors performance is provided by a kernel regression of TPR on FPR and other covariates reported in Panel C.
As in Panel B \textit{t}-statistics are given in parentheses. 
}
\end{table}

Panel B in 
Table \ref{tab:summary_statistics}
reports a logistic regression 
of the replacement indicator on county level per  capita GDP and the logarithm of the number of patients seen by 
the doctor. Column (2) adds an indicator of whether the doctor practices only in a township clinic, and an indicator 
of whether the doctor practices only in a county hospital. The reference level where both 
indicators are zero corresponds to doctors who have practiced in both a township clinic and a county hospital. 
Province fixed effects are included in the logistic regression. 
We use cluster 
robust standard errors in order to correct for spatial autocorrelation.
The logistic coefficient estimates do not contradict our a priori reasoning. While higher 
per capita GDP is associated with a 
lower likelihood of replacement, the corresponding coefficients are small 
and lack statistical significance.
The coefficients of the log number of patients are positive and statistically significant, 
indicating the higher distinguishing power of the Bayesian test when a doctor sees more patients. 
The FPR/TPR pairs of doctors with more patient cases can be estimated more precisely. Consequently, it is more likely 
for these doctors whose FPR/TPR pairs are below the machine ROC curve to be identified as replaced by our statistical procedure.
This finding is consistent with the lack of evidence contradicting Assumption 
\ref{ROC assump}
that 
doctors accumulate experience from additional patient cases. 

Doctors practicing only
at township clinics have higher replacement probability than reference level doctors practicing at both 
township clinics
and county hospitals, who in turn have higher replacement probability than doctors 
practicing only at county hospitals. 
Anecdotal evidence in China 
suggests that doctors with better medical school records are more likely to be assigned directly to 
county level hospitals than to have to work their ways up from township  clinics into county hospitals. 

To distinguish an analysis of the predictors of performance of doctors from 
the predictors of tests concluding whether they 
are to be replaced by the machine algorithm, 
we implement a kernel regression of TPR 
on FPR and a number of confounding features  
including county level per  capita GDP, the logarithm of the number of patients seen by 
the doctor, an indicator of whether the doctor practices only in a township clinic, and an indicator 
of whether the doctor practices only in a county hospital.\footnote{We thank the associate editor and the review team
for suggesting a direct analysis of the performance of doctors.}
Among doctors achieving the same FPR, 
those achieving a higher TPR 
can be considered to perform better.
A confounding variable with a positive coefficient enhances the TPR 
after controlling for FPR, and is interpreted as 
a variable that positively predicts performance.

Panel C in Table \ref{tab:summary_statistics} 
reports the results from this regression. The regression coefficients have signs that 
are consistent with a priori reasoning, but they are also mostly insignificant except for the coefficients on the 
indicator of whether the doctor practices only in a county hospital. For a given FPR, 
a higher TPR 
is positively associated with higher
county level per capita GDP, a larger number of patients seen by the doctor and
the indicator that the doctor practices only in a county hospital, 
but is negatively associated with the  indicator that the doctor practices only in a township clinic.
The statistical insignificance
of the coefficients on the logarithm of the number of treated patients in Panel C of Table 
\ref{tab:summary_statistics} 
and their significance in Panel B of Table \ref{tab:summary_statistics} 
provide 
further evidence that the replacement indicator is largely driven by the difference in 
the sample size.  A larger sample size makes it easier to reject 
the null hypothesis that the machine algorithm does not outperform the doctor.

\section{Synthetic Data Analysis}
\label{synthetic data analysis}

\subsection{An example of human-algorithm complementarity}

In our empirical dataset, only a small fraction of doctors lie above the machine ROC curve. 
Combining 
the machine algorithm and the doctor decisions results in aggregate FPR/TPR pairs that are still below the machine ROC curve. 
In addition, Jensen's inequality (elaborated in section \ref{Discussion: incentives and costs}) 
likely 
contributes to the combined aggregate FPR/TPR pair lying below the machine ROC curve. Even if we replace every doctor by the machine algorithm, 
unless the same point on the machine ROC curve is used to replace all doctors, the combined aggregate FPR/TPR pair is
necessarily below the machine ROC curve.
The end points of the replacement paths in Figure \ref{fg:ai-doctors-bayesian-score-300} correspond to replacing all doctors by the machine algorithm. 
The algorithms that we implemented use different points on the machine ROC curve to replace doctors. These points also differ
across different doctors.
In order for the combined aggregate FPR/TPR pair to be above the machine ROC curve, it is necessary that a sizable portion of the original doctor diagnoses lie above the machine ROC curve.

While not the case in the NFPC dataset,
other datasets might exist where a significant portion of human decision makers correctly use private information not available to the machine algorithm. Figure \ref{fg:example_bayesian_good_combination} reports the results from a synthetic data analysis where the combined aggregate FPR/TPR pair between doctors and the machine algorithm lies above the machine ROC curve.

\begin{figure}
  \centering
  \includegraphics[width=0.8\columnwidth]{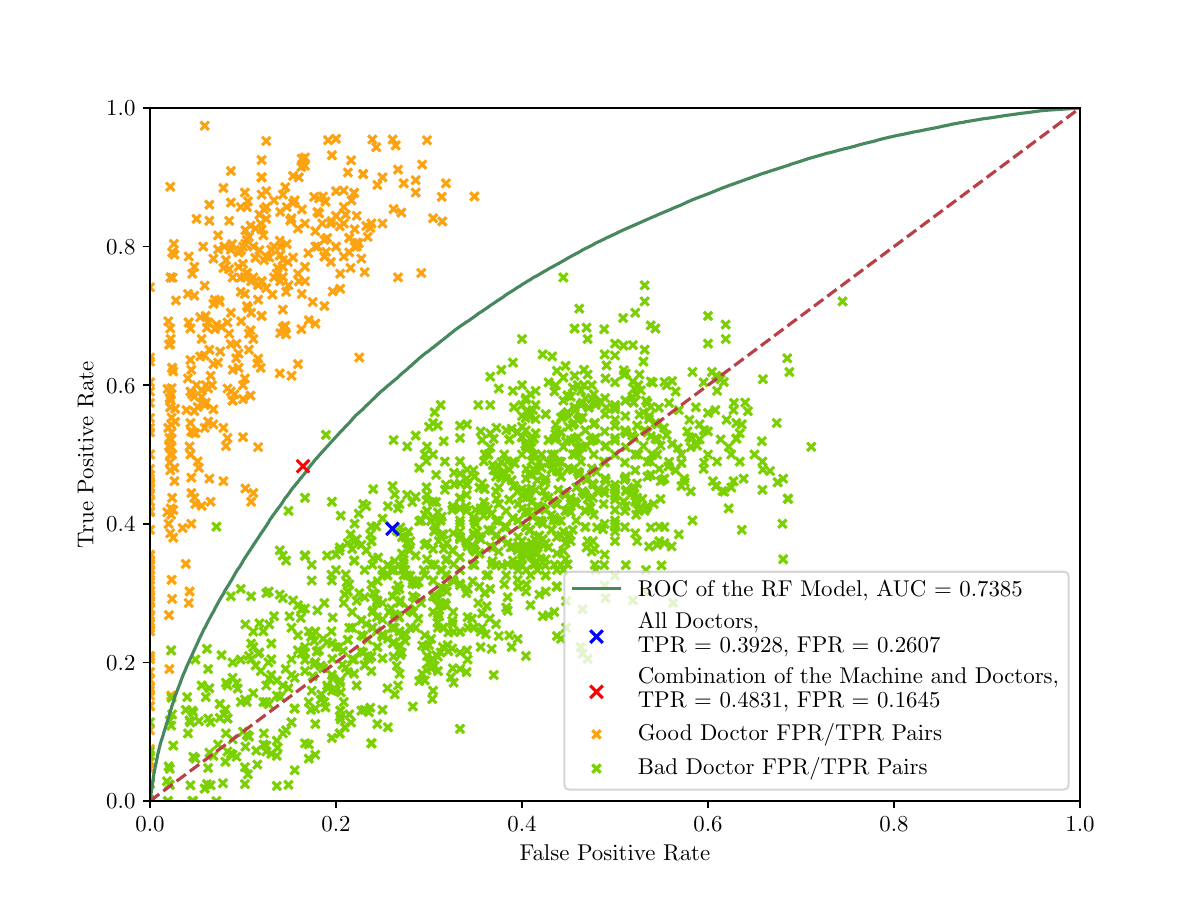}
  \caption{Synthetic data where combined FPR/TPR is above the machine ROC curve}
  \fnote{
  \textit{Notes}: In the figure, we simulate 750 doctors with better information than the machine algorithm and full information processing capabilities. 
  These doctors are represented as the orange FPR/TPR points. 
  Meanwhile, we simulate 1250 doctors with the same information but use misspecified prediction model in their decisions. 
  These doctors are plotted as the green FPR/TPR points. 
  Then we apply our Bayesian test approach to identify and replace the less capable doctors in the simulation. 
  As shown in the figure, the doctor-algorithm combined aggregate FPR/TPR point lies above the machine ROC curve, an example of complementarity. 
  }
  \label{fg:example_bayesian_good_combination}
\end{figure}

The data generating process of 
the simulation in
Figure \ref{fg:example_bayesian_good_combination} is as follows. The input features are three dimensional: $x_{i,j,1}$, $x_{i,j,2}$ and $u_{i,j}$. On the one hand, the first two features, denoted $x_{i,j,1}$ and $x_{i,j,2}$, are assumed to be observable both by the doctors and the machine algorithm. On the other hand, the last feature, denoted $u_{i,j}$, is only observable to the doctors and is not used by the machine algorithm. The three features are generated from independent normal distributions with pre-specified means and variances: $X_{i,j,1}\sim N\(0,1.95\)$, $X_{i,j,2} \sim N\(0,0.25\)$ and $U_{i,j} \sim N\(0,2\)$.
The ground truth variable $Y_{i,j} \in\{0,1\}$ is generated by 
\bs
Y_{i,j} = \mathds{1}\(
p\(X_{i,j,1}, X_{i,j,2}, U_{i,j}\) > \delta_{i,j}
\)\quad\text{where}\quad p\(X_{i,j,1}, X_{i,j,2}, U_{i,j}\) = \frac{e^{X_{i,j,1}+X_{i,j,2} + U_{i,j}}}{1+e^{X_{i,j,1}+X_{i,j,2}+U_{i,j}}}
\end{split}\end{align}
and where $\delta_{i,j}$ 
is uniformly distributed between $0$ and $1$. Using this specification of the data generating process, we simulated 600,000 patient cases and allocate them evenly to 2000 doctors. Each doctor is assigned 300 simulated patient cases.

Doctors are partitioned into two groups according to their information processing capacities. The first group, consisting of 750 {\it more capable} doctors, use the correct full information propensity score $p\(x_{i,j,1}, x_{i,j,2}, u_{i,j}\)$ with different incentives to diagnose patient cases. Their decision rule is given by, using $C_j$ 
that is uniformly distributed between $0.4$ and $1$ across doctors,
\bs
\hat Y_{i,j} = \mathds{1}\(p\(X_{i,j,1}, X_{i,j,2}, U_{i,j}\) > C_j\).
\end{split}\end{align}
The second group of doctors, consisting of 1250 {\it less capable} doctors, use a misspecified 
propensity score to diagnose patients. 
Their decision rule is
\bs
\hat Y_{i,j} = \mathds{1}\(q\(X_{i,j,1}, X_{i,j,2}, U_{i,j}\) > C_j\) \quad\text{where}\quad
q\(X_{i,j,1}, X_{i,j,2}, U_{i,j}\) = \frac{e^{-X_{i,j,1}+X_{i,j,2} + U_{i,j}}}{1+e^{-X_{i,j,1}+X_{i,j,2} + U_{i,j}}}
\end{split}\end{align}
is a misspecified 
propensity score function. 

We divide the 600,000 simulated observations 
into a training subset consisting of 240,000 observations, a validation subset consisting of 180,000 observations and 
a test subset consisting of the remaining 180,000 observations. We estimate a random forest model using only the training subset 
and the first two features $x_{i,j,1}$ and $x_{i,j,2}$, and use the baseline Bayesian test 
in equation \eqref{bayesian basic numerical} to combine doctors with the machine algorithm. The red point in Figure \ref{fg:example_bayesian_good_combination} represents the combination of doctors and the machine algorithm. It lies above the ROC curve of the random forest model, and also strictly dominates the blue point of the aggregate FPR/TPR pair of all doctors.
In summary, the synthetic data analysis represented by Figure \ref{fg:example_bayesian_good_combination} illustrates 
the scientific principle of complementarity in human algorithm collaborations. See for example \cite{bansal2021does}.\footnote{
We are grateful to an anonymous referee for pointing us to the science literature on human-algorithm collaboration complementary.}

\subsection{Aggregating information by predicted doctor models}

Our dataset contains doctor diagnoses information in addition to the eventual pregnancy outcomes. 
In this section we discuss whether estimating a model that predicts the label of doctors' diagnoses
can assist in improving the FPR/TPR tradeoff of medical decisions. 
A model of predicted doctors can be intuitively interpreted as doctors' collective wisdom. 
The experience of senior doctors can provide guidance to junior doctors with limited clinical exposures.

In the NFPC dataset, a model of predicted doctors does not appear to improve on the aggregate FPR/TPR pair. 
However, in several of the following synthetic data examples, a predicted doctor model 
improves upon 
the diagnoses of individual doctors. 
Each doctors' diagnosis may encode information heterogeneity, or incentive heterogeneity, or both. 
Consider a model of incentive heterogeneity, where doctors' decision rules are given by $\hat y_{i} = \mathds{1}\(p\(x_{i}\) > u_{i}\)$ such that $p\(x_{i}\)$ is the correctly specified propensity score, and $u_i$ is 
independently distribution with a strictly monotonic CDF $G\(\cdot\)$. 
On the one hand, the aggregate doctor FPR/TPR lies below the optimal ROC curve. On the other hand,
the predicted doctor model, given by $\mathbb{E}\[\mathds{1}\(p\(X_i\) > U_i\) \vert X_i\] = \mathbb{P}\(U_i \leq p\(X_i\)\) = 
G\(p\(X_i\)\)$, is a monotonic transformation of the 
true propensity score function and
recovers the optimal ROC curve. 

We next focus on several information model without incentive heterogeneity. 
Consider a baseline example where doctors have access to private information $u_i$ in addition to the publicly observed features $x_i$ and are aware of the correct full information propensity score model $p\(x_i, u_i\)$. 
If we could directly observe doctors' probabilistic assessment 
$p\(x_i, u_i\)$ and use it as the label to train a machine learning model, 
then the resulting 
model will necessarily coincide with the ground-truth machine learned propensity score model:
\bs
\mathbb{E}\[ p\(X_i, U_i\) \vert X_i\] = \mathbb{E}\[\mathbb{E}\[Y_i \vert X_i, U_i\] \vert X_i\] = \mathbb{E}\[Y_i \vert X_i\] = p\(X_i\).
\end{split}\end{align}
However, in reality we rarely solicit doctors' probabilistic assessment of the likelihood of an illness. Instead, we only observe the binary diagnosis 
by the doctors, which corresponds to observing $\hat y_i=\mathds{1}\(p\(x_i, u_i\) > c_0\)$ for some threshold value $c_0$. When we use $\hat y_i$ as the label to train a machine learning model, the propensity score of the resulting predicted doctor model is
\bs
q\(X_i\) = \mathbb{E}\[\mathds{1}\(p\(X_i, U_i\) > c_0\) \vert X_i\].
\end{split}\end{align}
In general, $q\(X_i\)$ is different from $p\(X_i\)$. The ROC curve implied by $q\(X_i\)$ coincides with $p\(X_i\)$ only when $q\(X_i\)$ is an increasing transformation of $p\(X_i\)$: $q\(X_i\)=\Lambda\(p\(X_i\)\)$, where $\Lambda\(\cdot\)$ is a strictly increasing function. 
There is no guarantee that this is necessarily the case. 

Consider now three scenarios in which doctors employ misspecified models to process private information.
In the first scenario, doctors process the full information $\(X_i, U_i\)$ using an incorrect propensity score model $q\(X_i, U_i\)$, but the ROC curve implied by the predicted doctor model $q\(X_i\)$ 
still coincides with the ground-truth machine ROC curve. 
Figure \ref{fg:doctor private information incorrect model 1} illustrates this scenario using the following data generating process. The ground-truth propensity score model is 
$p\(X_{i}\)=\exp\(X_{i}\) / \(1+\exp\(X_{i}\)\)$. The doctors use a misspecified propensity model of $q\(X_{i},U_{i}\)=\exp\(X_{i}+U_{i}\) / \(1+\exp\(X_{i}+U_{i}\)\)$. 
Both $X_i$ and $U_i$ are uniformly and independently distributed: $X_{i} \sim U\(-1, 1\)$ and $U_{i} \sim U\(-2, 2\)$. 
The ``$X_i$ only'' ROC curve is generated by $p\(X_{i}\)$; the ``$X_i$ and $U_i$'' ROC curve is generated by $q\(X_{i},U_{i}\)$. The ``Predicted Doctors on $X$'' ROC curve is generated by 
$q\(X_i\) \equiv \mathbb{E} \[\mathds{1}\(q\(X_i, U_i\) > c_0\) \vert X_{i}\].$

\begin{figure}
  \centering 
  \subfigure[Doctor using incorrect model based on private information: scenario 1]
  {\includegraphics[width=0.48\columnwidth]{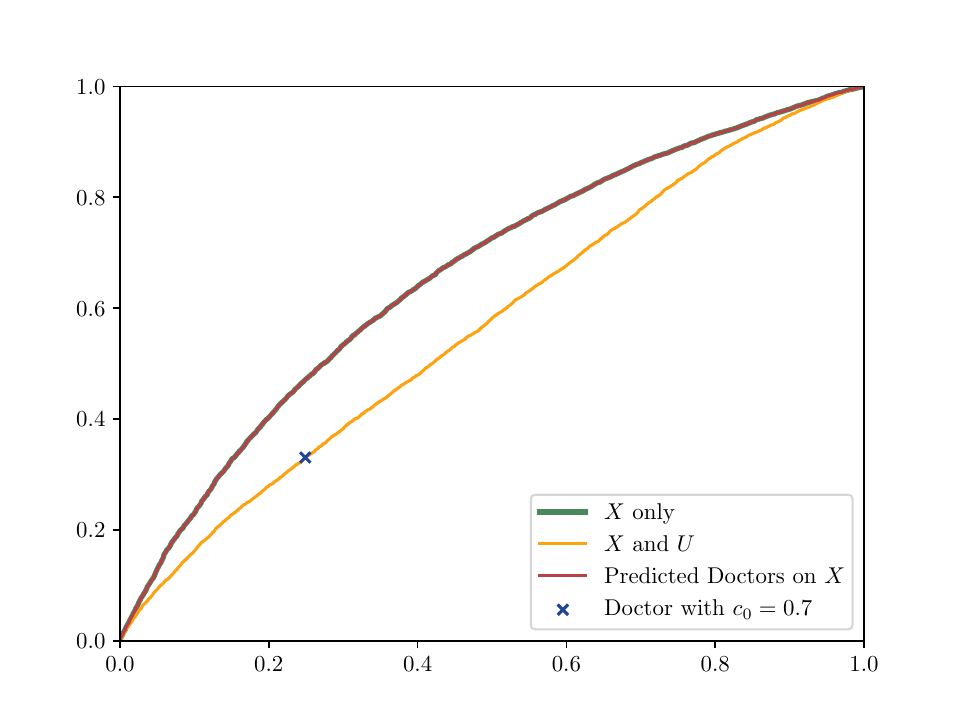}
  \label{fg:doctor private information incorrect model 1}}
  \subfigure[Doctor using incorrect model based on private information: scenario 2] {\includegraphics[width=0.48\columnwidth]{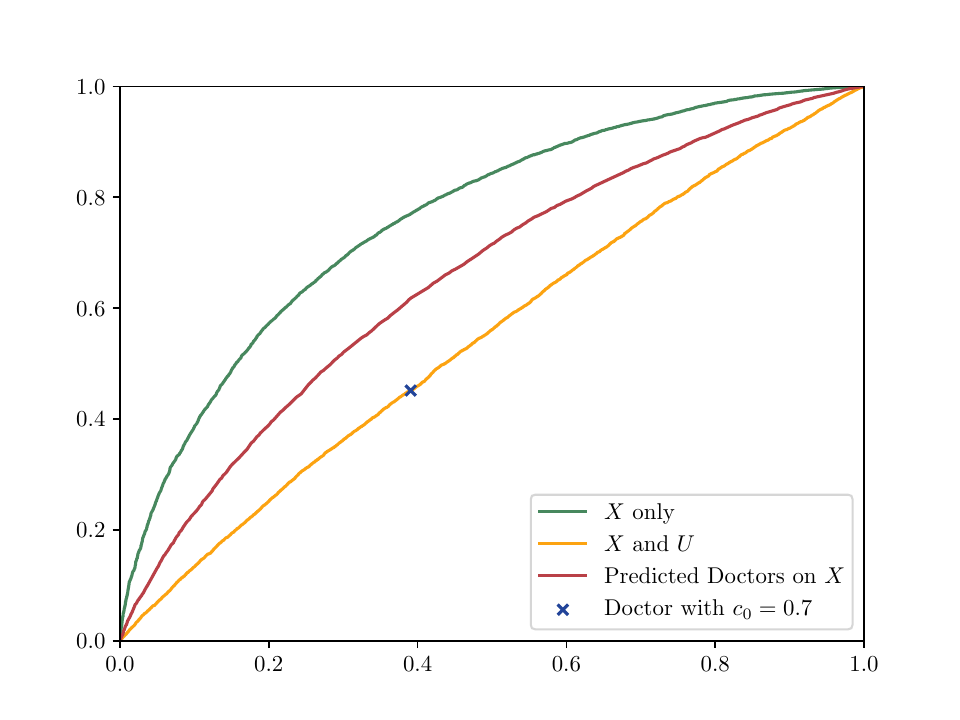}
  \label{fg:doctor private information incorrect model 2}}
  \subfigure[Doctor using incorrect model based on private information: scenario 3] {\includegraphics[width=0.48\columnwidth
  ]{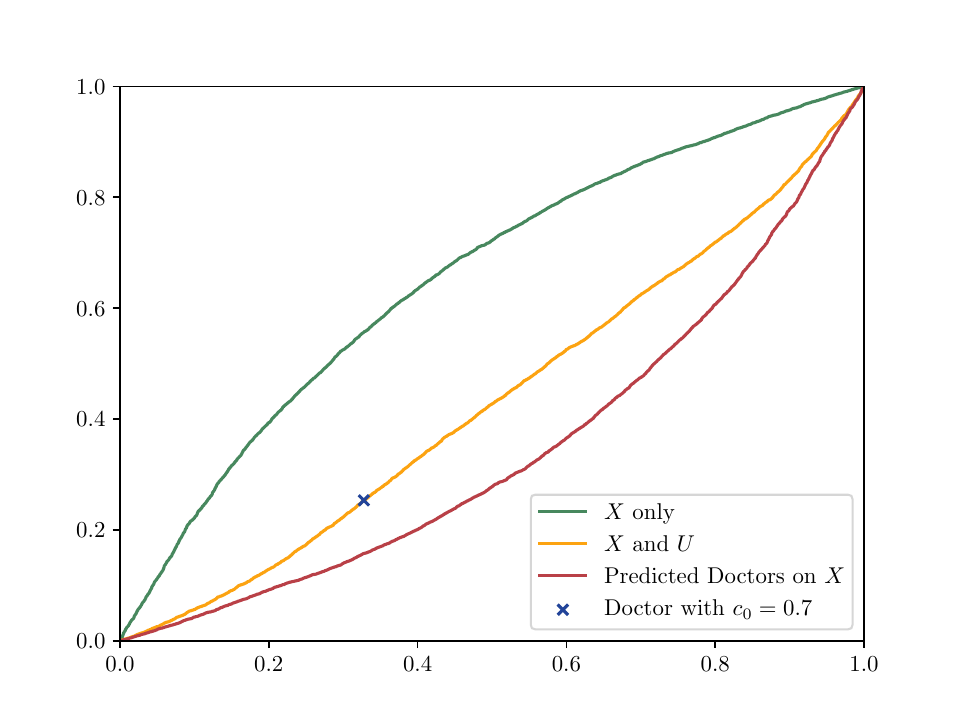}
  \label{fg:doctor private information incorrect model 3}}
  \caption{Three cases of learning from doctors}
  \fnote{\textit{Notes}: 
  Instead of predicting 
ground-truth outcome, this figure illustrates three synthetic data examples where machine algorithms 
are applied to predict doctor's diagnoses. 
  Panel \ref{fg:doctor private information incorrect model 1} shows a scenario where the machine 
algorithm improves upon 
the individual doctor diagnoses and generates the same ROC curve 
as a correctly specified model. 
  Panel \ref{fg:doctor private information incorrect model 2} shows a scenario where aggregating information from doctor diagnoses is beneficial but not sufficient to recover the optimal ROC curve. 
  The last case is shown in Panel \ref{fg:doctor private information incorrect model 3} where aggregating information from doctor diagnoses results in a ROC curve inferior to that of
the misspecified prediction model. 
  }
  \label{aggregating information}
\end{figure}

For any given 
cutoff threshold $c_0$, 
the predicted doctor model $q\(X_i\)$ 
is an increasing function of $X_i$, as is $p\(X_i\)$.
In this particular case, both $p\(X_i\)$ and $q\(X_i\)$ correspond to the same ROC curve. 
This curve lies above the ROC curve generated by $q\(X_i, U_i\)$ and the aggregate doctor FPR/TPR pair ($c_0=0.7$).
Collective wisdom improves upon the misspecified model employed by individual doctors and results in higher diagnosing quality as represented by the ROC curves. 

The second scenario, pictured in Figure \ref{fg:doctor private information incorrect model 2}, is generated by a two dimensional observable feature model. The ground-truth propensity score model 
is 
\bs
p\(X_{i1},X_{i2}\)=\exp\(X_{i1}+X_{i2}\)/\(1+\exp\(X_{i1}+X_{i2}\)\). 
\end{split}\end{align}
The doctors' (misspecified) information model is 
\bs
q\(X_{i1},X_{i2},U_{i}\)=\exp\(X_{i1}-X_{i2}+U_{i}\)/\(1+\exp\(X_{i1}-X_{i2}+U_{i}\)\).
\end{split}\end{align}
$X_{i1}, X_{i2}$ and $U_i$ are independently distributed: 
$X_{i1} \sim N\(0, 1\), X_{i2} \sim N\(0, 0.5\), U_{i} \sim N\(0,4\)$.

In this scenario, doctors also process the full information $\(X_{i1}, X_{i2}, U_i\)$ using an incorrect propensity score model. 
While the ROC curve corresponding to the predicted doctor model 
$q\(X_{i1}, X_{i2}\)\equiv \mathbb{E} \[\mathds{1}\(q\(X_{i1}, X_{i2}, U_i\) > c_0\) \vert X_{i1}, X_{i2}\]$ lies below the optimal ROC curve corresponding to the correct propensity score
model $p\(X_{i1},X_{i2}\)$, it still lies above the ROC curve corresponding 
to the individual doctor misspecified information model $q\(X_{i1},X_{i2},U_{i}\)$ and the aggregate doctor FPR/TPR pair. 
Collective wisdom in this case leads to limited improvement over the misspecified information model.

The third scenario is pictured in Figure \ref{fg:doctor private information incorrect model 3}. 
The ground-truth propensity score model 
is $p\(X_{i}\)=\exp\(X_{i}\)/\(1+\exp\(X_{i}\)\)$ but 
the doctors' incorrect information model is 
\bs
q\(X_{i},U_{i}\)=\exp\(-X_{i}+U_{i}\)/\(1+\exp\(-X_{i}+U_{i}\)\), 
\end{split}\end{align}
where $X_{i}$ and $U_{i}$ are distributed as in the first scenario depicted in 
Figure \ref{fg:doctor private information incorrect model 1}. 
In this model, doctors' information processing capacity is very limited. 
The predicted doctor model 
$q\(X_i\) = \mathbb{E}\[\mathds{1}\(p\(X_i, U_i\) > c_0\) \vert X_i\]$
exacerbates the misspecification of the doctors' private information propensity score model, and 
leads to a ROC curve that lies strictly below the ROC curve corresponding 
to the individual doctor misspecified propensity score model $q\(X_{i},U_{i}\)$.

The analysis in Figure \ref{aggregating information} shows that depending on the extent of misspecification of the doctors' information model, the predicted doctor model that aggregates the information among doctors using the subset of observable features available to the machine learning algorithm might either reduce or exacerbate the misspecification in the propensity score model. The resulting ROC curve of the predicted doctor model  might lie above or below
the ROC curve of the doctors' information model. 

\subsection{Replacement strategies initiated by doctors} 

For a given decision maker, recognizing
that machine
algorithms are more likely to be applicable in some candidate cases than 
in others is important for a replacement model to achieve human and machine
complementarity.
Formalizing such a general empirical replacement model 
is difficult because of 
the unknown diagnosing model employed by doctors in the data and  because
of our lack of knowledge of doctors' preferences for 
the machine algorithm. 
To illustrate these difficulties,
we experiment with a set of simulations in which doctors decide when 
to employ machine
learning algorithm. 

In the simulation setup, each patient is associated with three  independently 
generated normally distributed features $X_1, X_2, U$, where
$X_1 \sim N(0, 0.5)$, $X_2 \sim N(0, 1.5)$ and $U \sim N(0, 1)$. 
The probability of each patient being ill is a  logit 
function of the features: 
\begin{align}\begin{split}\nonumber
p\(X_1, X_2, U\) = e^{X_1 + X_2+U} / \(1 + e^{X_1 + X_2+U}\).
\end{split}\end{align}
The labels are generated according to the probability model 
$p\(X_1, X_2, U\)$ 
for 30,000 patient cases. 
The diagnosis of doctors is modeled   as a mixture between two decision rules. 
The first rule
employs the correct 
probability model $p\(X_1, X_2, U\)$. The
second rule replaces the    probability model $p\(x\)$ by an uninformative 
uniformly distributed random variable $r \sim U\(0,1\)$.
The second rule applies to a situation where the occasional lack of 
medical knowledge by doctors results in randomized diagnosis. 
Specifically, we simulate doctor diagnosis $\hat Y$ by 
\begin{align}\begin{split}\nonumber
\hat Y = 
\begin{cases}
\mathds{1}\(
e^{X_1 + X_2+U} / \(1 + e^{X_1 + X_2+U}\) > c_d
\), & \text{if}\ X_1 \geq 0,\\
r > c_d,  & \text{otherwise}.
\end{cases}
\end{split}\end{align}
In the above, $c_d = 0.5$ is the decision threshold for doctor diagnosis.
The simulated dataset is divided into
training, validation and test subsets by the ratio of 4:3:3. 
They are used to estimate model parameters, compute the ROC curve 
and calculate FPR/TPR pairs respectively.
We 
first train
a random forest 
model using features $X_1$ and $X_2$, where $U$ is interpreted as unobserved heterogeneity 
that is known to the doctors but not to the machine algorithm. 

\begin{figure}
  \centering
  \subfigure[Doctors well informed about themselves]{\includegraphics[width=0.45\columnwidth]{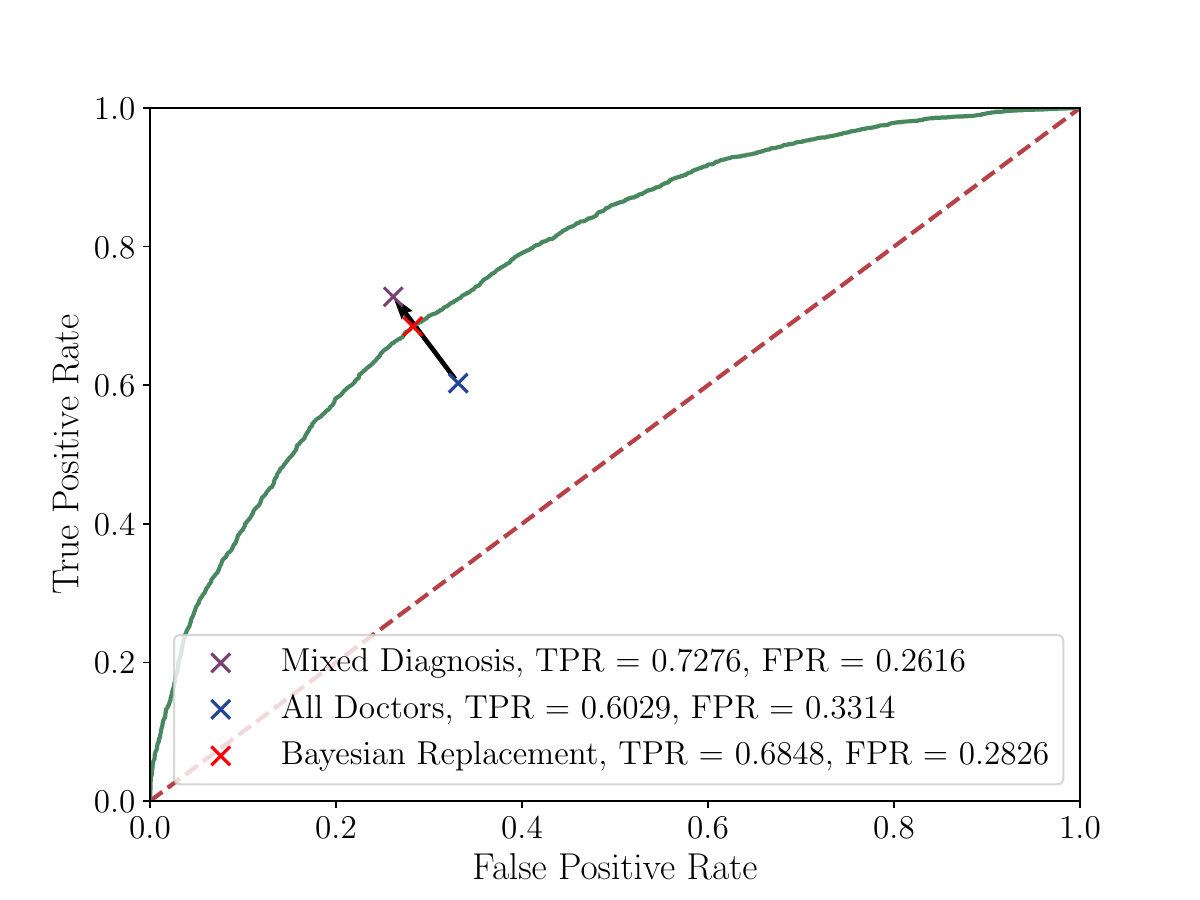}\label{fg:append5 fine-grained simulation 1}}
  \subfigure[Doctors informed about themselves]{\includegraphics[width=0.45\columnwidth]{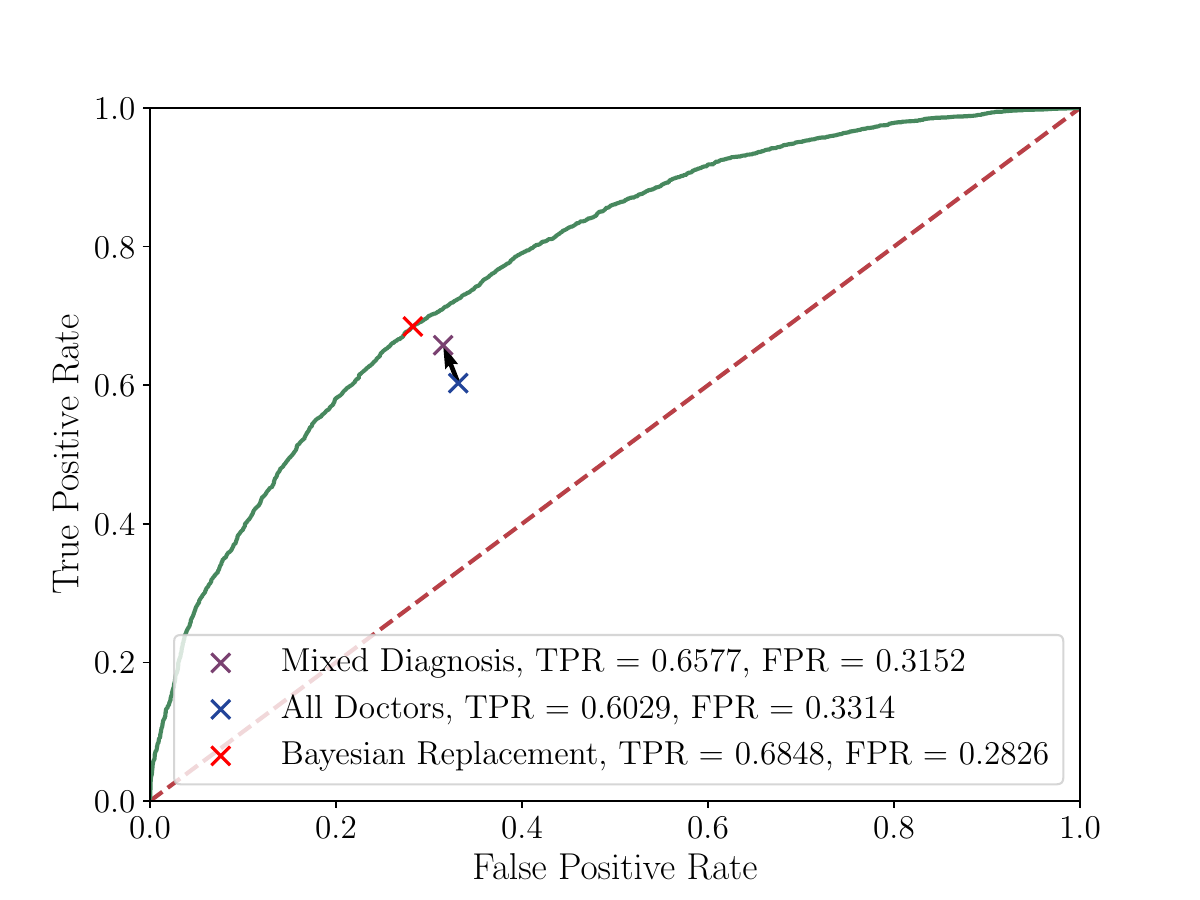}\label{fg:append5 fine-grained simulation 2}}
  \subfigure[Doctors not informed about themselves]{\includegraphics[width=0.45\columnwidth]{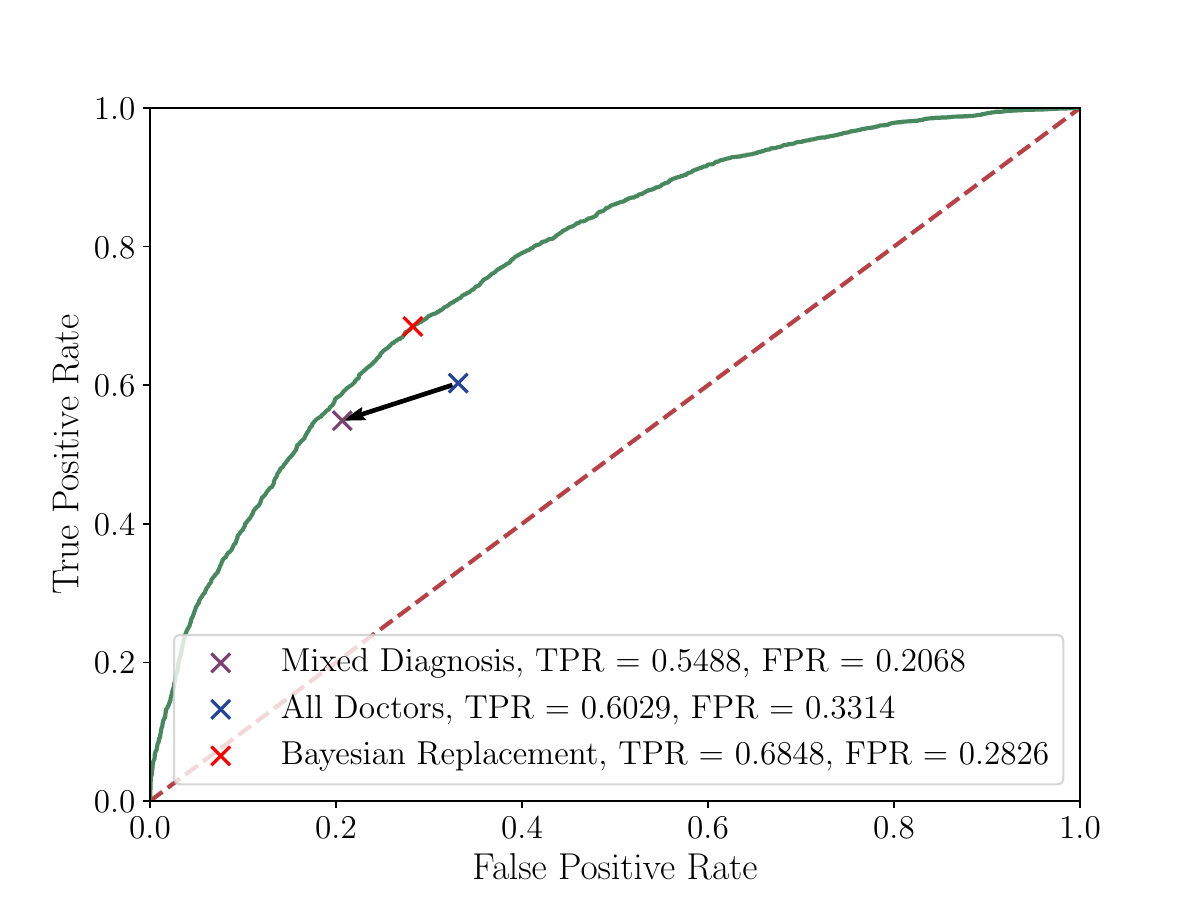}\label{fg:append5 fine-grained simulation 3}}
  \subfigure[Doctors misinformed about themselves]{\includegraphics[width=0.45\columnwidth]{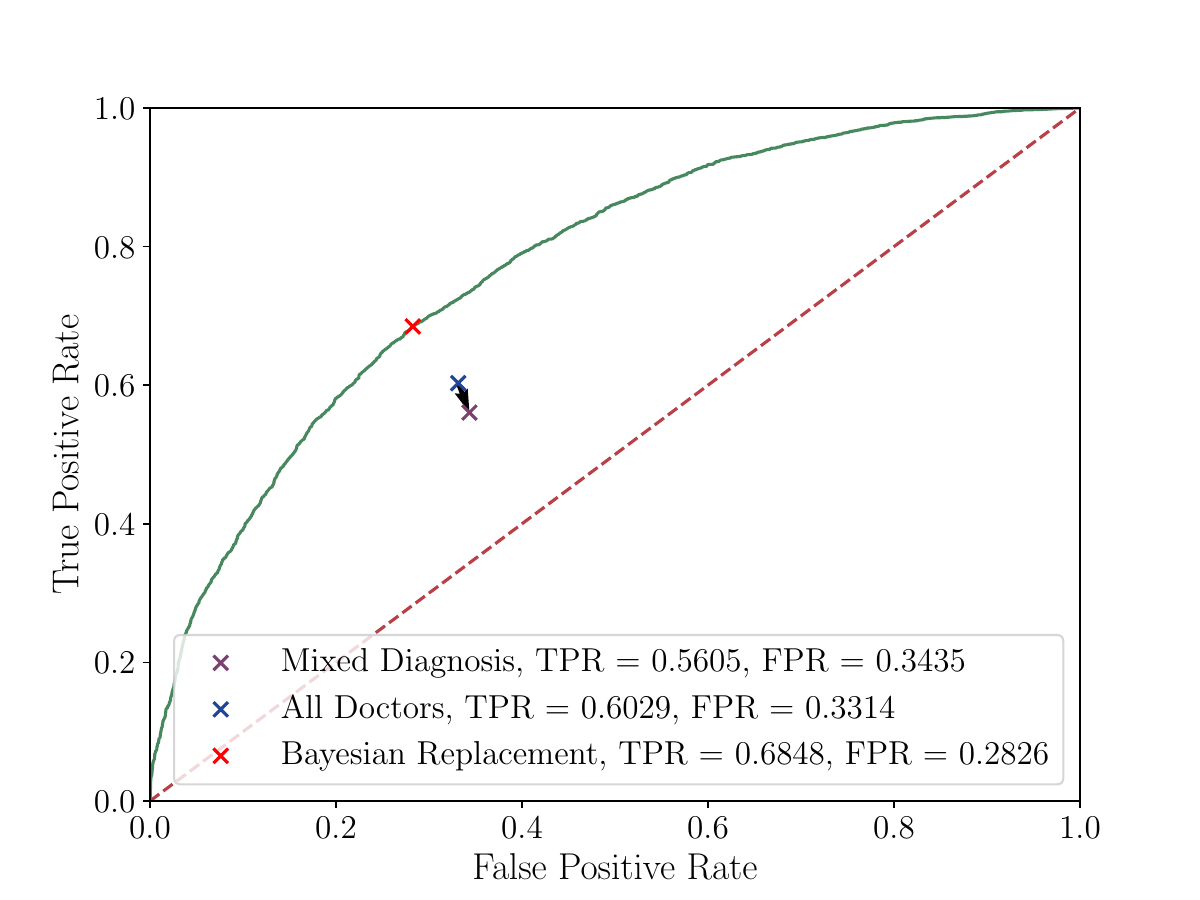}\label{fg:append5 fine-grained simulation 4}}
  \caption{Replacement result under different replacement rules}
  \label{fg:append5 fine-grained simulation}
  \fnote{
  \textit{Notes}:
  A flexible replacing model allowing doctors to choose whether to apply the algorithm's decision has the potential of achieving human-machine complementarity. 
  While, in this figure, we show that the performance of this strategy highly depends on doctors' choices. 
  If the doctors do not have a clear recognition of their information weakness, this strategy may not perform well. 
  }
\end{figure}

We simulate four replacement rules. They are illustrated in Figure \ref{fg:append5 fine-grained simulation}. In the first rule, shown in 
 Figure~\ref{fg:append5 fine-grained simulation 1}, doctors know when their 
medical knowledge is inadequate (e.g. when $X_1 < 0$) and choose to defer
to the machine algorithm in such a situation. 
This method achieves 
complementarity between the human doctor and the machine algorithm. 
It produces a FPR/TPR pair that is better than the Bayesian approach 
in wich  
either all or none of the patients for each doctor are replaced.
In the second rule, shown in 
 Figure~\ref{fg:append5 fine-grained simulation 2}, doctors have partial but 
incomplete knowledge of when their medical knowledge is inadequate. 
Doctors defer to the machine algorithm when $X_1 < 0.5$. 
The second rule improves over the doctors, but 
does not outperform the Bayesian approach.
Figure~\ref{fg:append5 fine-grained simulation 3} 
use a replacement rule of $X_2 \le -0.5$. In this case, the doctor does not know when they lack medical knowledge and the replacement reduces both FPR and TPR.
In Figure~\ref{fg:append5 fine-grained simulation 4}, 
doctors mistakenly defer to the machine algorithm when they do have medical knowledge (e.g. when $X_1 \geq  0$), but make their own decision when they 
are uninformed (e.g. when $X_1 < 0$).  The resulting FPR/TPR pair is 
strictly dominated by the original doctor diagnosis.

\section{Conclusion}\label{conclusion}

In this paper, we propose statistical tests for replacing human decision makers
with algorithms.
Decision rules that robustly dominate
diagnoses made by humans
can be determined based on the machine ROC curves.
We propose both a heuristic frequentist approach and an alternative Bayesian posterior loss function approaches. 
By replacing a subset of decision makers
with algorithms,
we can achieve higher quality
decision making outcome.
We experiment with the two approaches 
using a dataset of pre-pregnancy checkups, and find that 
the replacement of
a subset of doctors using the Bayesian approach significantly
improves the overall performance of
risky pregnancy detection.
Furthermore, our results indicate that
doctors who practice only at township clinics are more likely to be replaced, while
doctors who practice only at county hospitals are less likely to be replaced. 

Our analysis shows that both stable preference and decision-making quality are 
important for comparing machine algorithm with human decision makers. 
While allowing the juniors to make potentially suboptimal decisions for their patients and 
learn from the results has the potential of increasing the pool of high quality seniors, 
this is a very delicate issue involving complex welfare tradeoff between current and future generations of patients.
In panel datasets with rich dynamic 
information, human decision makers can engage in information updating 
and peer learning. 
Our current dataset is only cross-sectional. 
The eventual pregnancy outcomes are typically
not available to doctors when they diagnose consecutive patient cases.
Juniors might also benefit more generally from 
offline learning without having to make substantial and high stake decisions.



\section{Acknowledgments}

We thank the editors and two anonymous referees, Brendan Beare, Bo Honore, Xiaohong Chen, Ron Gallant, Bruce Hansen, Peter Hansen, Bentley MacLeod, Jack Porter, Adam Rosen, Andres Santos, George Tauchen, Valentin Verdier, and participants at various seminars and conferences for insightful comments. 
We also thank Yichuan Zhang, Xin Lin and Zhonghao Huang for excellent research assistance. 
This study was approved by the National Health Commission. Informed consents were obtained from all the NFPC participants.
We acknowledge funding support from the National Science Foundation (SES 1658950 to
Han Hong), the National Science Fund for Distinguished Young Scholars of China (71325007 to Ke Tang), the State’s Key Project of Research and Development Plan (2016YFC1000307 to Jingyuan Wang) and the National Natural Science Foundation of China (61572059 to Jingyuan Wang).

\FloatBarrier
\phantomsection
\addcontentsline{toc}{section}{References}
{
\setstretch{0.875}
\bibliography{ms_maintex225_withname}
\bibliographystyle{aer}
}

\clearpage 
\processdelayedfloats
\clearpage

\appendix
\thispagestyle{empty}
\begin{center}
\section{\Large \bf Appendix for manuscript \sc \\ Statistical tests for replacing human decision makers with algorithms}
\end{center}

\newpage



\setcounter{section}{1}


\subsection{Additional Lemmas and Proofs}


\begin{lemma}\label{better information higher roc}
Let $p\(X, U\) = \mathbb{P}\(Y=1 \vert X, U\)$ and $p\(X\) = \mathbb{P}\(Y=1 \vert X\)$ be two correctly specified propensity score functions. The ROC curve generated 
$p\(X, U\)$ lies above the ROC curve generated by $p\(X\)$. 
Specifically, assuming that $\mathbb{P}\(p\(X, U\) = c\) = 0$ for all $c$, then if the decision rules $\mathds{1}\(p\(x\) > c\)$ and 
$\mathds{1}\(p\(x, u\) > c'\)$ both achieve FPR level $\alpha$ for $0 \leq c, c' \leq 1$ and 
\bs
\mathbb{P}\left\{\mathds{1}\(p\(X\) > c\) \neq \mathds{1}\(p\(X, U\) > c'\)\right\} > 0, 
\end{split}\end{align} 
then $\mathds{1}\(p\(x, u\) > c'\)$ achieves higher TPR than $\mathds{1}\(p\(x\) > c\)$ does. 
\end{lemma}
\begin{proof}
If we denote the ROC curve generated by $p\(x\)$ as \bs
\beta_x\(\alpha\) = f_x\(\alpha\) = \frac{\mathbb{E}\[ Y \mathds{1}\(p\(X\) > c\)\]}{p} \quad\text{where $c$ satisfies}\quad \frac{\mathbb{E}\[ \(1-Y\) \mathds{1}\(p\(X\) > c\)\]}{1-p}=\alpha,
\end{split}\end{align}
then it can also be written as the solution of a constrained maximization program:
\bs
\beta_x\(\alpha\) = f_x\(\alpha\) = \max_{\phi\(\cdot\)}\frac{\mathbb{E} Y \phi\(X\)}{p}\quad\text{such that}\quad \frac{\mathbb{E}\[ \(1-Y\) \phi\(X\)\]}{1-p}=\alpha\ \text{and}\ 0\leq \phi\(x\) \leq 1.
\end{split}\end{align}
See e.g. \citeappendix{feng2022properties}.
Similarly we can denote the ROC curve generated by $p\(x,u\)$ as
\bs
&\beta_{x,u}\(\alpha\) = f_{x,u}\(\alpha\) = \max_{\phi\(\cdot,\cdot\)}\frac{\mathbb{E} Y \phi\(X,U\)}{p}\\
&\text{such that}\quad \frac{\mathbb{E}\[ \(1-Y\) \phi\(X,U\)\]}{1-p}=\alpha\ \text{and}\ 0\leq \phi\(x,u\) \leq 1.
\end{split}\end{align}
Since the class of functions $\phi\(x,u\)$ includes the class $\phi\(x,u\)\equiv \phi\(x\)$ as a subset, it follows immediately that $\beta_{x,u}\(\alpha\) \geq \beta_x\(\alpha\)$. 
For the second part of our statement, consider the sets 
\begin{align}\begin{split}\nonumber
S^{+} = \left\{\(x, u\): \mathds{1}\(p\(x, u\) > c'\) > 
\mathds{1}\(p\(x\) > c\)\right\}, \\ 
S^{-} = \left\{\(x, u\): \mathds{1}\(p\(x, u\) > c'\) < 
\mathds{1}\(p\(x\) > c\)\right\}.
\end{split}\end{align}
By assumption, $\mathbb{P}\(S^+ \cup S^-\) > 0$. 
If $\(x, u\) \in S^{+}$, then $p\(x, u\) > c'$; if $\(x, u\) \in S^{-}$, then $p\(x, u\) \leq c'$. 
By assumption, $\mathbb{P}\(p\(X, U\) = c'\) = 0$. Thus, 
\begin{align}\begin{split}\nonumber
\int\int \[\mathds{1}\(p\(x, u\) > c'\) - \mathds{1}\(p\(x\) > c\)\]\(p\(x, u\) - c'\)f_{X,U}\(x, u\) \mathrm{d}x\mathrm{d}u > 0.
\end{split}\end{align}
The difference in power (TPR multiplied by $p$) then satisfies 
\begin{align}\begin{split}\nonumber
\int\int \[\mathds{1}\(p\(x, u\) > c'\) - \mathds{1}\(p\(x\) > c\)\] & p\(x, u\)f_{X,U}\(x, u\) \mathrm{d}x\mathrm{d}u \\ 
> c' \int\int & \[\mathds{1}\(p\(x, u\) > c'\) - \mathds{1}\(p\(x\) > c\)\]f_{X,U}\(x, u\) \mathrm{d}x\mathrm{d}u \\ 
= c' \int\int & \[\mathds{1}\(p\(x, u\) > c'\) - \mathds{1}\(p\(x\) > c\)\]p\(x, u\)f_{X,U}\(x, u\) \mathrm{d}x\mathrm{d}u.
\end{split}\end{align}
The equality is due to the assumption that both decision rules achieve the same FPR:
\bs
\int\int & \[\mathds{1}\(p\(x, u\) > c'\) - \mathds{1}\(p\(x\) > c\)\]
\(1-p\(x,u\)\)
f_{X,U}\(x, u\) \mathrm{d}x\mathrm{d}u=0.
\end{split}\end{align}
We then have 
\begin{align}\begin{split}\nonumber
\int\int \[\mathds{1}\(p\(x, u\) > c'\) - \mathds{1}\(p\(x\) > c\)\] & p\(x, u\)f\(x, u\) \mathrm{d}x\mathrm{d}u > 0, 
\end{split}\end{align}
i.e., $\mathds{1}\(p\(x, u\) > c'\)$ achieves higher TPR. 
\end{proof}

\begin{lemma}\label{strictly concave roc}
Consider a model where 
the human decision maker follows the classification rule 
$\hat Y_i = \mathds{1}\(p\(X_i\) > c\(X_i, U_i\)\)$ 
for the $i$th patient case, where $p\(X_i\)$ is a correctly specified propensity score function $p\(X_i\) = \mathbb{P}\(Y_i=1\vert X_i\)$ and $U_i$ is a random variable that does not contain information about
$Y_i$ beyond those in $X_i$: 
$\mathbb{P}\(Y_i=1 \vert X_i, U_i\) = \mathbb{P}\(Y_i=1 \vert X_i\)$. Then the aggregate human decision maker FPR/TPR pair lies strictly below the ROC curve traced out by $\mathds{1}\(p\(X_i\) > c\)$ when $c$ varies between $0$ and $1$ under the following conditions:
\begin{enumerate}
\item The p-score  $p\(x\)$ is uniformly bounded away from $0$ and $1$ on $x\in \mathcal X$, the support of $X$.
\item $p\(X\)$ is continuously distributed with a bounded and strictly positive density.
\item\label{strict concavity condition 3}  There is no $c \in \[0, 1\]$ such that $\mathbb{P}\left\{\mathds{1}\(p\(X\) > c\) = \mathds{1}\(p\(X\) > c\(X, U\)\)\right\} = 1$. 
\end{enumerate}
\end{lemma}
\begin{proof}
First, we write human FPR/TPR pair as  $\theta_H=\(\alpha_H, \beta_H\)$, where
{\begin{align}\begin{split}\nonumber
\alpha_{H} =\frac{1}{1-p} \int \lambda\(x\) \(1-p\(x\)\) f_X\(x\) \mathrm{d}x\quad\text{and}\quad \beta_{H} =\frac{1}{p} \int \lambda\(x\) p\(x\) f_X\(x\) \mathrm{d}x.
\end{split}\end{align}
}
In the above we have used $p=\mathbb{P}\(Y_i=1\)$ and
\bs
\lambda\(x\)=\int \mathds{1}\(p\(x\)>c\(x,u\)\)f_{U \vert X}\(u \vert x\)\mathrm{d}u.
\end{split}\end{align}
We can then always find a threshold $c^*\in \[0,1\]$ that satisfies
\begin{align}\begin{split}\nonumber
			 \alpha_{O}\(c^*\) =& \frac{\int \mathds{1}\(p\(x\) > c^*\) \(1-p\(x\)\) f_X\(x\) \mathrm{d}x}{1-p}
= \frac{\int \lambda\(x\) \(1-p\(x\)\) f_X\(x\) \mathrm{d}x}{1-p}\equiv\alpha_H. 
\end{split}\end{align}
Given $c^*$, we then next define, on the ROC curve:
\begin{align}\begin{split}\nonumber
\beta_{O}\(c^*\) =& \frac{\int \mathds{1}\(p\(x\) > c^*\) p\(x\) f_X\(x\) \mathrm{d}x}{p}.
\end{split}\end{align}
By Neyman-Pearson arguments, $\phi^*\(x\)\equiv\mathds{1}\(p\(x\) > c^*\)$ solves
\begin{align}\begin{split}\nonumber
\max_{\phi\(\cdot\)}
	&\(1-c^*\)  \int \phi\(x\) p\(x\) f_X\(x\) \mathrm{d}x
- c^*  \int \phi\(x\) \(1-p\(x\)\) f_X\(x\) \mathrm{d}x\\
	&= p \(1-c^*\) \beta_{O}\(c^*\) - \(1-p\) c^* \alpha_{O}\(c^*\) \geq 
p \(1-c^*\)\beta_{H} -  \(1-p\) c^* \alpha_{H}.
\end{split}\end{align}
The difference between the two sides can be further written as
\begin{align}\begin{split}\nonumber
	&	p \(1-c^*\) \beta_{O}\(c^*\) - \(1-p\) c^* \alpha_{O}\(c^*\) -\[ p \(1-c^*\) \beta_{H} -  \(1-p\) c^* \alpha_{H}\]\\
	&= \int  \(\mathds{1}\(p\(x\) > c^*\) - \lambda\(x\)\) \(p\(x\) - c^*\) f_X\(x\) \mathrm{d}x,
\end{split}\end{align}
which is strictly positive unless  $\mathbb{P}\(\lambda\(X\)= \mathds{1}\(p\(X\) > c^*\) \) = 1$. This is ruled out by the stated assumptions. In particular, by condition \ref{strict concavity condition 3}, with positive probability, either one of the following two events (A and B) holds. 
In event A, $p\(X\) > c^*$ but $p\(X\) \leq c\(X, U\)$, in which case $\mathds{1}\(p\(X\) > c^*\)=1$ but $\lambda\(X\) < 1$. In event B, $p\(X\) \leq c^*$ but $p\(X\) > c\(X, U\)$, in which case $\mathds{1}\(p\(X\) > c^*\)=0$ but $\lambda\(X\) > 0$. Together 
$\mathbb{P}\(\lambda\(X\) \neq \mathds{1}\(p\(X\) > c^*\) \) \geq
\mathbb{P}\(A \cup B\) > 0$. \end{proof}

\begin{proof}[Proof of Lemma \ref{distributionFPRTPR}]
Recall the definitions of sample and population FPR/TPR
pairs:
\begin{align}\begin{split}\nonumber
\hat{\alpha} = \frac{\sum_{i=1}^n \(1-Y_i\) \widehat{Y}_i}{\sum_{i = 1}^n \(1-Y_i\)}, \quad 
\hat{\beta} = \frac{\sum_{i=1}^nY_i \widehat{Y}_i}{\sum_{i =
1}^n Y_i},\quad
\alpha = \frac{\mathbb{E} \[\(1 - Y_i\) \widehat{Y}_i\]}{\mathbb{E}\[1 - Y_i\]}, \quad
\beta = \frac{\mathbb{E}Y_i \widehat{Y}_i}{\mathbb{E}Y_i}.
\end{split}\end{align}
It follows from these definitions that for
$\theta = \(\alpha, \beta\)'$ and 
$\hat\theta = \(\hat{\alpha}, \hat{\beta}\)'$,
\begin{align}\begin{split}\nonumber
\sqrt{n}\(\hat\theta-\theta\)
= \hat H^{-1} \frac{1}{\sqrt{n}}
\sum_{i=1}^n
\(
\begin{array}{c}
\(1-Y_i\) \(\hat Y_i - \alpha\) \\
Y_i \(\hat Y_i - \beta\)
\end{array}
\)\quad\text{where}\quad
\hat H =
\biggl(
\begin{array}{cc}
1 - \hat p & 0 \\
0 & \hat p
\end{array}
\biggr).
\end{split}\end{align}
By the Law of Large Numbers (LLN), $\hat H$ converges to $H$
in probability:
\begin{align*}
\hat H = H + o_{\mathbb{P}}\(1\)\quad\text{where}\quad
H =
\biggl(
\begin{array}{cc}
1 - p & 0 \\
0 & p
\end{array}
\biggr).
\end{align*}
By the multivariate central limit theorem (CLT), 
\begin{align}\begin{split}\nonumber
\frac{1}{\sqrt{n}}
\sum_{i=1}^n
\(
\begin{array}{c}
\(1-Y_i\) \(\hat Y_i - \alpha\) \\
Y_i \(\hat Y_i - \beta\)
\end{array}
\)
\overset{d}{\longrightarrow} N\(0, \Omega\)\quad\text{where}\quad
\Omega=
\(
\begin{array}{cc}
\sigma_{\alpha}^2 & \\
0 & \sigma_{\beta}^2
\end{array}
\).
\end{split}\end{align}
The relations $Y_i^2 = Y_i$,
$\(1-Y_i\)^2 = 1-Y_i$ and  
$\hat Y_i^2 = \hat Y_i$ can be used to calculate that
\begin{align*}
\sigma_{\alpha}^2 =& \mathbb{E} \[\(1-Y_i\)^2 \(\hat Y_i - \alpha\)^2 = \(1-p\) \alpha\(1-\alpha\)\], \\
\sigma_{\beta}^2 =& \mathbb{E} \[Y_i^2 \(\hat Y_i - \beta\)^2 = p \beta\(1-\beta\)\].
\end{align*}
It then follows from the multivariate Slutsky's Lemma that
\begin{align}\begin{split}\nonumber
	\sqrt{n}\(\hat\theta-\theta\)
\overset{d}{\longrightarrow} N\(0,
\Sigma\) \quad\text{where}\quad
\Sigma=
\biggl(
\begin{array}{cc}
\frac{\alpha\(1-\alpha\)}{1-p} & 0\\
0 & \frac{\beta\(1-\beta\)}{p}
\end{array}
\biggr).
\end{split}\end{align}
A typical concern with the application of central limit theorems is the quality of approximation in the finite sample when each of $p, \alpha$ and $\beta$ is close to either $0$ or $1$. In particular, when the true population proportion of $p$ is very close to 
either $0$ or $1$
relative to the sample size, both the numerator and the denominator in either $\hat\alpha$ or $\hat\beta$
may converge to a Poisson type limit random variable. 
Even in large samples, 
either $\hat\alpha$ or $\hat\beta$ remains random and may not converge to a population limit. In our application, the population parameter $p$ is about $5\%$.  Given that the sample size for each doctor is at least a few hundred, the asymptotic normal limit distribution still offers a reasonable approximation.

When either $\alpha$ or $\beta$ is close to 
either $0$ or $1$, $\hat\alpha$ or $\hat\beta$ still converge
to $\alpha$ or $\beta$, but their limit distributions and convergence rates can 
both be different. 
The normal approximation in the finite sample is likely to be more accurate for doctors whose FPR and TPR are 
bounded away from $0$ and $1$.
While the frequentist and Bayesian approaches are not directly comparable, employing the Bayesian methodology does allow us to sidestep issues related to the asymptotic distribution approximation.\end{proof}

\subsection{The 
Random Forest Algorithm}\label{MachineAlgorithm}
We implemented a random forest algorithm with $N$ estimators
and at most $M$ features per node. 
For each node in a tree, no more than $M$
features will be considered to obtain the best split. The algorithm works as follows:

\begin{enumerate}

	\item For $i = 1$ to $N$:

	\begin{enumerate}

		\item Draw a bootstrap sample $\tilde{\mathbb{D}}_i$ from the training data $\mathbb{D}$ with replacements that has the same sample size as $\mathbb{D}$.

		\item Grow an unpruned tree $T_i$ using $\tilde{\mathbb{D}}_i$ by repeating the following steps for each node of the tree until the nodes are pure or until the number of leaves for a node falls below the minimum number of samples required to split:

		\begin{enumerate}

			\item Randomly choose $M$ features from the $d$ dimensional input feature vector.

			\item Select the best of the $M$ features to split 
   using the Gini impurity criterion.

			\item Split the node into two subnodes using the best feature.

		\end{enumerate}

	\end{enumerate}

	\item Obtain the 
output for the trees $\{ T_i \}_{i=1}^N$.

\end{enumerate}

Given an input feature vector $\mathbf{x}$, the random forest 
predicts the propensity score function 
by aggregating the results of $N$ trees, where the class probability $m\(\mathbf{x}\)$ is 
the average of 
the class probability of each tree in the forest.

%

\subsection{Supplement to the Frequentist Approach}\label{macine decision thresholds}

\paragraph{Frequentist approach machine decision thresholds:}
In sections \ref{heuristic frequentist approach} and \ref{asymptoticResult},
for each of the replaced doctors $j$, we generate the machine decision thresholds with a 
specific cutoff value $c_{j}$ using the following algorithm. 
First, we find the largest ($c_{j, 1}$) and smallest ($c_{j, N}$) cutoff thresholds 
corresponding to the end points of the dominating segment of the machine ROC curve, i.e., 
points B and A in 
Figure \ref{heuristic frequentist comparison}. 
To compute FPR/TPR pairs, 
we choose $N$ cutoff thresholds 
along this dominating interval of the machine ROC curve between points $B$ and $A$
using the following procedure:
\begin{enumerate}
    \item The stepsize of threshold spacing is set to $c_{j,\text{step}}=\frac{c_{j,N}-c_{j,1}}{N-1}$;
    \item For each $l=0,\ldots,N-1$, the $l + 1$-th threshold between $c_{j,1}$ and $c_{j,N}$ is set to 
    \bs
    c_{j,l + 1}=c_{j,1} + l c_{j, \text{step}}. 
    \end{split}\end{align}
\end{enumerate}
For each $l=1,\ldots,N$,
the $l$-th aggregate FPR/TPR pair is calculated using the decisions of all retained doctors and machine decisions based on $c_{j, l}$ for each replaced doctor. 

\paragraph{A test between human FPR/TPR and machine ROC curve:} An alternative is to test whether the population value of $\theta_H$ is above or below the machine ROC curve. The status quo of human decision making is typically 
chosen as the null hypothesis of 
$H_0: \theta_{H} \ \text{lies above ROC}$.
This implies that a priori, \emph{we maintain confidence in human decision making} unless overwhelming evidence
suggests otherwise. 
As the machine ROC curve is represented by
$\beta = g\(\alpha\)$, 
we can rewrite the null and alternative hypotheses as
\begin{align}\begin{split}\nonumber
H_0: \beta_{H} \geq g\(\alpha_{H}\) \quad\text{against}\quad
H_1: \beta_{H} < g\(\alpha_{H}\).
\end{split}\end{align}
Under the least favorable null hypothesis, the asymptotic distribution of a test statistic
$\hat t = \hat\beta_{H} - g\(\hat\alpha_{H}\)$ follows from combining Lemma \ref{distributionFPRTPR} with the Delta method. More precisely, if we define
$\hat B = \(-g'\(\hat\alpha_H\), 1\)$, then the following central limit theorem holds under the least favorable null hypothesis:
\begin{align}\begin{split}\nonumber
\mathbb{P}\(
\frac{
\sqrt{n}\(\hat\beta_H - g\(\hat\alpha_H\)\)
}{
\sqrt{
\hat B \hat\Sigma \hat B'
}
}
\leq \Phi^{-1}\(\alpha\)
\)\rightarrow \alpha.
\end{split}\end{align}
In the above, we use $\Phi^{-1}\(\alpha\)$ to denote the $\alpha$-th percentile of the standard normal distribution. 
A test that rejects the null hypothesis when 
$\frac{
\sqrt{n}\(\hat\beta_H - g\(\hat\alpha_H\)\)
}{
\sqrt{
\hat B \hat\Sigma \hat B'
}
}$ is less than $\Phi^{-1}\(\alpha\)$ will have asymptotic size $\alpha$. 

\subsection{Robustness Analysis}\label{Robustness Analysis}

To check the robustness of our approaches 
we also conduct the experiments 
based on doctors who diagnosed at least 500 patients.
We have 367 such doctors and 495,320 patient cases in total.
The classification set that is used to identify doctors
contains 
197,978 patient cases in the training sample and 148,579 patient cases in the validation sample. 
The remaining 148,763 cases are used to evaluate the
performance of doctors and the algorithm. 
Other settings of the algorithm implementation are kept the same. 
We identify 135 doctors as replaced among 367 doctors. 
In other words, 36.8\% of doctors are replaced by the algorithm and 63.2\% of human doctors are 
retained in the performance data set.

Figure \ref{fg:ai-doctors-500} shows the results of this experiment. Similar to
Figure \ref{fg:ai-doctors-300}, 
the aggregate FPR/TPR pair of
doctors on the performance set 
(the blue point) is 0.1942/0.2135.
The AUC of the random forest model is 0.6892.
The highest 
(yellow) point on the cyan interval has a FPR of 0.1947 and a TPR of 0.3419.
These results represent an improvement of 60.1\% in the TPR
and a 
slight worsening of 0.3\% in the FPR. 
The 
deterioration of the FPR 
is numerically possible because the FPR/TPR pairs are calculated out of sample using the performance data set.


\begin{figure}
  \centering
\subfigure[ROC curve and FPR/TPR pairs of doctors and machine decisions]{\includegraphics[width=0.7\columnwidth]{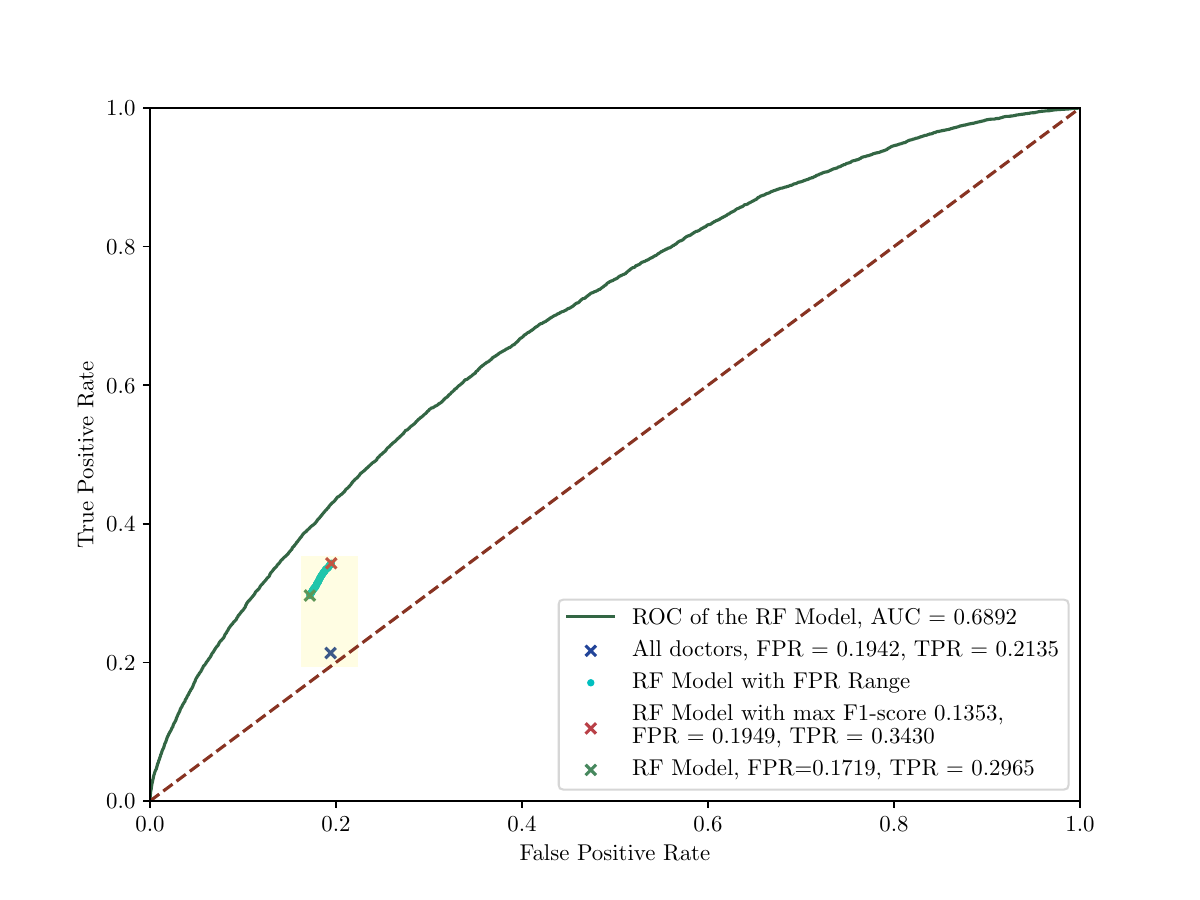}
\label{fg:ai-doctors-500-upper}
}
\subfigure[Zoomed-in version]{\includegraphics[width=0.7\columnwidth]{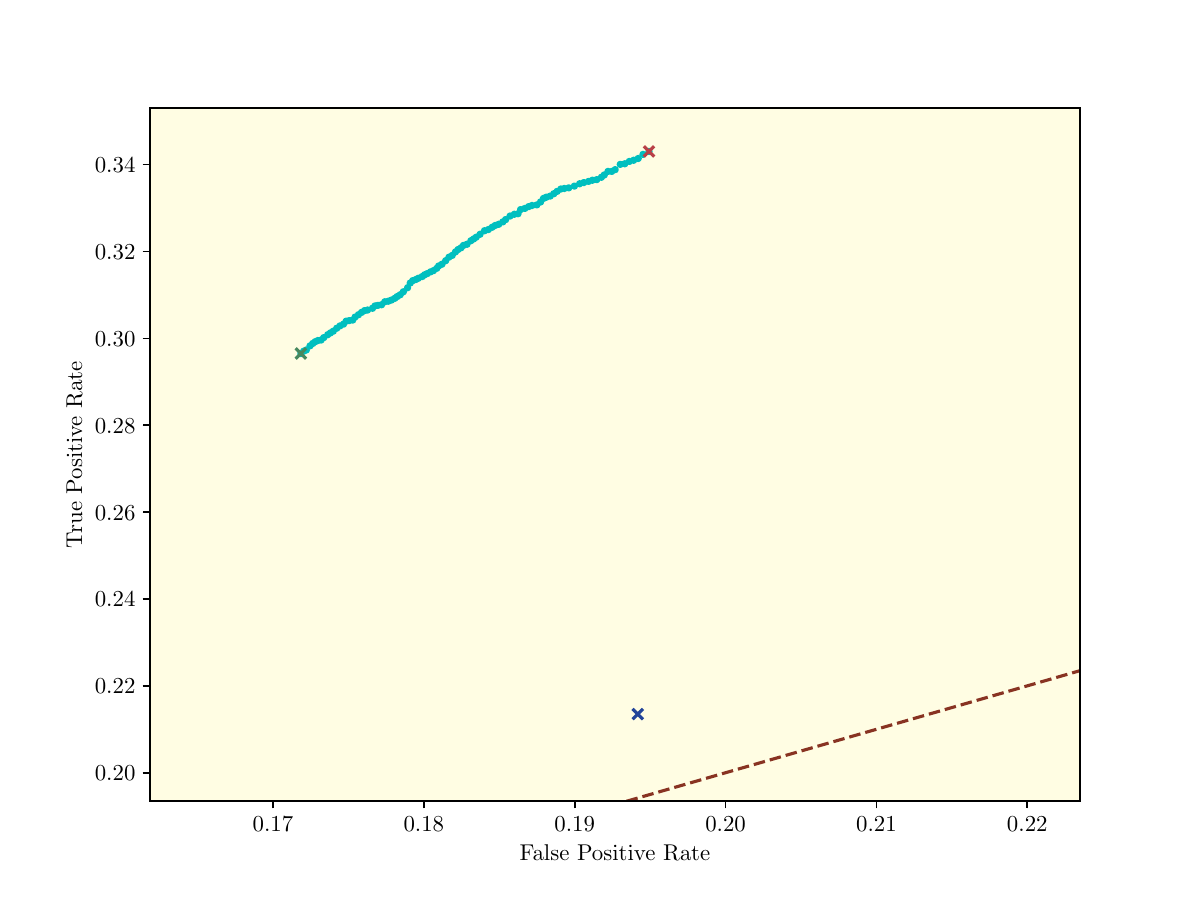}
\label{fg:ai-doctors-500-lower}
}

  \caption{Empirical results of combining doctors' and machine decisions: the heuristic frequentist approach (doctors' diagnoses >= 500)} 
\fnote{\textit{Notes}: Panel 
\subref{fg:ai-doctors-500-upper} 
draws the ROC curve in the test set of the Random Forest classifier together with the
blue point representing the FPR/TPR pair of all doctors. The cyan interval collects the FPR/TPR pairs of frequentist 
replacement strategies  along the dominating segment of the machine ROC curve. 
Panel \subref{fg:ai-doctors-500-lower} 
provides a magnified view
of the cyan interval.
}
\label{fg:ai-doctors-500}
\end{figure}

The lower left (green) point of
the cyan interval achieves 0.1719 for the FPR and  0.2965 for the TPR,
representing an improvement of 38.9\% in the TPR and a reduction of 11.3\% in the FPR. 
The red point on the cyan curve achieves a maximum F1 score of 0.1353. 
The overall performance of experimenting with doctors diagnosing at least 500 patients does not 
 differ substantially from the previous experiments with doctors diagnosing at least 300 patients.
Figure \ref{frequentist replaced doctor scatter plot-500} shows a scatter plot of the FPR/TPR pairs of both replaced and retained doctors against the ROC curve in the performance data set.

\begin{figure}
  \centering
\includegraphics[width=0.8\columnwidth]{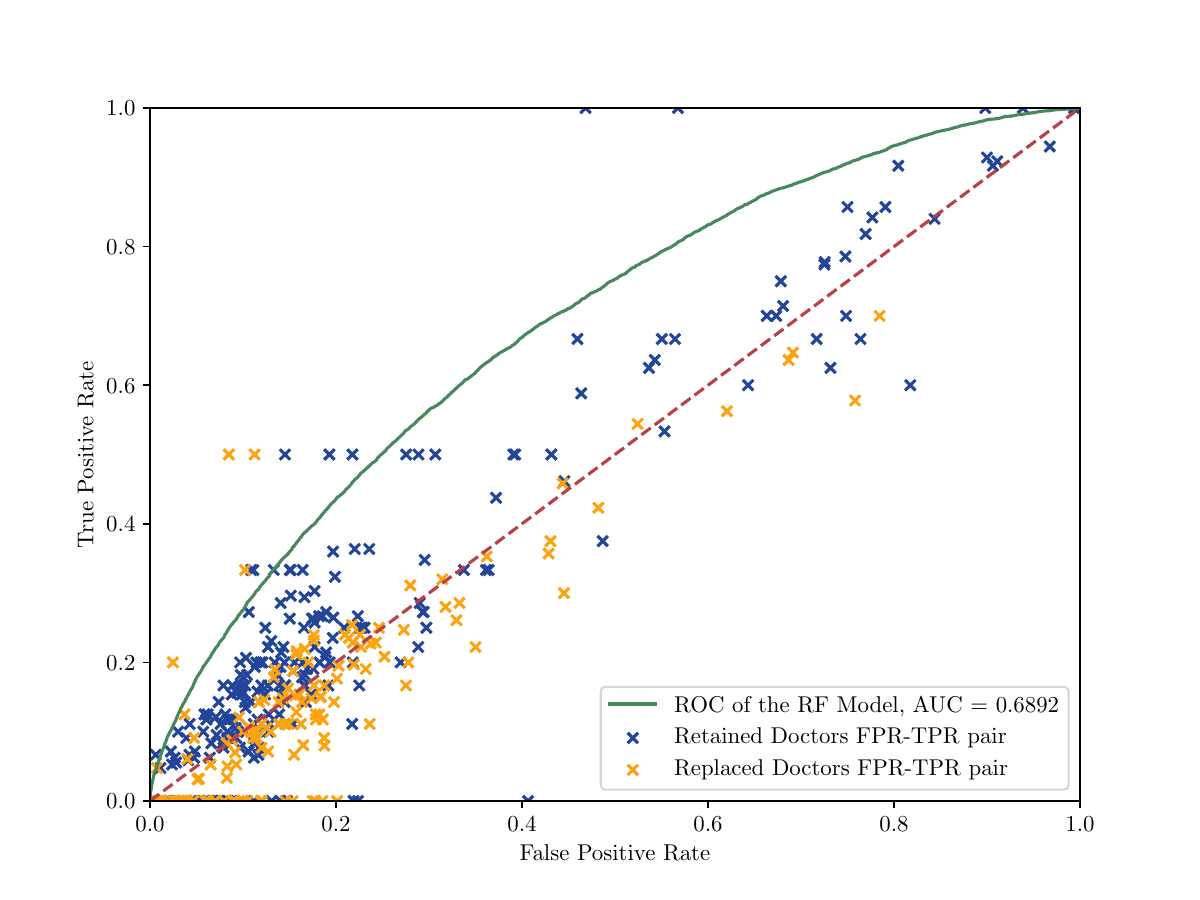}
\caption{Scatter plot of frequentist doctor replacement in the performance data set (doctors' diagnoses >= 500)}
\fnote{\textit{Notes}: The FPR/TPR pairs in the scatter plot are calculated in the test set of patient cases for each doctor using the raw data. It is possible for the TPR/FPR pairs of a small fraction of replaced doctor to lie above the machine ROC curve in the test set.
}
  \label{frequentist replaced doctor scatter plot-500}
\end{figure}

Figure \ref{fg:ai-doctors-bayesian-curve-500}
shows the experiment results for section 
\ref{bayesian result}
using only doctors who diagnosed at least 500 cases.
Among all 367 such doctors, 201 are classified 
as replaced doctors. 
Therefore, 166 (45.2\%) human doctors
and 201 (54.8\%) machine doctors are involved
in the evaluation on the performance data set.
The aggregate  FPR and TPR for combined decision making are 0.1783 and 0.3292, respectively, 
representing an improvement of
54.2\% on the TPR and a reduction of 8.2\% on the FPR. 
Figure \ref{bayesian replaced doctor scatter plot-500} is a scatter plot of the FPR/TPR pairs
of both replaced and retained doctors against the ROC curve in the performance data set.

\begin{figure}
  \centering
  \subfigure[ROC curve and FPR/TPR pairs of doctors and machine decisions]{\includegraphics[width=0.7\columnwidth]{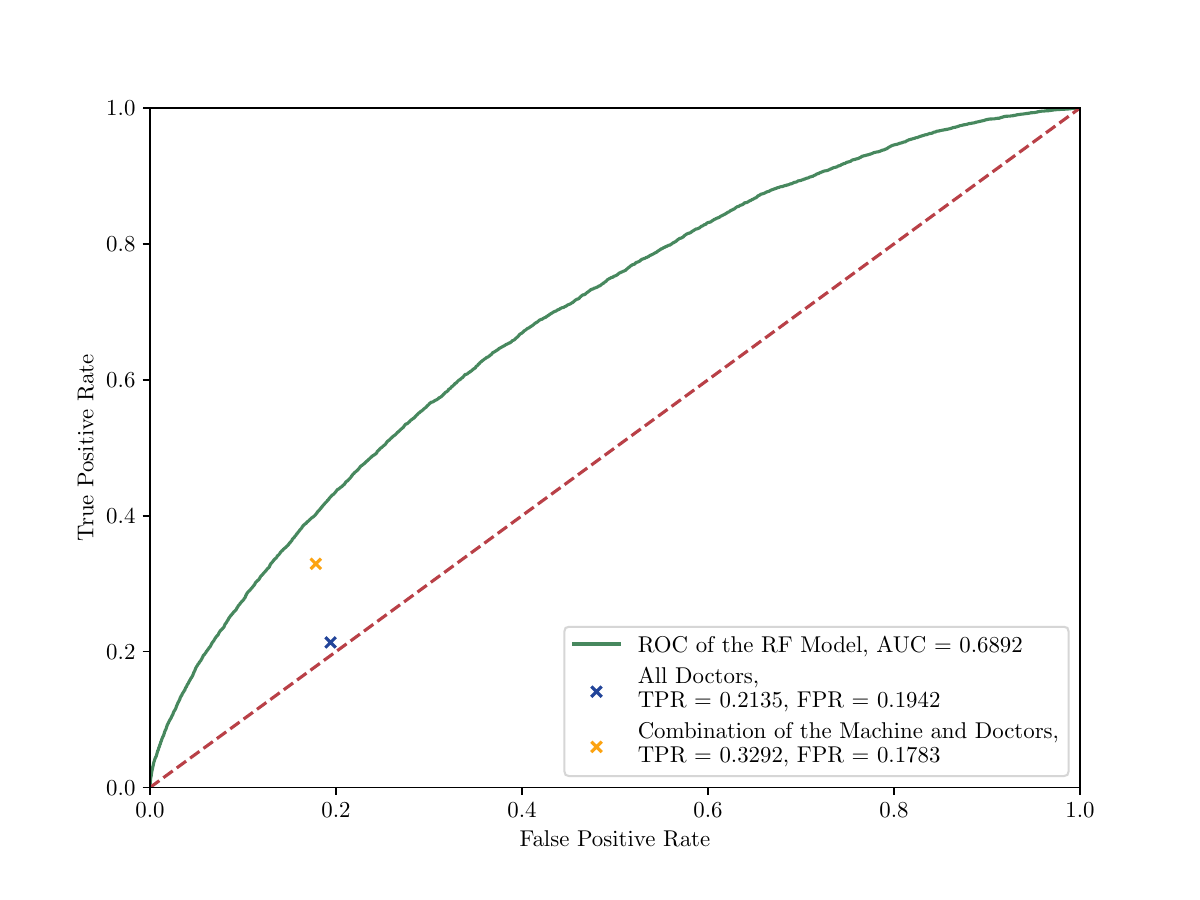}\label{fg:ai-doctors-bayesian-curve-500}}
\subfigure[Scatter plot of Bayesian doctor replacement in the performance data set]{\includegraphics[width=0.7\columnwidth]{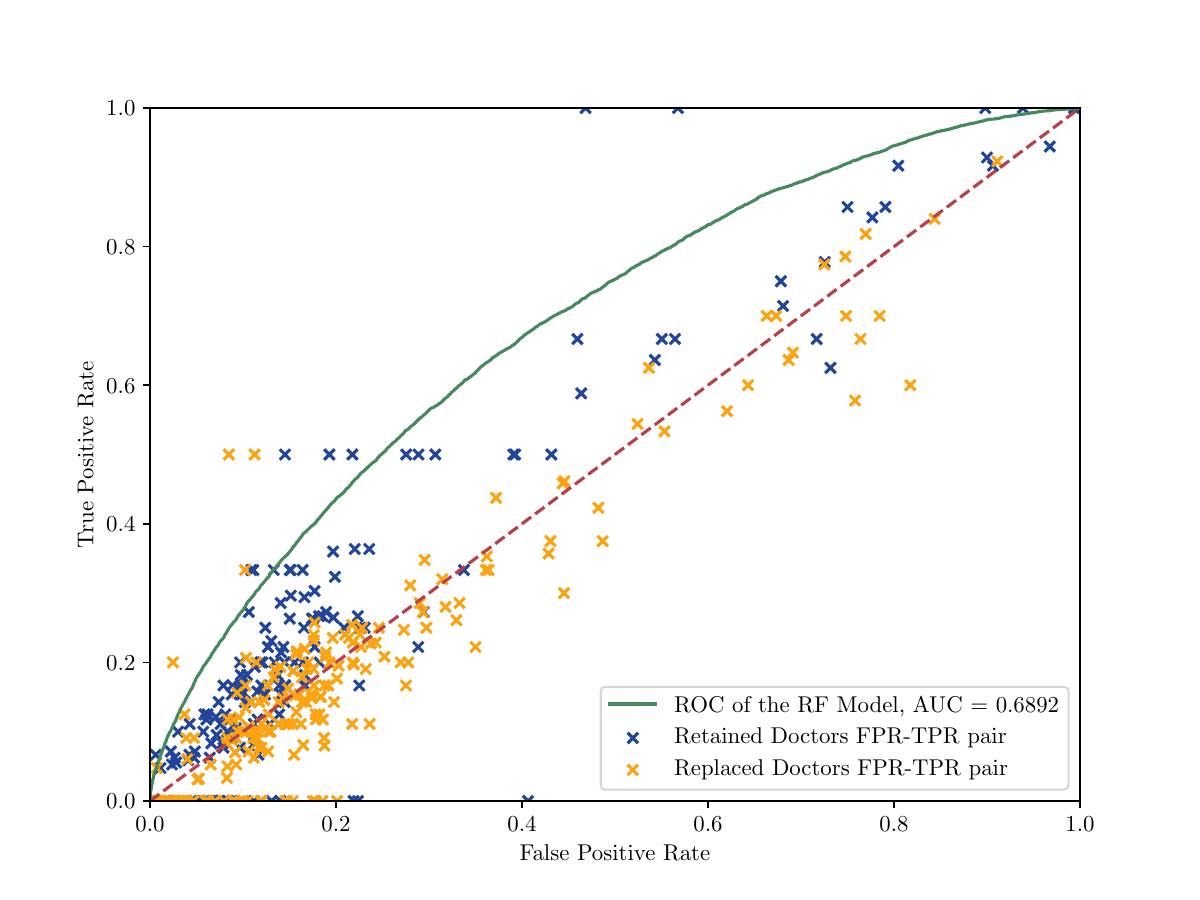}
\label{bayesian replaced doctor scatter plot-500}}
  \caption{Empirical results of the Bayesian approach (doctors' diagnoses >= 500)} 
\fnote{\textit{Notes}: The panel \subref{fg:ai-doctors-bayesian-curve-500} draws the ROC curve in the test set of the Random Forest classifier, the
FPR/TPR pair of all doctors, and the FPR/TPR pair of doctor/machine combination. The panel \subref{bayesian replaced doctor scatter plot-500} scatterplot represents the FPR/TPR pairs of individual doctors, such that yellow points correspond to the replaced doctors and the blue points non-replaced doctors.}
  \label{fg:ai-doctors-bayesian-500}
\end{figure}

The replacement paths  shown in 
the 
Figure \ref{fg:ai-doctors-bayesian-score-500},
which correspond to the alternative loss 
functions in equations \eqref{euclidean distance loss function} to 
\eqref{complement set decomposition weight loss function} and use only doctors who diagnosed at least 500 patient cases, are similar to Figure \ref{fg:ai-doctors-bayesian-score-300} for 
 doctors diagnosing at least 300 cases. 

\begin{figure}
  \centering
  \subfigure[ROC curve and Replacement Paths]{\includegraphics[width=0.7\linewidth]{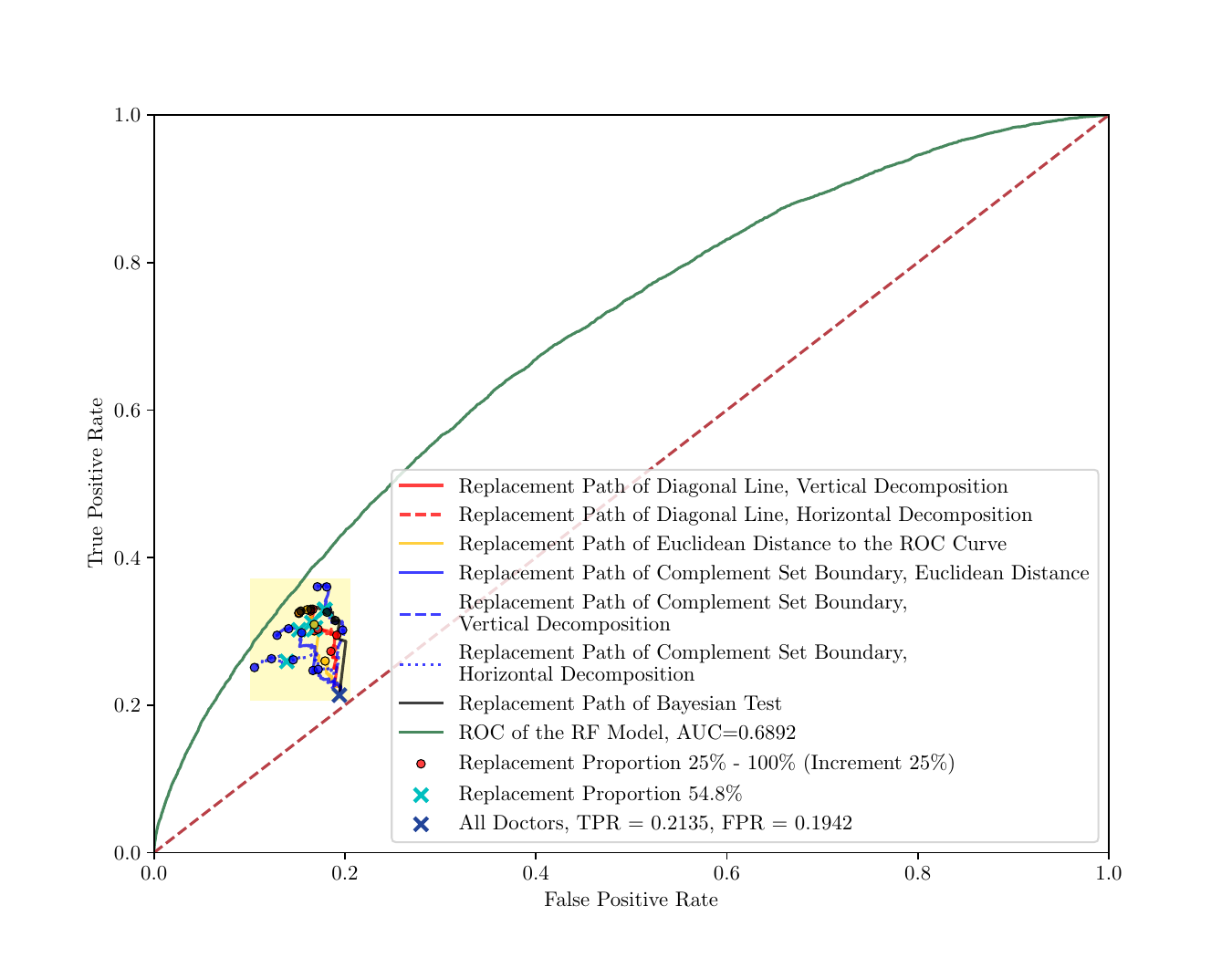}\label{fg:ai-doctors-bayesian-score-horizonal-upper-500}}
  \subfigure[Zoomed-in Figure of Replacement Paths]{\includegraphics[width=0.7\linewidth]{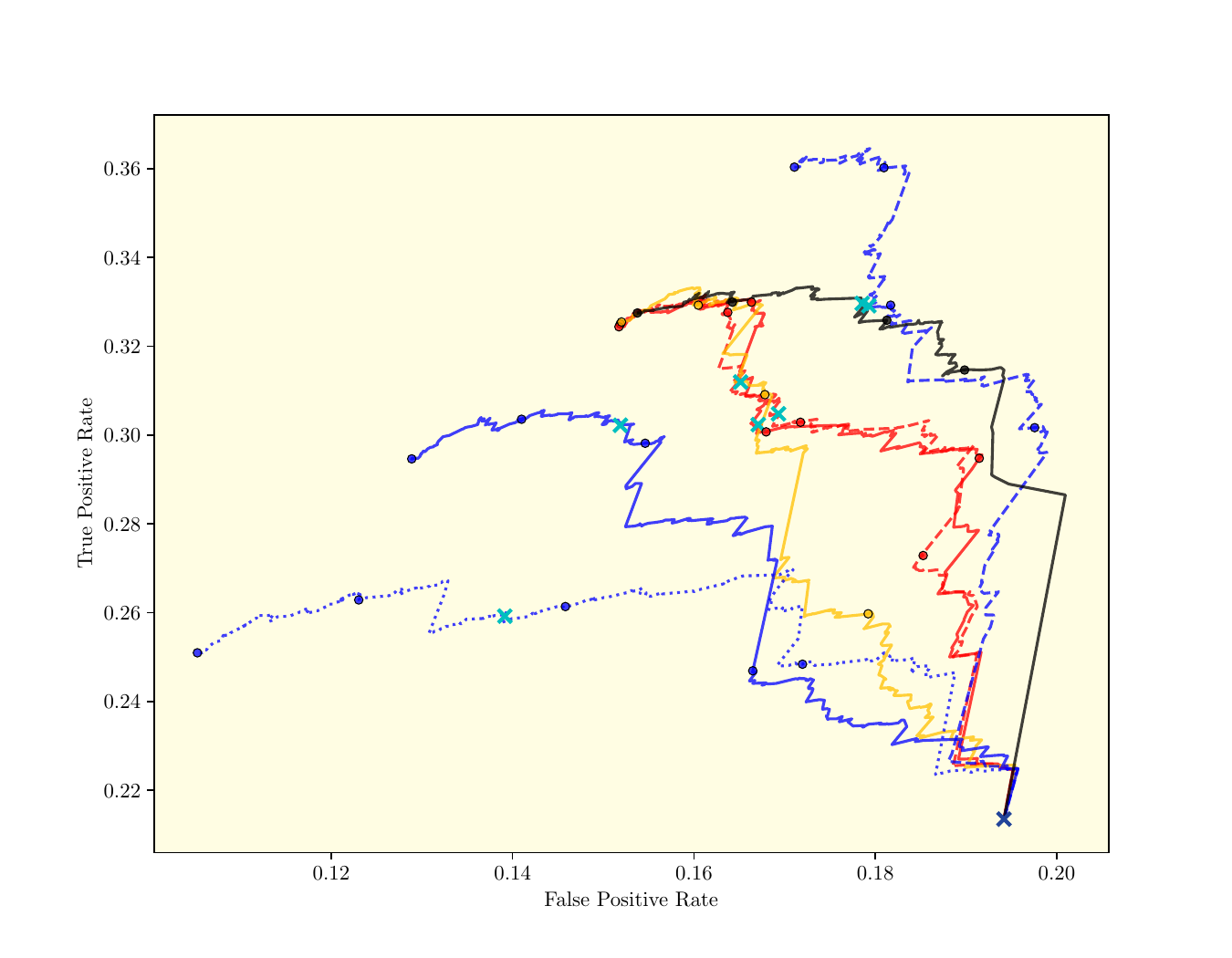}\label{fg:ai-doctors-bayesian-score-horizonal-lower-500}}
  \caption{Bayesian Approach with Loss Functions Constructed by Euclidean Distance and Randomized Decomposition (Doctors' Diagnoses >= 500)}
\fnote{Notes: The replacement paths trace the dynamics of the aggregate FPR/TPR pairs of combining 
the doctor decisions with the machine decisions, when the replacement ratio increases from 0\% of decisions entirely by human to 100\% of decisions entirely by machine algorithms. Panel \subref{fg:ai-doctors-bayesian-score-horizonal-lower-300}
magnifies a portion of panel \subref{fg:ai-doctors-bayesian-score-horizonal-upper-300} containing the replacement paths.
}
  \label{fg:ai-doctors-bayesian-score-500}
\end{figure}


\subsection{Randomized Replacement by the Machine Algorithm}\label{randomized replacement}

The experiments in section \ref{classify doctors} identify doctors to be replaced by the machine algorithm. 
In this section we consider an extension where the machine decision is randomly accepted by the doctors according 
to a pre-specified probability. 
Each replaced doctor will accept the machine decision with a probability of $\lambda \in \[0,1\]$ for each patient case.
For each replaced doctor $j$, we generate decisions for their patient cases in the performance data set using the following randomized rule that combines $\hat y_{i,M}$ and $\hat y_{i,j}$ in Assumption 1:
\bs
\bar Y_{i,j} = \mathds{1}(N_{i,j} \leq \lambda) \hat Y_{i,M}
 + \mathds{1}(N_{i,j} > \lambda)\hat Y_{i,j},
\end{split}\end{align}
where $N_{i,j}$ is independent and uniformly distributed on $\[0,1\]$. 
The larger the $\lambda$, the more likely the doctor will accept the machine algorithm to make the decision.


Figure \ref{fg:ai-doctors-bayesian-300-acprate} shows the results of applying the 
randomized decision rule with different acceptance rates $\lambda$ to the doctors 
diagnosing at least 300 patient cases identified by the Bayesian test in 
section \ref{the Bayesian approach}. 

\begin{figure}
  \centering
  \subfigure[Replaced doctors involved]{\includegraphics[width=0.7\columnwidth]{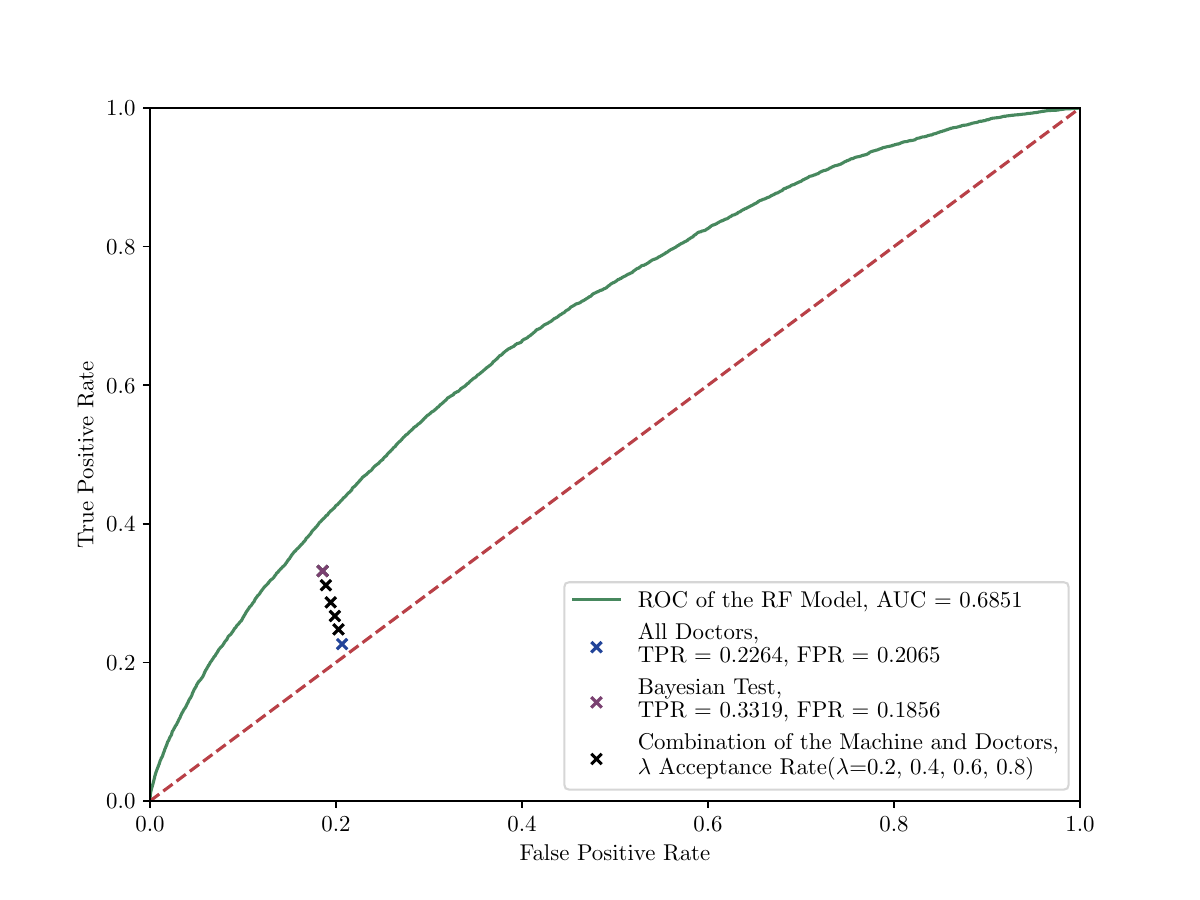}\label{fg:ai-doctors-bayesian-300-acprate}}
  \subfigure[All doctors involved]
  {\includegraphics[width=0.7\columnwidth]{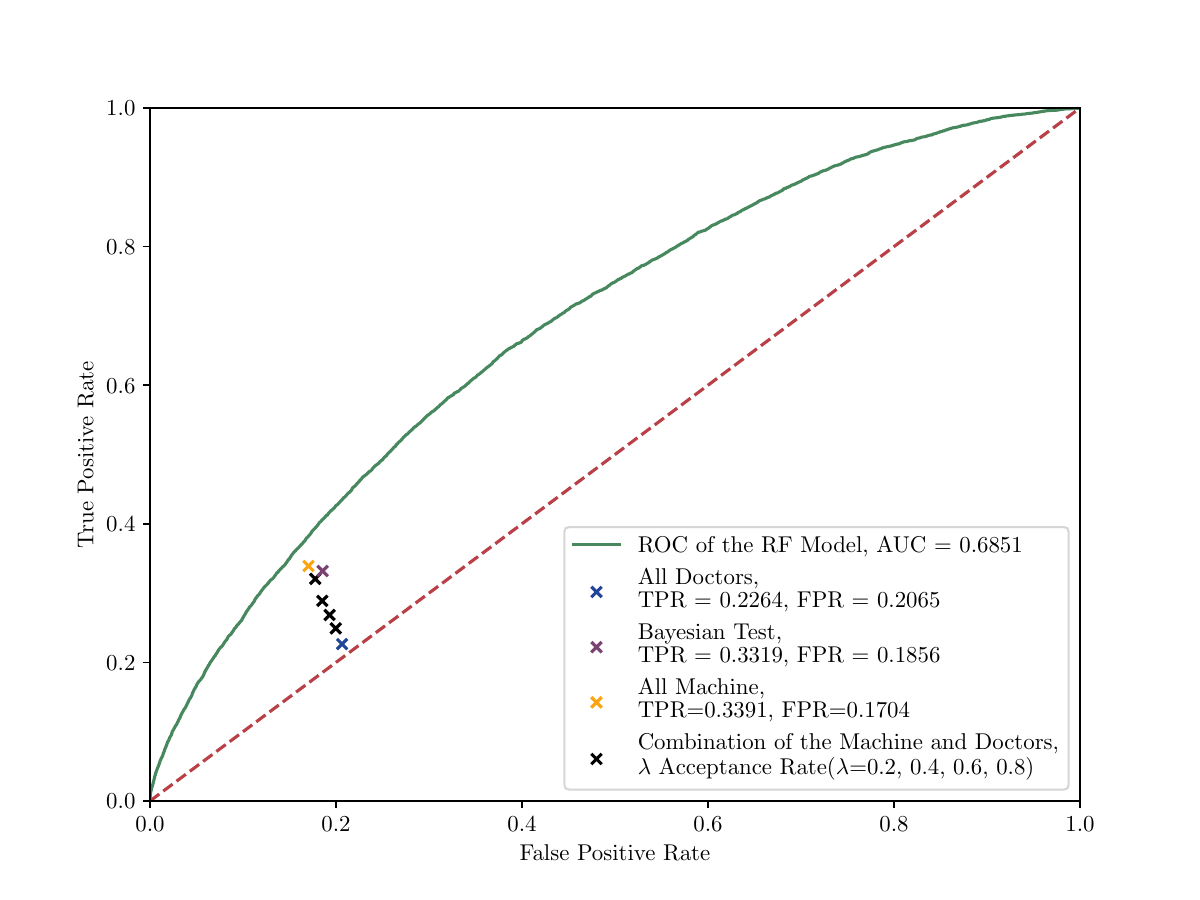}\label{fg:ai-doctors-bayesian-all-300-acprate}
  }
  \caption{Results of the Bayesian approach under different acceptance rates}
  \label{fixed lambda bayesian}
\fnote{\textit{Notes}: 
Figure \subref{fg:ai-doctors-bayesian-300-acprate} shows 
the results of applying the 
randomized decision rule with different acceptance rates $\lambda$ to the replaced doctors 
diagnosing at least 300 patient cases.  The FPR/TPR pairs are tabulated in Panel A of Table \ref{tab:perf_to_acceptrate}.
Figure \subref{fg:ai-doctors-bayesian-all-300-acprate} includes all doctors as in  Panel B of Table \ref{tab:perf_to_acceptrate}.
Replacing all doctors only marginally improves upon the Bayesian test approach.
}
\end{figure}

The black crosses on the figures correspond to $\lambda \in \{0.2, 0.4, 0.6, 0.8\}$. 
As the acceptance rate $\lambda$ grows, 
the black cross gradually moves from the blue cross (where all doctors make their own decisions; $\lambda = 0$) to the purple cross (where replaced doctors 
rely on machine decisions; $\lambda = 1$). 
The pairs of FPR/TPR under various levels of acceptance rates are tabulated in panel A of Table \ref{tab:perf_to_acceptrate}. 
The overall FPR/TPR performance improves  monotonically when 
the replacement probability $\lambda$ 
increases from $0$ to $1$.

\begin{table}[htbp]
\centering
\caption{Overall performance under different acceptance rates of machine decision}
\begin{tabular}{ccccc}
\toprule
\multicolumn{5}{l}{Panel A: Replaced doctors involved} \\
\midrule
Accept Rate $\lambda$ & FPR   & Reduction of FPR & TPR   & Improvement of TPR \\
\midrule
0.0   & 0.2065 & -     & 0.2264 & - \\
0.2   & 0.2026 & 1.89\% & 0.2462 & 8.75\% \\
0.4   & 0.1977 & 4.26\% & 0.2683 & 18.51\% \\
0.6   & 0.1944 & 5.86\% & 0.2863 & 26.46\% \\
0.8   & 0.1897 & 8.14\% & 0.3110 & 37.37\% \\
1.0   & 0.1856 & 10.12\% & 0.3319 & 46.60\% \\
\toprule
\multicolumn{5}{l}{Panel B: All doctors involved} \\
\midrule
Accept Rate $\lambda$ & FPR   & Reduction of FPR & TPR   & Improvement of TPR \\
\midrule
0.0   & 0.2065 & -     & 0.2264 & - \\
0.2   & 0.1994 & 3.44\% & 0.2485 & 9.76\% \\
0.4   & 0.1919 & 7.07\% & 0.2713 & 19.83\% \\
0.6   & 0.1848 & 10.51\% & 0.2893 & 27.78\% \\
0.8   & 0.1777 & 13.95\% & 0.3148 & 39.05\% \\
1.0   & 0.1704 & 17.48\% & 0.3391 & 49.78\% \\
\bottomrule
\end{tabular}
\label{tab:perf_to_acceptrate}
\fnote{\textit{Notes}: Panel A tabulates the pairs of FPR/TPR under various levels of acceptance rates $\lambda$ when the randomizeddecision is applied to the replaced doctors. Panel B includes all the doctors.
The FPR/TPR pairs improves monotonically when the replacement probability $\lambda$ increases from 0 to 1.
}
\end{table}

We also consider a setting where
all 
doctors are randomized into being replaced by artificial intelligence.
Figure \ref{fg:ai-doctors-bayesian-all-300-acprate} shows the results when different acceptance rates $\lambda$ are used to randomize all doctors. When a patient case is randomized into being replaced by the machine learning, we continue to apply the 
baseline Bayesian test in \eqref{bayesian basic numerical} to determine the decision threshold on the machine ROC curve. 
Replacing identified doctors (45.2\% of the total headcount) produces a much better FPR/TPR pair than randomly replacing the 
same proportion of doctors' decisions. In addition, replacing all doctors (the yellow cross) is not saliently better than 
replacing the identified doctors only (the purple cross). 

Instead of relying on a single constant acceptance rate $\lambda$, 
patient cases can also be randomized by individualized acceptance rates $\lambda$'s that 
vary with the diagnosing doctor.
We increase the acceptance probability $\lambda$ when the posterior dominance probability 
$\hat q_{\max}$ in the Bayesian test \eqref{bayesian basic numerical}  goes up. 
Figure \ref{linear acceptance rate incapable more} displays two scenarios. 
In scenario 1, we rank the posterior dominance probability 
$\hat q_{\max}$ of all doctors 
using the baseline Bayesian test \eqref{bayesian basic numerical} 
and assign them an acceptance rate $\lambda$ that depends linearly on their ranks. 
The acceptance rate $\lambda$ of the least dominated doctor, with rank 1, is set to 0, 
while $\lambda$ for the most dominated doctor, with rank $n$, is set to 1. 
The acceptance probability for a 
doctor whose rank $r$ is between 1 and $n$ is $\lambda=\frac{r-1}{n-1}$. 
Scenario 2 is similar to scenario 1, 
but only the replaced doctors identified by the baseline Bayesian test 
\eqref{bayesian basic numerical} are randomized.


\begin{figure}
  \centering
  \subfigure[Scenario 1: all doctors linear acceptance rate]{\includegraphics[width=0.7\columnwidth]{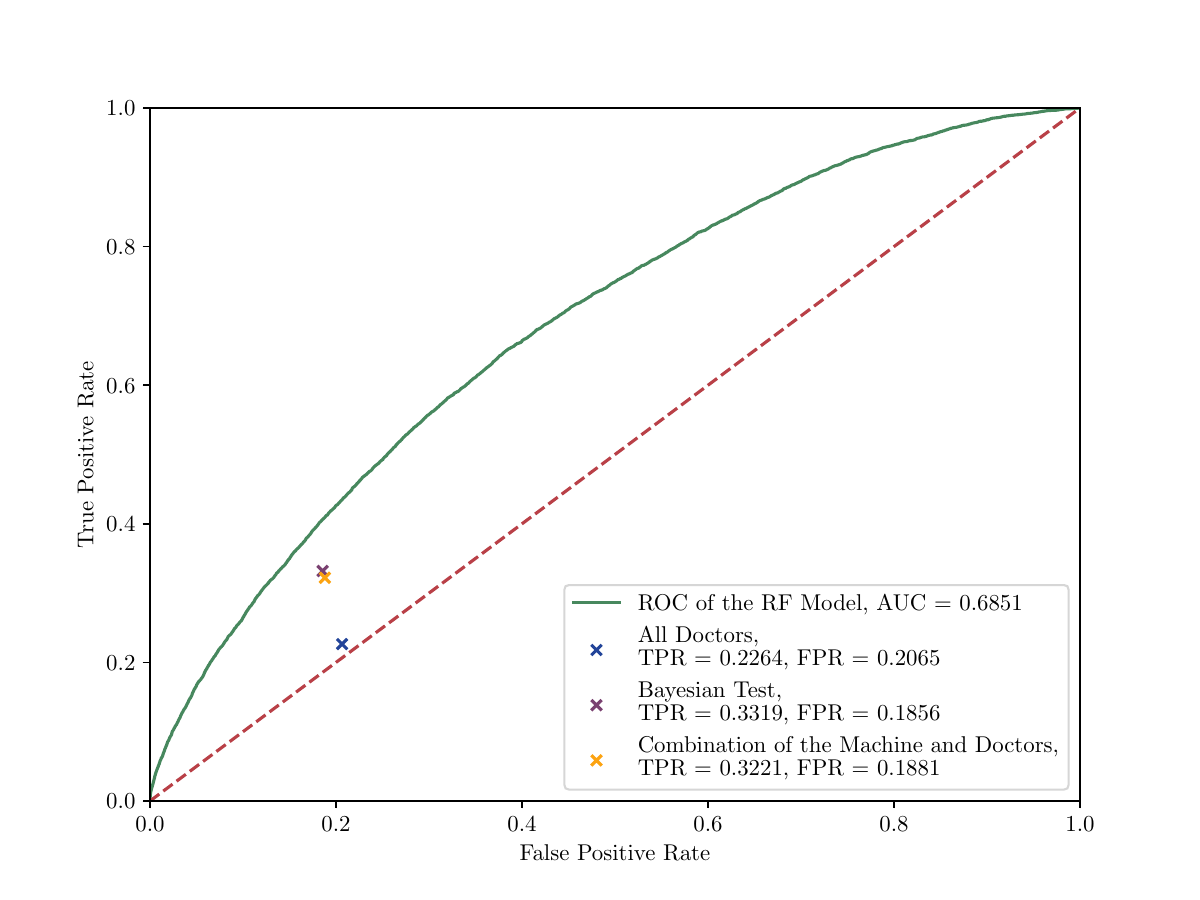}\label{fg:linear acceptance rate all doctors more 300}}
  \subfigure[Scenario 2: replaced doctors linear acceptance rate]
  {\includegraphics[width=0.7\columnwidth]{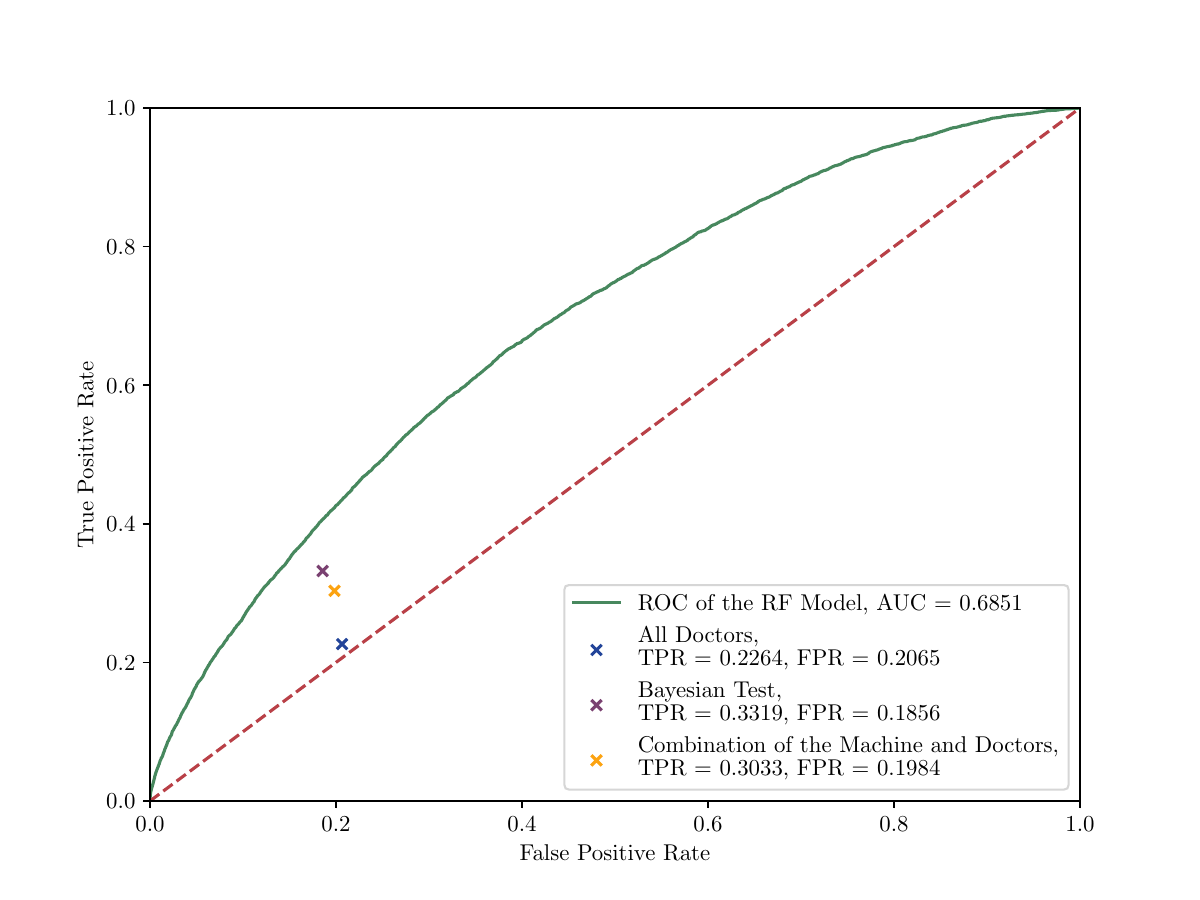}\label{fg: linear acceptance rate less capable doctors more 300}}
  \caption{Results of the Bayesian approach using capability dependence acceptance rates: replaced doctors more like to adopt the machine decisions}
  \label{linear acceptance rate incapable more}
\fnote{\textit{Notes}: 
In this figure individualized acceptance rates specific to each doctors are used to randomize patient cases.
Scenario 1 ranks the posterior dominance probability $\hat q_{\max}$ of all doctors using the baseline Bayesian test 
\eqref{bayesian basic numerical} and assigns them an acceptance rate $\lambda=\frac{r-1}{n-1}$ that depends linearly 
on their ranks. 
Scenario 2 is similar to scenario 1, but only the replaced doctors identified by the 
baseline Bayesian test  \eqref{bayesian basic numerical} are randomized.
}
\end{figure}

Figure \ref{linear acceptance rate incapable less} illustrates the consequence of an alternative 
procedure 
for determining the acceptance rates $\lambda$ based on the reverse 
rank. 
Scenario 3 and scenario 4 are counterparts to 
scenarios 1 and 2. 
The difference is that here $\lambda$ is calculated as $\frac{n-r}{n-1}$. We interpret the reverse rank ordering as retained doctors being 
more willing to accept machine decisions.

\begin{figure}
  \centering
  \subfigure[Scenario 3: all doctors involved] {\includegraphics[width=0.7\columnwidth]{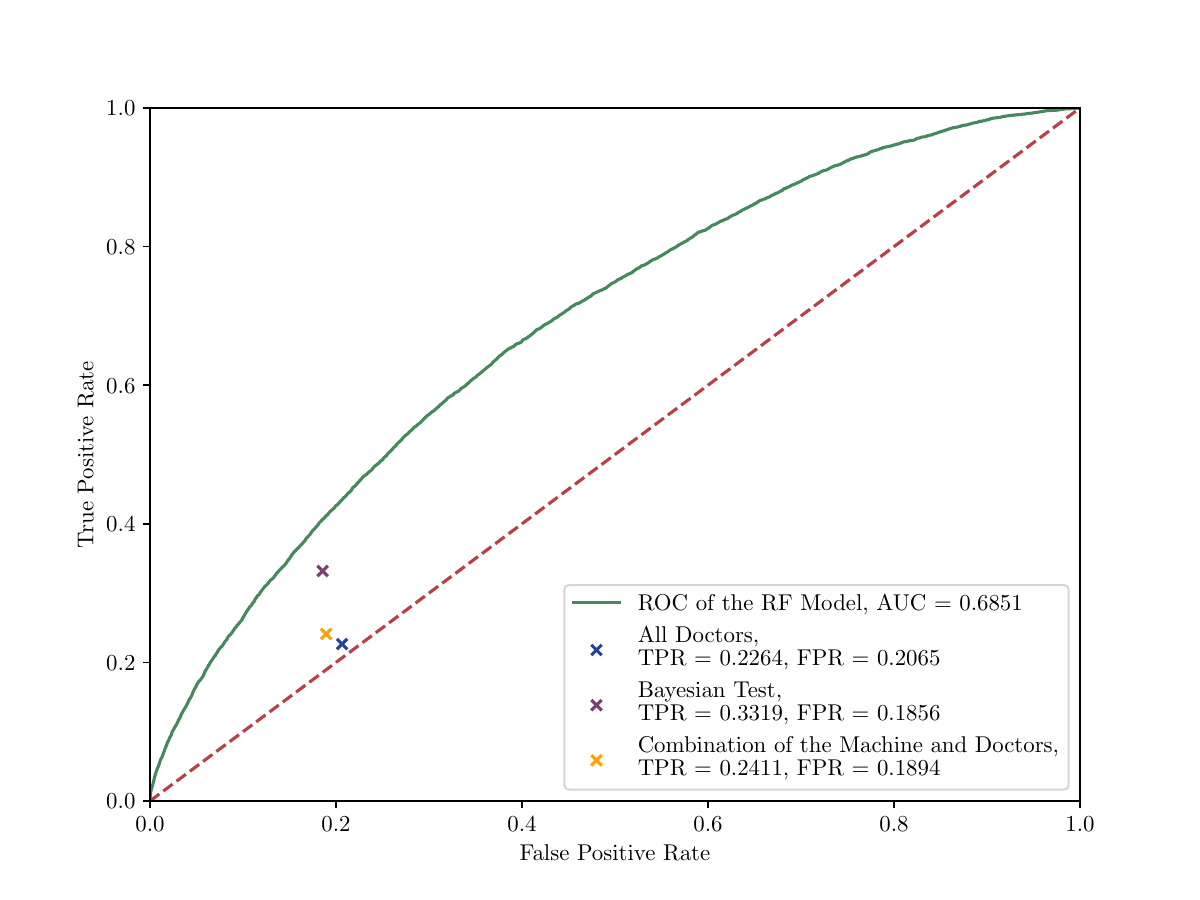}\label{fg:linear acceptance rate all doctors less 300}}
  \subfigure[Scenario 4: less capable involved]
  {\includegraphics[width=0.7\columnwidth]{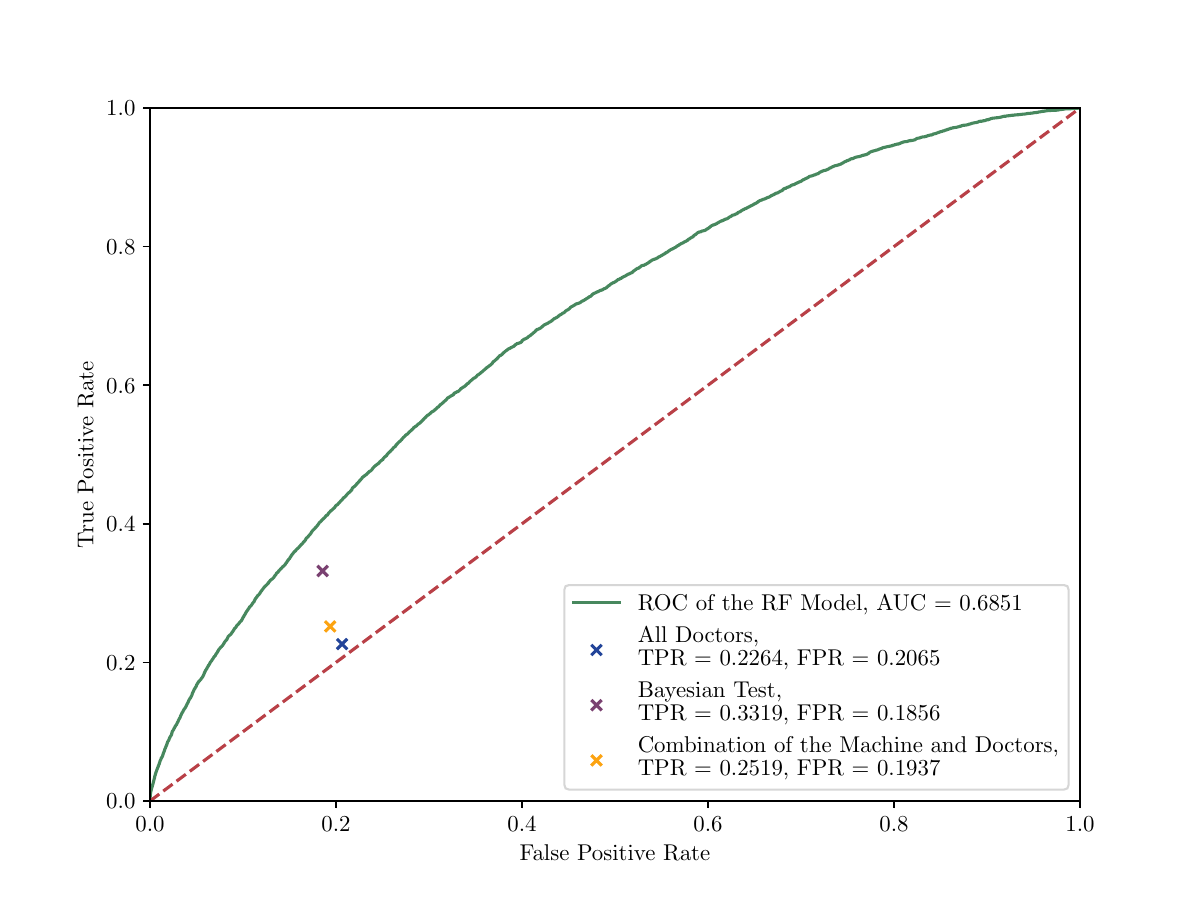}\label{fg: linear acceptance rate less capable doctors less 300}}
  \caption{Results of the Bayesian approach using capability dependence acceptance rates: capable doctors more willing to accept the machine decisions}
  \label{linear acceptance rate incapable less}
\fnote{\textit{Notes}: 
This figure shows the counterpart to Figure \ref{linear acceptance rate incapable more}. The acceptance rate $\lambda=\frac{n-1}{n-1}$ depends on the reverse rank implying that
retained doctors are 
more willing to accept machine decisions.
}
\end{figure}

Compared with scenarios 1 and 2, the aggregate FPR/TPR pairs 
in scenario 3 and 4 are inferior. 
If replaced doctors are not willing to adopt the decision from a machine algorithm, the average diagnosis quality will not improve significantly. 
These results accord with our intuition. 
On the one hand, improving the overall FPR/TPR performance calls for higher acceptance rate 
for replaced doctors. 
On the other hand, using the machine algorithm to  
replace the retained doctors' decisions with high probability is less likely to be beneficial. 

\subsection{A patient-case specific substitution algorithm}\label{fine-grained replacement}

There is a larger space for replacing a particular human decision maker
only in specific contexts. 
In this appendix 
we experimented with an alternative decision making process
during which only some of the patient cases of a given doctor are diagnosed by the machine algorithm. 
Given 
that the diagnosis by the doctor is observed in each data point in addition to the outcome label, we 
use 
the machine algorithm not only to predict the likelihood of the outcome of abnormal birth, but also to evaluate when the doctor's
diagnosis are more likely to differ from the ground-truth outcome. 
The recognization that
machine algorithms are more likely to be applicable in some candidate cases than in others for a given decision maker
has been emphasized 
by \citeappendix{donahue2022human}, 
\citeappendix{raghu2019algorithmic}  and \citeappendix{raghu2019direct}
among others.  

The alternative substitution strategy we consider involves several steps. 
In the first step, we train a model, denoted as  $F_{\text{mis}}\(x\)$, to predict
the likelihood of a doctor misjuding a patient case. The label in this model is $\mathds{1}\(\hat Y \neq Y\)$, where $\hat Y$ is the diagnosis
by the doctor whether this is a high risk prenancy case and $Y$ is the outcome of whether the observation corresponds to an abnormal 
birth. 
In the second step, we train another model,  denoted as  $F_{\text{doc}}\(x\)$, to 
predict the likelihood of a doctor diagnosing a patient case as high risk. In the third step, we use 
the patient cases
for a given doctor to estimate the doctor's ratio between the cost of false positive and the cost of false negative. 
The alternative substitution strategy combines information from these steps, and replaces a given patient case by the machine algorithm 
when the likelihood of the doctor's diagnosis and the patient outcome being different is higher than a threshold. The threshold
value depends on how likely a doctor will diagnose a patient case as high-risk. When the doctor's cost of false positive is higher
than the cost of false negative, the threshold should be lower when the doctor is more likely to diagnose a patient case as high-risk.
Conversely,  when the doctor's cost of false positive is lower than the cost of false negative, the threshold should be higher when the doctor is more likely to diagnose a patient case as high-risk. 
Whether a doctor is likely to diagnose a patient case as high-risk is indicated
by whether $F_{\text{doc}}\(x\)$ is higher than a given cut-off value.

Specifically, the replacement rule for a patient case with features $x$ is given by
\begin{equation*}
 \begin{cases}
    F_{\text{mis}}\(x\) > c_r, \quad \text{if} \ F_{\text{doc}}\(x\) > c_d, \\
    F_{\text{mis}}\(x\) > \frac{c_r}{t}, \quad \text{otherwise}.
  \end{cases}
\end{equation*}
In the above, $t$ is the ratio between the cost of false negative and the cost of false positive for the given doctor. This ratio $t$ is calculated as follows 
\begin{equation}
  \label{eq to illustrate slope of roc tangent}
   t = (1 - \mathbb{P}(Y = 1)) / \(s_t \cdot \mathbb{P}(Y = 1)\),
\end{equation}
where $s_t$ is the slope of the tangent of the machine ROC curve at the point
of the machine ROC curve which dominates the largest area of the simulated doctor's 
posterior distribution as described in \eqref{bayesian basic numerical}. 
\begin{figure}[hp]
  \centering
  \includegraphics[width=0.85\columnwidth]{"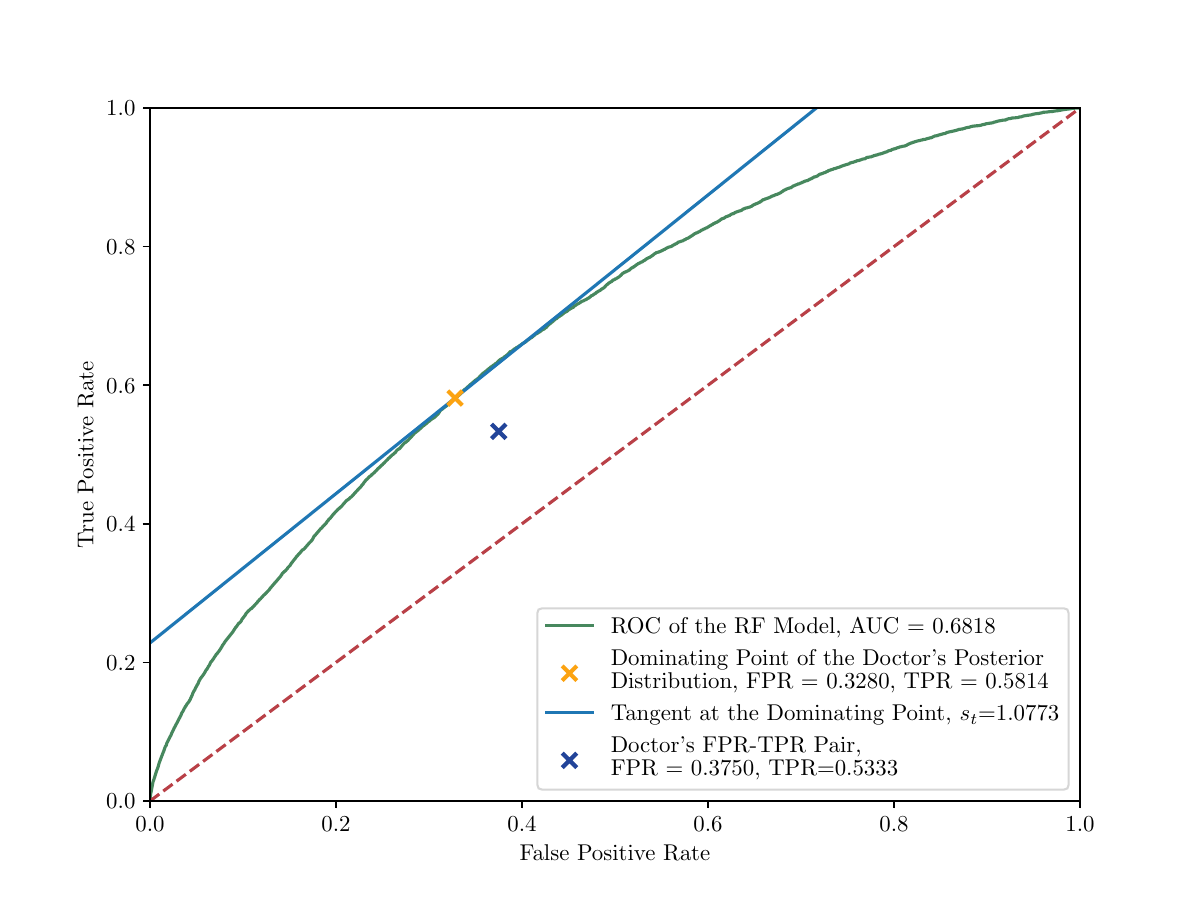"}
  \caption[Calculating t]{Calculation of ratio between the cost of false positive and the cost of false negative}
  \label{fg:appendixes5 calculating t}    
\fnote{\textit{Notes}: 
This figure illustrates the calculation of 
the ratio between the cost of false negative and the cost of false positive for the given doctor
in equation \eqref{eq to illustrate slope of roc tangent}. 
}
\end{figure}
As illustrated in 
Figure \ref{fg:appendixes5 calculating t}, the slope of the aqua line tangent to ROC curve at the yellow cross is $s_t$. 
If 
$C_{0R}$ denotes the false positive cost of misclassifying $y=0$ as $y=1$, and 
$C_{1A}$ denotes the false negative cost of misclassifying $y=1$ as $y=0$, 
the expected cost of a decision rule is 
\bs
&C_{0R} \mathbb{P}\(Y=0, \hat Y=1\) + C_{1A} \mathbb{P}\(Y=1, \hat Y=0\)\\
&= 
C_{0R} \(1 - \mathbb{P}\(Y=1\)\) \text{FPR} + C_{1A} \mathbb{P}\(Y=1\) \(1 - \text{TPR}\)\\
&= C_{0R} \(1 - \mathbb{P}\(Y=1\)\) \text{FPR} - C_{1A} \mathbb{P}\(Y=1\) \text{TPR}
+ C_{1A} \mathbb{P}\(Y=1\).
\end{split}\end{align}
Therefore the slope of the tangent of the ROC curve at the optimizing pair of $\(\text{FPR}, \text{TPR}\)$ that minimizes the 
expected cost is 
\bs
s_t 
= \frac{
C_{0R} \(1 - \mathbb{P}\(Y=1\)\)
}{
C_{1A} \mathbb{P}\(Y=1\)
}, 
\end{split}\end{align}
implying that the ratio $t$ is given by equation \eqref{eq to illustrate slope of roc tangent}:
\bs
t = \frac{C_{1A}}{C_{0R}} = \frac{
1 - \mathbb{P}\(Y=1\)
}{
\mathbb{P}\(Y=1\) s_t
}.
\end{split}\end{align}

Our second alternative substitution strategy is a variant of the first one. 
We train a model 
$F_{\text{mis}}'\(x; \hat Y\)$ where the doctor's diagnosis enters directly into the feature set. The corresponding replacement rule
for each patient case becomes 
\begin{equation*}
 \begin{cases}
    F_{\text{mis}}'\(x; \hat Y\) > c_r, \quad \text{if} \ \hat Y=1,\\ 
    F_{\text{mis}}'\(x; \hat Y\) > \frac{c_r}{t}, \quad \text{otherwise}.
  \end{cases}
\end{equation*}


We experimented with the alternative substitution strategies 
on the set of 
doctors with at least 300 patient cases. 
First, we estimate $F_{\text{doc}}$, $F_{\text{mis}}$ and $F'_{\text{mis}}$ as three Random 
Forest models using 
data generating by the same sample splitting scheme 
as in the Bayesian test approach.
Next we calculate the coefficient $t$ in 
\eqref{eq to illustrate slope of roc tangent}
for every doctor. The machine ROC is drawn using the 
validation set. Both  $\mathbb{P}(Y = 1)$ and the Bayesian posterior distribution are evaluated on the combination of
the training set and the validation set. 

We set $c_d$ to be 0.5 in the experiment, and find that the results are not sensitive to 
the choice of $c_d$.
Then we calibrate the replacement threshold $c_r$ for each doctor by equating the implied 
global replacement rate $R$ to the Bayesian test approach. 
For the first alternative strategy, given the values of $t$ and $c_d$, 
the global replacement rate $R$ is related to each doctor's $c_r$ threshold by the following equation:
\begin{align}\begin{split}\nonumber
R & = 
\mathbb{P}\(F_{\text{mis}}\(x\) > c_r \vert F_{\text{doc}}\(x\) > c_d\)\mathbb{P}\(F_{\text{doc}}\(x\) > c_d\) \\ 
& + \mathbb{P}\(F_{\text{mis}}\(x\) > \frac{c_r}{t} \bigg\vert F_{\text{doc}}\(x\) \leq c_d\)\mathbb{P}\(F_{\text{doc}}\(x\) \leq c_d\),
\end{split}\end{align}
where the probabilities are estimated using the validation set. 


Similarly, for the second alternative strategy, the threshold $c_r$ is determined by equating 
\begin{align}\begin{split}\nonumber
R & = \mathbb{P}\(F'_{mis}\(x;\hat{Y}\) > c_r \big\vert \hat{Y} = 1\)\mathbb{P}\(\hat{Y} = 1\) \\ 
& + \mathbb{P}\(F'_{mis}\(x;\hat{Y}\) > \frac{c_{r}}{t} \bigg\vert \hat{Y} = 0\)\mathbb{P}\(\hat{Y} = 0\).
\end{split}\end{align}
Finally, the test set is used to compute the FPRs and the TPRs 
from replacing a subset of each doctor's patient cases by the machine algorithm. 
The Bayesian test approach using 95\% credible level replaces about 46\% of the doctors 
and $R=59.1\%$ of all patient cases.
The experiment result is shown in Figure~\ref{fg:appendixes5 fine-grained replacement}. 

\begin{figure}[hp]
  \centering
\subfigure[ROC curve and replacement paths]{\includegraphics[width=0.7\columnwidth]{"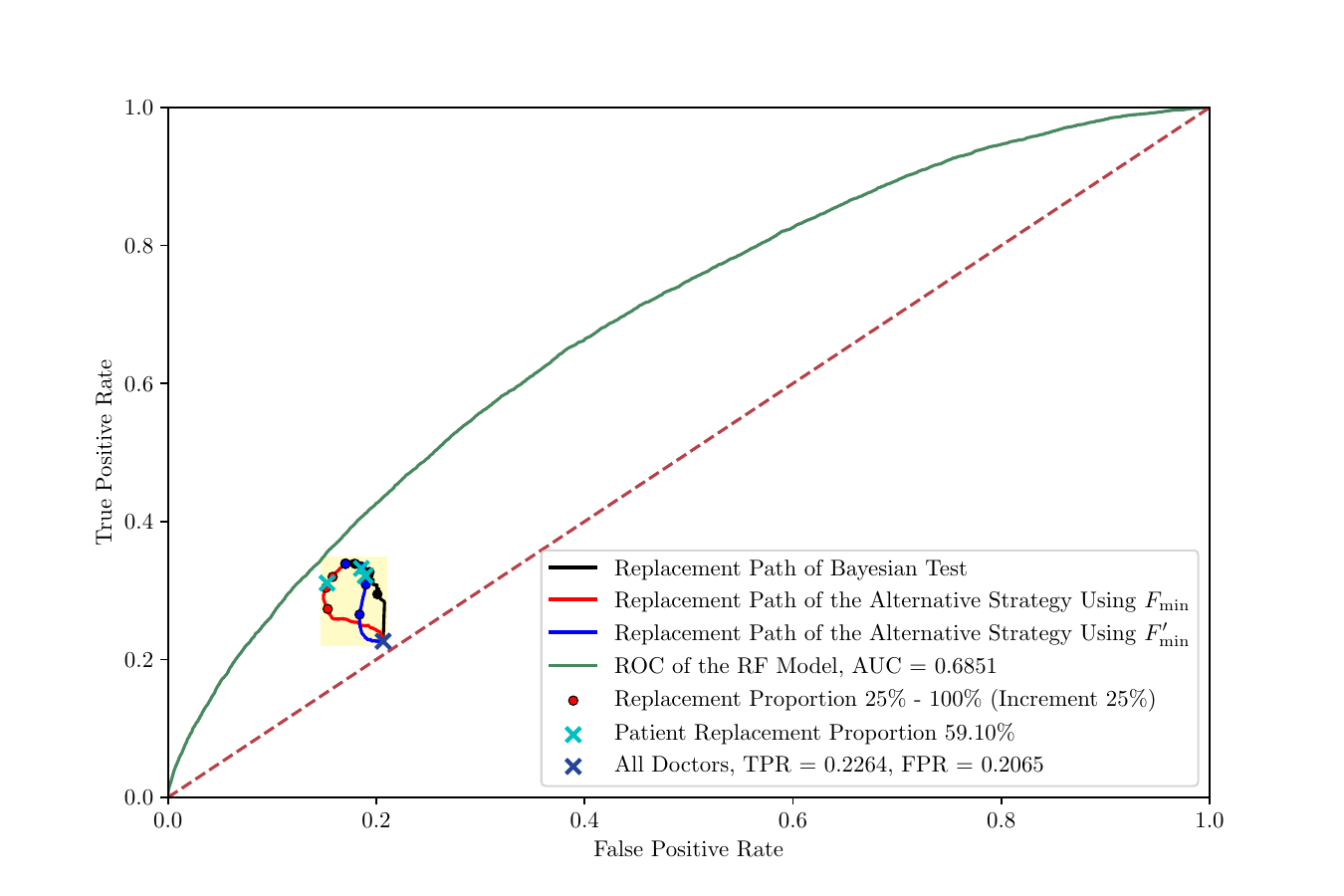"}}
\subfigure[Zoomed-in version of replacement paths]{\includegraphics[width=0.7\columnwidth]{"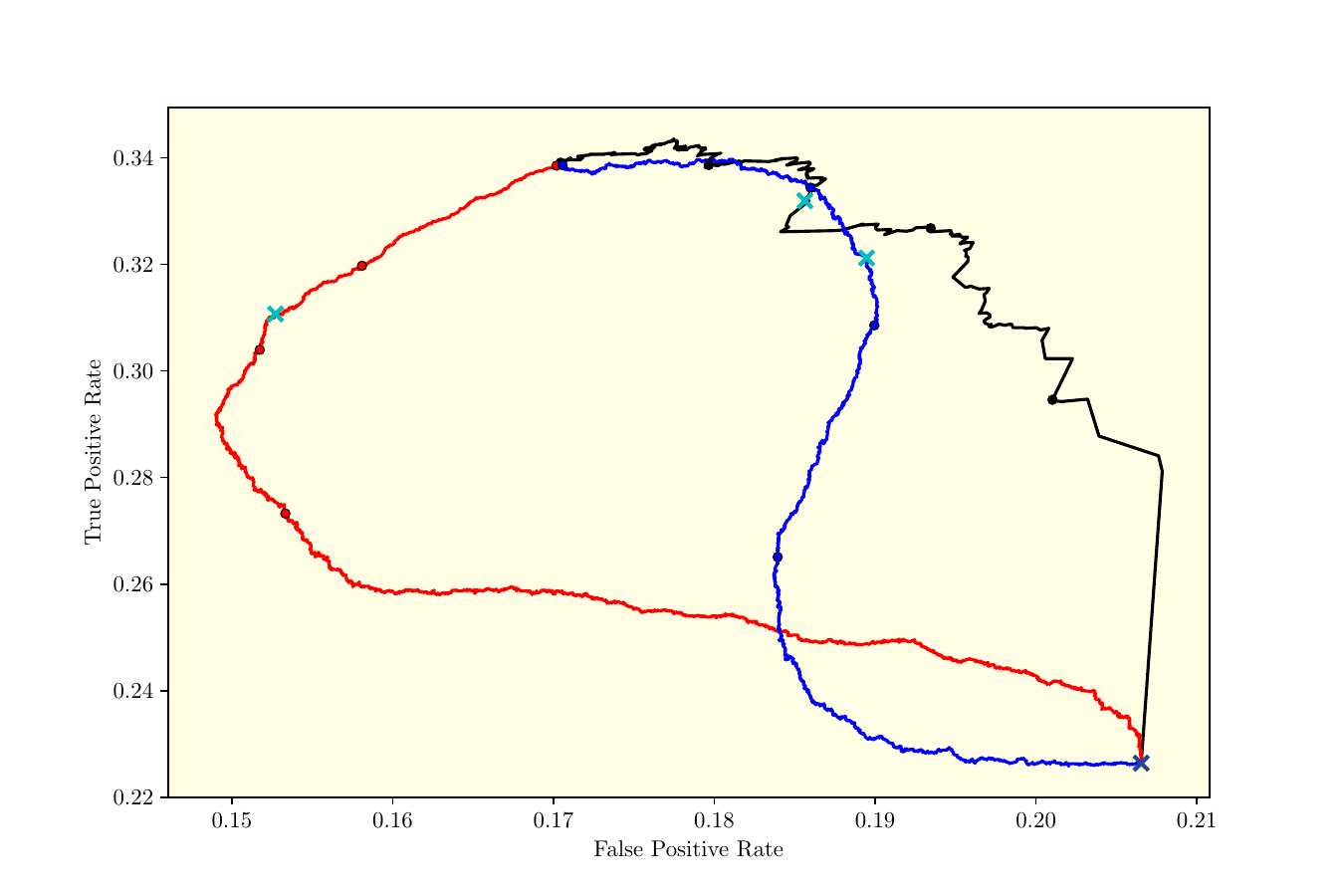"}
\label{panel b of fg:appendixes5 fine-grained replacement}
}
  \caption{The replacement paths of two alternative replacement approaches}
\fnote{\textit{Notes}: Panel \subref{panel b of fg:appendixes5 fine-grained replacement} plots the paths
of FPR/TPR when the replacement rate $R$ varies from $0$ to $100\%$.
While the replacement
paths of the two alternative strategies are very different from the replacement path of the Bayesian test,
none of the replacement paths dominate the others in FPR/TPR performance.}
  \label{fg:appendixes5 fine-grained replacement}
\end{figure}

{
\setstretch{0.875}
\bibliographystyleappendix{aer}
\bibliographyappendix{references}
}

\end{document}